\documentclass[a4paper,reqno]{amsart}

\usepackage[reqno]{amsmath}
\usepackage{amssymb,amsfonts,amsthm}
\usepackage{geometry}
\usepackage[english]{babel} %% English standard typographical conventions
\usepackage{graphicx}
\usepackage{lmodern}
\usepackage{enumitem}
\usepackage{booktabs}
\usepackage{float}
\usepackage{array}
\usepackage{makecell}

\usepackage{tabularx}
\usepackage{array}
\usepackage{booktabs}
\newcolumntype{L}{>{\raggedright\arraybackslash}X}
\newcolumntype{P}[1]{>{\raggedright\arraybackslash}p{#1}}

\usepackage{aliascnt}

\usepackage[]{color} %%more colors, monochrome=everything black
\usepackage[dvipsnames]{xcolor} %% more colors
\usepackage{csquotes} %% necessary for biber
\usepackage{array} %% for arrays
\usepackage{cancel}
\usepackage{tikz}
\usetikzlibrary{backgrounds}

\usepackage{mathtools,todonotes} %% add-on to amsmath
\usepackage{tikz-cd}  %%For commutative diagrams with tikz
\usepackage[all,cmtip]{xy} %% For diagrams
\usepackage[bbgreekl]{mathbbol} %%For bold greek letters
\usepackage{mathrsfs}
\usepackage{hyperref} %% hyperlinks
\usepackage[capitalize,noabbrev]{cleveref} %% smart cross-references

\theoremstyle{definition} %%upright text, extra space above and below
\newtheorem{definition}{Definition}
\numberwithin{definition}{section}

\theoremstyle{plain} %% italic text, extra space above and below
\newaliascnt{theorem}{definition}
\newtheorem{theorem}[theorem]{Theorem}
\aliascntresetthe{theorem}

\newaliascnt{proposition}{definition}
\newtheorem{proposition}[proposition]{Proposition}
\aliascntresetthe{proposition}

\newaliascnt{lemma}{definition}
\newtheorem{lemma}[lemma]{Lemma}
\aliascntresetthe{lemma}

\newaliascnt{corollary}{definition}

\aliascntresetthe{corollary}

\newaliascnt{claim}{definition}

\aliascntresetthe{claim}

\newaliascnt{assumption}{definition}
\newtheorem{assumption}[assumption]{Assumption}
\aliascntresetthe{assumption}

\theoremstyle{remark}
\newaliascnt{example}{definition}
\newtheorem{example}[example]{Example}
\aliascntresetthe{example}

\theoremstyle{remark} %% upright text, no extra space above or below
\newaliascnt{remark}{definition}
\newtheorem{remark}[remark]{Remark}
\aliascntresetthe{remark}

\crefname{definition}{Definition}{Definitions}
\crefname{theorem}{Theorem}{Theorems}
\crefname{proposition}{Proposition}{Propositions}
\crefname{lemma}{Lemma}{Lemmas}
\crefname{corollary}{Corollary}{Corollaries}
\crefname{claim}{Claim}{Claims}
\crefname{remark}{Remark}{Remarks}
\crefname{example}{Example}{Examples}

\usepackage[
  bibstyle=alphabetic,
  citestyle=alphabetic,
  useprefix,
  giveninits=true,
  minbibnames=99,
  maxbibnames=99,
  sorting=nyt,
  sortcites=true,
  backend=biber
]{biblatex}
\renewbibmacro{in:}{} %% removes in: before journals

\DeclareSourcemap{
  \maps[datatype=bibtex]{
    \map[overwrite]{ %%If DOI is present, doesn't print arXiv
      \step[fieldsource=doi, final]
      \step[fieldset=url, null]
      \step[fieldset=eprint, null]
    }
    \map[overwrite]{ %% If only arXiv is present, doens't print pages and eid (eliminates duplicates)
      \step[fieldsource=eprint, final]
      \step[fieldset=pages, null]
      \step[fieldset=eid, null]
      \step[fieldset=journal, null]
    }  
  }
}

\bibliography{bibliography}  %% Name of the file with the bibliography
\newcommand{\Ker}[1]{\mathrm{Ker}{(#1)}}

\newcommand{\intb}{{\int_\Sigma}}
\newcommand{\intc}{{\int_\Gamma}}
\newcommand{\Tr}[0]{\text{Tr}}

\newcommand{\de}{\mathrm{d}}

\newcommand{\og}{\Omega_\Gamma}

\newcommand{\rog}{\varrho_\Gamma} 

\DeclareMathOperator{\Ima}{Im}

\newcommand{\qsp}[2]{\,\ensuremath{\raise.5ex\hbox{$#1$}\big\slash\raise-.5ex\hbox{$#2$}}} 
\newcommand{\pard}[2]{\frac{\delta#1}{\delta#2}}

\newcommand{\en}{\epsilon_n}
\newcommand{\emm}{\epsilon_m}
\newcommand{\iz}{\iota_\zeta}

\makeatletter
\newcommand{\zzlabel}[1]{\ifmeasuring@\else\ltx@label{#1}\fi} %%new label (necessary if amsmath is present)
\makeatother

\newcounter{terms}[equation] %%counter for terms in a equation
\newcommand{\unl}[2]{\underline{#1}_{\refstepcounter{terms}\zzlabel{#2}\textcolor{blue}{\theterms}}} %%underlines, counts a term and put the corresponding number. Put the term in the first slot and a label in the second

\newcommand{\reft}[2]{(\ref{#1}.[\ref{#2}])} %%refers to a term in a equation as (equation.term)

\RequirePackage{tikz-cd}
\RequirePackage{amssymb}
\usetikzlibrary{calc}
\usetikzlibrary{decorations.pathmorphing}

\tikzset{curve/.style={settings={#1},to path={(\tikztostart)
    .. controls ($(\tikztostart)!\pv{pos}!(\tikztotarget)!\pv{height}!270:(\tikztotarget)$)
    and ($(\tikztostart)!1-\pv{pos}!(\tikztotarget)!\pv{height}!270:(\tikztotarget)$)
    .. (\tikztotarget)\tikztonodes}},
    settings/.code={\tikzset{quiver/.cd,#1}
        \def\pv##1{\pgfkeysvalueof{/tikz/quiver/##1}}},
    quiver/.cd,pos/.initial=0.35,height/.initial=0}

\tikzset{tail reversed/.code={\pgfsetarrowsstart{tikzcd to}}}
\tikzset{2tail/.code={\pgfsetarrowsstart{Implies[reversed]}}}
\tikzset{2tail reversed/.code={\pgfsetarrowsstart{Implies}}}
\tikzset{no body/.style={/tikz/dash pattern=on 0 off 1mm}}
    
\title{The reduced Dirac structure of General Relativity on manifolds with corners}

\author{Alberto S. Cattaneo, Filippo Fila-Robattino and Manuel Tecchiolli}

\renewcommand\footnotemark{}

\thanks{We acknowledge partial support of the SNF Grant No. 200021 227719 and of the Simons Collaboration on Global
Categorical Symmetries. This research was (partly) supported by the NCCR SwissMAP, funded by the Swiss
National
Science Foundation. This article is based upon work from COST Action 21109 CaLISTA, supported by COST
(European
Cooperation in Science and Technology)(www.cost.eu), MSCA-2021-SE-01-101086123 CaLIGOLA, and MSCA-DN
CaLi-
ForNIA -101119552. FFR acknowledges funding from the EU project Caligola HORIZON-MSCA-2021-SE-01, Project
ID:
101086123.}
\begin{document}
\begin{abstract}
In this paper, the corner Poisson structure of four-dimensional Palatini–Cartan gravity is derived. Building on the classical description of gravity on manifolds with boundary, specifically on the boundary constraint algebra, a pre-Dirac structure on the space of corner fields is obtained together with a reduction procedure that yields a maximal Dirac structure, identified as the graph of a Poisson bivector field, on the reduced space of corner fields. It is further shown that this Poisson structure admits an equivalent affine Poisson description, which naturally exhibits the reduced corner theory as a $BF$-like theory and leads to a BF$^2$V formulation. This provides the basis for a unified framework for the bulk, boundary, and corner structures of Palatini–Cartan gravity.
\end{abstract}
\maketitle

\tableofcontents

\section{Introduction}

The study of gauge theories on manifolds with boundary has revealed a rich interplay between variational principles, symplectic geometry, and homological methods. In particular, the BV--BFV formalism provides a coherent framework in which bulk dynamics, boundary geometry and the corresponding reduced phase space can be treated simultaneously. Starting from a classical gauge theory in the bulk, it is possible to associate a BFV theory to the boundary, encoding the symplectic structure and the constraint algebra arising from gauge invariance. Under suitable assumptions, the degree-zero BFV cohomology reproduces the algebra of observables on the reduced phase space.

There is a canonical way of relating bulk and boundary gauge theories. Starting from a BV theory on a manifold with boundary, the failure of the BV action to satisfy the master equation is encoded by boundary contributions, which naturally give rise to BFV data on the boundary. In favorable situations, this leads to a compatible BV--BFV theory and provides a homological description of the reduced phase space.

For four-dimensional Palatini--Cartan gravity, however, this program encounters substantial difficulties. As shown in \cite{CS2017b}, the singular nature of the boundary presymplectic structure obstructs a direct implementation of the BV--BFV construction. A different strategy was therefore developed in \cite{CCS2020}, where one starts from the classical boundary theory and studies the reduction of a coisotropic submanifold of boundary fields. The resulting reduced phase space is then described cohomologically by a BFV theory, providing a resolution of the corresponding quotient. More recently, it was shown in \cite{CC25} that this BFV theory can also be recovered from the bulk BV description after a suitable BV pushforward procedure. Equivalently, the same structure arises from an AKSZ-type construction associated with the BFV data. In this way, the classical and BV descriptions are ultimately reconciled, despite the failure of the naive BV--BFV construction.

Many situations of physical and mathematical interest, however, naturally involve manifolds with corners. Examples include locality and gluing constructions in classical and quantum field theory, the decomposition of spacetime into finite regions, and the study of observables associated with lower-dimensional strata. In these situations, the boundary itself acquires a non-trivial stratification, and one is naturally led to ask whether the structures appearing at codimension one admit analogous extensions to higher codimensions. This question is naturally formulated within the framework of BF$^{(k)}$V theories, where the compatibility between symplectic and homological structures is expected to persist along the entire stratification.

The corner structure of four-dimensional Palatini--Cartan gravity was investigated in \cite{CC2022cor} from the BF$^k$V perspective. Starting from the BFV theory associated with the boundary, the induced codimension-two structure was shown to carry a local $P_\infty$-algebra and to define a pre-BF$^2$V theory. The resulting corner presymplectic form remains singular, and the codimension-two theory is therefore naturally described before reduction. While this construction captures the full homotopy structure inherited from the BFV theory, it does not lead directly to a strict BF$^2$V theory.

The perspective adopted in the present work is instead inspired by the classical construction of \cite{CCS2020}. Rather than starting from the BFV extension and studying the corner theory induced by it, we begin with the classical boundary constraint algebra and investigate the geometric structure generated at codimension two. Our primary aim is to construct and characterize the reduced classical corner geometry, in direct analogy with the construction of the reduced boundary phase space in \cite{CCS2020}. By performing the relevant reduction at the classical level, we remove the degeneracies responsible for the singular corner geometry and obtain a genuine geometric structure on the reduced space of corner fields. The corresponding BF$^2$V structure is not imposed as an initial target, but emerges naturally from the resulting geometry.

More precisely, starting from the boundary constraint algebra of Palatini--Cartan gravity, we show that the residual codimension-two contributions define a pre-Dirac structure on a suitable space of corner fields. Using techniques from generalized geometry, we identify a reduction procedure that removes the obstruction to maximality and yields a genuine Dirac structure on a reduced corner phase space. We then prove that this Dirac structure is the graph of a Poisson bivector field, thereby providing a geometric characterization of the reduced corner degrees of freedom. Since the resulting Poisson bivector field naturally determines the corresponding homological data, a strict BF$^2$V structure can then be obtained directly from the reduced corner geometry.

\subsection*{Structure of the paper}
The paper is organized as follows. In Section \ref{sec:gb}, we recall the
geometric background used throughout the work. We begin with stratified
manifolds and manifolds with corners, in order to fix the distinction between
bulk, boundary and corner strata. We then review the symplectic and Poisson
language needed to describe reduced phase spaces, with particular emphasis on
coisotropic reduction, first-class constraints, and momentum maps. This is
followed by a brief recollection of the BF$^{(k)}$V formalism, which provides
the cohomological framework organizing the data attached to strata of increasing
codimension. Finally, we review the standard Courant algebroid, Dirac
structures, and the two basic classes of Dirac structures arising as graphs of
closed two-forms and Poisson bivectors. This section fixes the geometric
dictionary used later: symplectic geometry describes the codimension-one phase
space, while generalized geometry is the natural language for the
codimension-two structure.

In Section \ref{sec: corner ft ft} we explain the general mechanism by which corner structures arise in classical field theory. Starting from a Lagrangian field theory on a manifold with boundary, the variation of the action produces a presymplectic form on the space of boundary fields. After quotienting the kernel of this presymplectic form, one obtains a symplectic boundary field
space. The equations of motion then define a coisotropic submanifold of admissible boundary data; equivalently, they give first-class constraints. On a manifold with boundary, these constraints are represented by Hamiltonian vector fields and their coisotropic reduction gives the usual reduced phase space of Cauchy data. When we deal with a manifold with corners, however, integration by parts on the boundary produces residual one-forms supported on the corner. Thus a boundary constraint no longer determines only a vector field: it determines a generalized section
\begin{align*}
        X_i+\beta_i
        \in
        T\mathcal F_\Gamma\oplus T^*\mathcal F_\Gamma,
\end{align*}
where $\mathcal F_\Gamma$ is an appropriate space of corner fields (to start with, before reduction, we can take it to be equal to the space of boundary fields). The vector part records the local transformation of corner fields, while the one-form part records the variation of the corner charge (when it is exact) or the flux associated with the global part of the field transformation. We show that these pairs naturally form a generalized distribution in the standard Courant algebroid, and we formulate the two basic conditions required of it: isotropy with respect to the Courant pairing and involutivity with respect to the Dorfman bracket. This gives the abstract notion of a pre-Dirac corner structure. We then describe the reduction by the directions in the kernel of the residual one-forms, which removes the obstruction to maximality and may turn the pre-Dirac structure into a genuine Dirac structure. The general construction is illustrated in two model examples: Yang--Mills theory, where the reduction can be seen explicitly and leads to a Poisson corner structure, and four-dimensional $BF$ theory, where the classical corner structure is already Poisson but the full cohomological description also requires the higher $P_\infty$ data coming from reducibility.

In Section \ref{sec:pc_boundary} we review the boundary structure of
four-dimensional Palatini--Cartan gravity. We recall the bulk fields
\begin{align*}
        (e,\omega)
        \in
        \Omega^1_{\mathrm{nd}}(M,V)\oplus\mathcal A_M
\end{align*}
and the Palatini--Cartan action\footnote{In dimension $4$. Moreover, notice that $e^p$ is the $p$-times wedge product of $e$.}
    \begin{align}
    S_M=\int_{M}\frac1{2}e^{2} F_\omega+\frac1{4!}\Lambda e^4
    \end{align}
Its variation gives the boundary Noether
one-form
\begin{align*}
        \alpha_\Sigma
        =
        \int_\Sigma \frac12 e^2\,\delta\omega
\end{align*}
and hence a degenerate presymplectic form on pre-boundary fields. We recall the
first reduction which produces the geometric boundary phase space
$(\mathcal F_\Sigma,\omega_\Sigma)$. We then describe the constraint surface
of admissible Cauchy data in terms of the three families of boundary
constraints
\begin{align*}
        L_c,\qquad P_\xi,\qquad H_\lambda ,
\end{align*}
corresponding respectively to internal Lorentz transformations, diffeomorphisms
tangent to the boundary, and normal deformations. These constraints are
first-class and define a coisotropic submanifold of the boundary phase space.
We also recall the BFV description of this boundary theory, stressing the role
of the constraint algebra and the obstruction that appears when one attempts to
descend directly to a strict codimension-two BF$^2$V theory.

Section \ref{sec: corner PC} contains the main construction of the paper. Restricting the boundary constraints to the corner
produces four residual one-forms,
\begin{align*}
        \mathcal J_c,\qquad
        \mathcal E_\zeta,\qquad
        \mathcal K_\eta,\qquad
        \mathcal F_\lambda ,
\end{align*}
where the boundary vector field decomposes at the corner as
\begin{align*}
        \xi|_\Gamma=\zeta+\eta\partial_m .
\end{align*}
These one-forms are paired with the push-forwards of the Hamiltonian vector
fields of the corresponding boundary constraints. In this way we obtain a
generalized distribution
\begin{align*}
        \tilde D
        =
        \operatorname{span}
        \bigl\{
        \mathbb J_c+\mathcal J_c,\,
        \mathbb E_\zeta+\mathcal E_\zeta,\,
        \mathbb K_\eta+\mathcal K_\eta,\,
        \mathbb F_\lambda+\mathcal F_\lambda
        \bigr\}
        \subset
        T\mathcal F_\Gamma\oplus T^*\mathcal F_\Gamma .
\end{align*}
We prove that this distribution is isotropic and involutive; hence it is a
pre-Dirac structure on the space of corner fields. We then identify the kernel
of the one-form components, namely the directions of the pre-corner fields that
are invisible to all corner charge variations. Quotienting by this distribution
removes the obstruction to maximality. The reduced corner fields are described
in terms of the natural variables
\begin{align*}
        E=\frac12 e^2,
        \qquad
        \Omega=e(\omega-\omega_0),
\end{align*}
where $\omega_0$ is an arbitrarily chosen reference connection, together with transverse data fixed by the torsion, Einstein and Bianchi relations. On the resulting reduced corner space $P_\Gamma$, the induced Dirac structure is shown to be maximal and, more precisely, to be the graph of
a Poisson bivector field:
\begin{align*}
        D
        =
        \operatorname{Graph}(\pi).
\end{align*}
This provides the reduced Poisson geometry of the Palatini--Cartan corner
structure.
% In Section \ref{sec: corner_charge_algebra}, we extract the algebraic
% content of the reduced Poisson corner structure in terms of corner charges.
% The one-form components of the Dirac structure are interpreted as infinitesimal
% charge variations. When such a one-form is exact,
% \begin{align*}
%         \beta_i=\delta Q_i ,
% \end{align*}
% the corresponding transformation admits an integrable corner charge
% $Q_i$; when it is not exact, it represents a non-integrable charge
% variation, or flux. On the reduced Poisson corner space, integrable charges
% carry the Poisson bracket
% \begin{align*}
%         \{Q_i,Q_j\}_\Gamma
%         =
%         \pi(\delta Q_i,\delta Q_j).
% \end{align*}
% We use this bracket to describe the corner charge algebra associated with what we interpret as the
% surviving (at the corner) global components of the gauge symmetries. In particular, we explain
% how equivariance of the corresponding moment map is expressed as closure of the
% charges under the Poisson bracket. This final section makes explicit the physical meaning of the reduced Dirac structure: it does not merely encode local transformations of corner fields, but also the global nature of such transformations, i.e. the visible (at the corner) effects of a gauge transformation.

Finally, in Section \ref{sec: BF2V}, we pass from the reduced Poisson geometry to the
strict BF$^2$V description. Since a Poisson bivector field can be encoded as a
degree-two function satisfying the classical master equation on the shifted
cotangent bundle of the reduced corner space, the Poisson structure on
$P_\Gamma$ naturally determines a BF$^2$V theory. We first write the BF$^2$V
data associated with the Poisson bivector obtained in
Section \ref{sec: corner PC}. We then perform a field redefinition that
exhibits the same structure in an affine Poisson form. In this form the
BF$^2$V action becomes
\begin{align*}
        S^{PC}_\Gamma
        =
        \int_\Gamma
        \frac12 E[c,c]
        +
        \mu\,e\,d_{\omega_0}c
        +
        \frac12\mu^2F_{\omega_0},
\end{align*}
where the ghost $\mu$ packages the diffeomorphism and normal-deformation
parameters through
\begin{align*}
        \mu
        =
        \iota_\zeta e+\lambda\epsilon_n+\eta\epsilon_m .
\end{align*}
This affine formulation makes the relation with four-dimensional $BF$ theory
transparent and identifies the affine term as the classical source of the
central-extension data relevant for the associated corner algebra.

\subsection*{Acknowledgments}
The authors would like to thank Francesco Bonechi and Giovanni Canepa for the illuminating remarks. 
\section{Geometrical background}\label{sec:gb}

\subsection{Stratified spaces}\label{strat_spac}
\begin{definition}[Stratified manifold\footnote{See \cite{Pflaum2001} for a detailed treatment.}]
A \emph{stratified manifold} is a topological space $X$ equipped with a decomposition
\begin{align*}
X &= \bigsqcup_{\alpha \in A} S_\alpha,
\end{align*}
where each $S_\alpha$ is a (locally closed) smooth manifold of dimension $d_\alpha$, called a \emph{stratum}, and the following conditions hold:
\begin{enumerate}
  \item Each $S_\alpha$ is a smooth manifold.
  \item For any two strata $S_\alpha, S_\beta$, if
  \begin{align*}
    S_\beta \,\cap\, \overline{S_\alpha} &\neq \varnothing,
  \end{align*}
  then $S_\beta \subset \overline{S_\alpha}$ and $\dim S_\beta < \dim S_\alpha$. This is sometimes called the \emph{frontier condition}.
  \item The topology of $X$ is such that each closure $\overline{S_\alpha}$ is a union of strata.
\end{enumerate}
\end{definition}
\begin{remark}
The stratification induces a partial order on the index set $A$ defined by
\begin{align*}
\beta \preceq \alpha \quad\Longleftrightarrow\quad S_\beta \subset \overline{S_\alpha},
\end{align*}
which, by the frontier condition, is well-defined and acyclic.
\end{remark}
\begin{example}[The Circular Cone]
Consider the quadratic cone
\begin{align*}
X &= \{(x,y,z)\in\mathbb{R}^3 : z^2 = x^2 + y^2\}.
\end{align*}
It admits a two-stratum decomposition:
\begin{align*}
S_2 &= X \setminus \{(0,0,0)\},\\
S_0 &= \{(0,0,0)\}.
\end{align*}
Here the apex is a zero-dimensional stratum, and the smooth part is two-dimensional.
\end{example}

\begin{example}[Algebraic Varieties and Whitney Stratification]
Let $V \subset \mathbb{C}^n$ be an affine algebraic variety defined by the vanishing of polynomials $f_1,\dots,f_m$. For each $k=0,\dots,n$, set
\begin{align*}
V_k &= \bigl\{\,p\in V : \operatorname{rank}\bigl(Df(p)\bigr)=n-k\bigr\},
\end{align*}
where $Df(p)$ is the Jacobian matrix of $(f_1,\dots,f_m)$ at $p$. Equivalently, $V_k$ is the locus where $V$ is a smooth complex submanifold of complex dimension $k$ (real dimension $2k$).

Then
\begin{align*}
V &= \bigsqcup_{k=0}^n V_k
\end{align*}
defines a stratification of $V$ satisfying:
\begin{enumerate}
  \item \emph{Frontier condition:} if $V_j\cap\overline{V_k}\neq\emptyset$, then $j<k$ and $V_j\subset\overline{V_k}$.
  \item \emph{Whitney regularity:} for each $j<k$, the pair $(V_j,V_k)$ satisfies Whitney’s Conditions (A) and (B), ensuring that near any $p\in V_j$ the tangent spaces to $V_k$ vary continuously and intersect $T_pV_j$ in the expected dimension.
\end{enumerate}
\end{example}
\begin{example}[Square in $\mathbb{R}^2$]
Let 
\begin{align*}
X &= [0,1]^2 \;\subset\;\mathbb{R}^2.
\end{align*}
We stratify $X$ by
\begin{align*}
S_0 &= (0,1)^2,\\
S_1 &= \bigl(\{0\}\times(0,1)\bigr)\,\cup\,\bigl(\{1\}\times(0,1)\bigr)\,\cup\,\bigl((0,1)\times\{0\}\bigr)\,\cup\,\bigl((0,1)\times\{1\}\bigr),\\
S_2 &= \{(0,0),(0,1),(1,0),(1,1)\}.
\end{align*}
Here $S_2$ consists of the four \emph{corners} of the square—points where two boundary edges meet.  The interior $(0,1)^2$ is $S_0$, the open edges are $S_1$, and the corner points are $S_2$.
\end{example}
\begin{definition}[Manifold with corners]\label{def:strat_man}
A \emph{manifold with corners} is a smooth manifold $M$ of dimension $N$ equipped with a stratification
\begin{align*}
  M = \bigsqcup_{k = 0}^{K} M^{[k]},
\end{align*}
for some $K \leq N$, satisfying the following:
\begin{enumerate}
  \item Each $M^{[k]}$ is the codimension-$k$ stratum, consisting of points locally diffeomorphic to an open subset of
  \begin{align*}
    \{x_1 = x_2 = \dots = x_k = 0\} \subset \mathbb{R}^k \times \mathbb{R}^{N-k},
  \end{align*}
  i.e., modeled on open neighborhoods of $[0, \infty)^k \times \mathbb{R}^{N - k}$ with exactly $k$ vanishing coordinates.

  \item Each stratum $M^{[k]}$ is a smooth manifold of dimension $N - k$ (without boundary).

  \item The stratification satisfies the frontier condition:
  \begin{align*}
    \overline{M^{[k]}} \supset \bigsqcup_{\ell > k} M^{[\ell]}.
  \end{align*}

  \item The manifold $M$ admits a smooth structure modeled on open subsets of $[0, \infty)^k \times \mathbb{R}^{N-k}$ for $k = 0, \dots, K$, with smooth transition maps between overlapping charts.
\end{enumerate}
\end{definition}

\begin{remark}\label{rem:two_strata}
In the following sections, we will restrict our attention to manifolds with corners admitting only strata of codimension at most $2$, that is:
\begin{align*}
  M = M^{[0]} \sqcup M^{[1]} \sqcup M^{[2]}.
\end{align*}
\end{remark}
\begin{remark}
    Notice that, even though a manifold with boundary is a special case of a manifold with corners, consisting of only two strata, one of codimension-$0$ (the bulk) and one of codimension-$1$ (the boundary), the definition of manifold with corners usually refers to stratified structures with higher depth.
\end{remark}
\subsection{Graded symplectic and Poisson geometry}\label{grad_symp_poiss}
In this section, we briefly recall some notions of classical symplectic and Poisson geometry and describe their interplay with graded geometry.

\begin{definition}[Symplectic manifold]
    A symplectic manifold is a pair $(F,\varpi)$, where $\varpi$ is a closed two-form on the smooth manifold $F$ (possibly infinite dimensional) such that the natural map $\mathfrak{X}(F)\to \Omega^1(F) $ sending $X\mapsto \iota_X\varpi$ is an isomorphism (or just injective if $F$ is infinite dimensional). If the non-degeneracy condition is dropped, $\varpi$ is called pre--symplectic.\footnote{Strictly speaking, one also requires the rank of the characteristic distribution of $\omega$ to be constant.} 
\end{definition}
\begin{definition}[Poisson manifold ]
A Poisson structure on $F$ is an $\mathbb{R}$-bilinear antisymmetric bracket $\{\cdot,\cdot\}$ on $C^\infty(F)$ satisfying the Leibniz rule
\[
\{f,gh\} = \{f,g\}h + g\{f,h\}, \qquad \forall f,g,h \in C^\infty(F),
\]
and the Jacobi identity
\[
\{f,\{g,h\}\} + \{g,\{h,f\}\} + \{h,\{f,g\}\} = 0.
\]
These requirements are equivalent to the existence of a Poisson bivector field $\pi \in \Gamma(\wedge^2 TF)$ such that
\[
\{f,g\} = \pi(\de f,\de g), \quad \text{and}\quad [\pi,\pi]=0,
\]
where $[\cdot,\cdot]$ is the Schouten bracket.

A manifold endowed with a Poisson structure $(F,\pi)$ is called a Poisson manifold, while $\mathcal{C^1}(F)$ is called a Poisson algebra. 
\end{definition}

In general, a symplectic structure induces a Poisson algebra. Indeed, letting $(F,\varpi)$ be symplectic, if one considers any functions that have a Hamiltonian vector fields---we call them Hamiltonian functions---, i.e. any $f$ such that there exists $X_f\in\mathfrak{X}(F)$ for which $\iota_{X_f}\varpi=\de f$,\footnote{Notice that, because of the non-degeneracy of $\varpi$, if $X_f$ exists it is unique. Furthermore, when $F$ is finite dimensional, every function is Hamiltonian.} one can define the Poisson bracket
    \[
    \{f,g\}\coloneqq \iota_{X_f}\iota_{X_g}\varpi.
    \]

\begin{definition}\label{def_co_is_la}
    Let $(F, \varpi)$ be a symplectic manifold. Then, a submanifold $C\subseteq F$ is called
    \begin{itemize}
  \item[--] \emph{Isotropic} if $T_xC \subseteq T_xC^\perp$,
  \item[--] \emph{Coisotropic} if $T_xC^\perp \subseteq T_xC$,\footnote{In the infinite-dimensional case, there si also the strong notion of split lagrangian submanifold: in this
case, $T_xC$ is required to be isotopic and to have an isotropic complement for all x.}
  \item[--] \emph{Lagrangian} if $T_xC = T_xC^\perp$,
\end{itemize}
for all $x$ in $C$, where
    \[
    T_xC^\perp\coloneqq\{v\in T_xF \ | \ \varpi_x(v,w)=0 \ \forall w\in T_xC\}.
    \]
\end{definition}

\begin{remark}\label{rem: coisotropic infdim}
    We are mostly interested in the case of a coisotropic submanifold, as that will be the relevant object describing the phase space of the field theories we will take into consideration. Furthermore, when $F$ is infinite-dimensional, we have to resort to a more algebraic definition, which in the finite-dimensional case is equivalent to the previous one. In particular, letting $I$ be the vanishing ideal of $C$, i.e.
    \[
    I\coloneqq \{f\in\mathcal{C^1}(F) \ \vert \ f|_C=0\},
    \]
we require\footnote{In Dirac’s terminology, $C$ is defined by first-class constraints.}
    \[
    \{I,I\}\subseteq I.
    \]
\end{remark}

In general, the coisotropic submanifold $C$ is not symplectic, as the restriction of the symplectic form $\varpi|_C$ is degenerate. In fact, the kernel of $\omega$ restricted to $C$ is the orthogonal bundle $T^\perp C = \cup_{x\in
C} (T_xC)^\perp \subseteq TC$. Moreover, sections on $T^\perp C$ are involution, so we have a regular and
involutive distribution (which, in the finite-dimensional case is automatically integrable by the Frobenius theorem). The goal is then to study the coisotropic reduction $\underline{C}\coloneqq C/T^{\perp} C$. However, in many cases, the quotient $\underline{C}$ is not smooth, hence the geometrical description of the quotient becomes particularly difficult to compute. One solution is to only describe it algebraically, by finding the algebra of functions on $\underline{C}$. An even better solution is to find a cohomological resolution of the quotient in the symplectic setting.

We start by assuming that the vanishing ideal is generated by differentially independent Hamiltonian constraints $\phi_i\in\mathcal{C^1}(F)$, which is to say that $C=\{\phi_i=0\}$. Then, their Hamiltonian vector fields $X_i\in\mathfrak{X}(F)$ are by construction in $\Ker{\varpi|_C}$. Indeed, one immediately sees
    \[
    \iota_{X_i} \varpi|_C=\delta\phi_i|_C=0.
    \]
Notice also that $\mathcal{C^1}(C)=\mathcal{C^1}(F)/I$, as two functions on $F$ whose difference vanishes on $C$ will be the same on $C$. Furthermore, we can also consider the normalizer
    \[
    N(I)\coloneqq \{f\in \mathcal{C^1}(F) \ \vert \ \{I,f\}\subseteq I\}
    \]
and see that $N(I)/I$ is exactly the subalgebra of functions of $\mathcal{C^1}(F)/I$ that are invariant with respect to the action of $I$. In other words,
    \[
    N(I)/I=\{f\in\mathcal{C^1}(C) \ | \ \{I,f\}=0 \}=\{f\in\mathcal{C^1}(C) \ | \ \{\phi_i,f\}=X_i(f)=0 \ \text{for all }i\}\simeq \mathcal{C^1}(C)^I\simeq \mathcal{C^1}(\underline{C}),
    \]
as it is clear that the functions invariant under the actions of the Hamiltonian vector fields $X_i$ are constant along the leaves of the characteristic foliation of $\varpi|_C$, hence defining the functions on $\underline{C}$.

\begin{remark}\label{rem: MW red}
    In the particular case where the constraints $\phi_i$ are the components $\mu_i=\phi_i$ of a moment map $\mu\colon F\to \mathfrak{g}^*$, where $\mathfrak{g}$ is the Lie algebra of a compact Lie group $G$, and such that $0\in \mathfrak{g}^*$ is a regular value, it follows by the Marsden--Weinstein theorem \cite{MW1974} that the coisotropic reduction $\underline{C}=\mu^{-1}(0)/G $ is smooth and symplectic, with symplectic form simply given by $\pi_*(\varpi_C)$, where $\pi\colon C\rightarrow\underline{C}$.
\end{remark}

As mentioned before, another solution is to obtain a cohomological resolution of the space of functions on $\underline{C}$, by means of the so-called BFV formalism. Starting from the simpler example of \cref{rem: MW red}, we notice that the components of the momentum map satisfy the equivariance condition, i.e. \[ \{\mu_i,\mu_j\}=f^k_{ij}\mu_k, \qquad \mathrm{for} \ i=1,\dots,\dim G, \] where $f^k_{ij}$ are the structure constants of $\mathfrak{g}$. The MW theorem can then be reinterpreted with the language of supergeometry. Specifically, one considers the supermanifold\footnote{See the next section for the relevant definitions.} $F\times T^*\mathfrak{g}[1]$, where $\mathfrak{g}[1]$ is just the Lie algebra $\mathfrak{g}$ whose coordinates are now odd (shifted by 1). It is easy to see that $\mathcal{C^1}(T^*\mathfrak{g}[1])\simeq \wedge^\bullet \mathfrak{g}^*\otimes \wedge^\bullet \mathfrak{g}$, and therefore \[ \mathcal{C^1}(F\times T^*\mathfrak{g}[1])\simeq \mathcal{C^1}(F)\otimes \wedge^\bullet \mathfrak{g}^*\otimes \wedge^\bullet \mathfrak{g}. \] On $F\times T^*\mathfrak{g}[1]$ one can then define a differential $Q$ which is the twisting of the Koszul differential and the Chevalley--Eilenberg differential, such that, if we define $(c^i,c^\dag_i)$ as coordinates of $T^*\mathfrak{g}[1]$, we have in components \[ S=\mu_ic^i -\frac{1}{2}f^k_{ij}c^i c^j c^\dag_k, \qquad \mathrm{and} \qquad Q=\{S,\cdot\}, \] where $S$ is commonly refered to as BRST charge \cite{BRS}. One can then show that $H^0_Q\simeq \mathcal{C^1}(\underline{C})$. The proof relies on the fact that $\Ker{Q_0}$ roughly selects the $\mathfrak{g}$-equivariant functions on $M$, while $\Ima(Q_{-1})$ defines the vanishing ideal.\footnote{In reality, the proof requires a few more technical steps. For a complete review of the BRST formalism, we refer to \cite{Figueroa}.} As it turns out, this construction can be generalized to the case where one has constraints $\{\phi_i\}_{i=1,\dots,N}$ which are now just components of a map $\phi\colon F \rightarrow W$, where $W$ is just an $N$-dimensional vector space. Furthermore, we have the coisotropic condition
\[ \{\phi_i,\phi_j\}=f_{ij}^k \phi_k, \qquad \text{where}\qquad f^k_{ij}\in\mathcal{C^1}(F). \]
The BFV \cite{BV3, Stasheff1997} prescription is essentially the same as the BRST one, i.e. one constructs a function $S\in\mathcal{C^1}(F\times T^*W[1])$ such that $Q=\{S,\cdot\}$ is a cohomological vector field, and such that its degree-0 cohomology is isomorphic to the algebra of functions (as a Poisson algebra) $\mathcal{C^1}(\underline{C})\simeq H^0_Q(F\times T^*\mathfrak{g}[1])$.\footnote{When $\underline{C}$ is not smooth, we take $H^0_Q$ as the definition of $\mathcal{C^1}(\underline{C})$.} In particular, one has \[ S=\phi_i c^i -\frac{1}{2}f^k_{ij}c^ic^jc^\dag_k + R, \] where $R$ is determined recursively and contains terms of higher degree in the ``antighosts'' $c^\dag$.

\subsection{The BF$^k$V formalism}\label{subsec: BF^kV}

\begin{definition}[Supermanifold]\label{sup_man}
Let $M$ be a smooth $p$-manifold. A supermanifold of dimension $p\mid q$
with body $M$ is a locally ringed space $\mathcal F = (M,\mathcal O_{\mathcal F})$
such that:
\begin{enumerate}
  \item $\mathcal O_{\mathcal F}$ is a sheaf of supercommutative $\mathbb{R}$-algebras on $M$;
  \item for every point $x \in M$ there exists an open neighborhood
  $U \subseteq M$ with $x \in U$ and an isomorphism of locally ringed spaces
  \begin{align*}
    (U,\mathcal O_{\mathcal F}|_U)
    \cong
    (U,\mathcal C^\infty_U \otimes_\mathbb{R} \wedge^\bullet \mathbb{R}^q),
  \end{align*}
  where $\mathcal C^\infty_U$ denotes the sheaf of smooth real--valued functions on $U$.
\end{enumerate}
\end{definition}

\begin{definition}[dgs supermanifold]\label{dgs_man}
A differentially graded symplectic supermanifold (dgs supermanifold) of degree $k$
is a triple $(\mathcal F,\mathcal Q,\varpi)$ where:
\begin{enumerate}
\item $\mathcal F$ is a $\mathbb Z$--graded supermanifold;\footnote{We will usually assume that the $\mathbb Z$--grading is compatible with the $\mathbb Z_2$--parity, i.e. the parity of a homogeneous element equals its degree modulo $2$. This grading is also called the “ghost number”, when dealing with superfields in the BF$^k$V formalism.}
\item $\mathcal Q\in\mathfrak{X}(\mathcal F)_1$ is a vector field of degree $+1$ satisfying
\begin{align*}
[\mathcal Q,\mathcal Q]=0,
\end{align*}
that is, $\mathcal Q$ is a cohomological vector field;
\item $\varpi\in\Omega^2(\mathcal F)_k$ is symplectic form of degree $k$ on $\mathcal F$;
\item $\mathcal Q$ preserves $\varpi$, i.e.
\begin{align*}
\mathcal L_{\mathcal Q}\varpi = 0.
\end{align*}
\end{enumerate}
\end{definition}

\begin{definition}[BF$^k$V manifold]\label{bfkv_man}
A BF$^k$V manifold is a $(k-1)$--Hamiltonian manifold, that is a quadruple
$(\mathcal F,\mathcal Q,\varpi,\mathcal S)$ where:
\begin{enumerate}
  \item $(\mathcal F,\mathcal Q,\varpi)$ is a dgs supermanifold of degree $k-1$;
  \item $\mathcal S\in C^\infty(\mathcal F)_k$ is a function of degree $k$ whose
  Hamiltonian vector field with respect to $\varpi$ coincides with $\mathcal Q$, i.e.
  \begin{align*}
    \iota_{\mathcal Q}\varpi = \delta\mathcal S,
  \end{align*}
  or equivalently
  \begin{align*}
    \mathcal Q(f) = \{\mathcal S,f\}
  \end{align*}
  for all $f\in C^\infty(\mathcal F)$, where $\{\cdot,\cdot\}$ denotes the
  Poisson bracket of degree $1-k$ induced by $\varpi$.
\end{enumerate}
In this situation, $\mathcal S$ automatically satisfies the classical master equation
\begin{align*}
  \{\mathcal S,\mathcal S\} = 0,
\end{align*}
which is equivalent to $[\mathcal Q,\mathcal Q]=0$.
\end{definition}

\begin{remark}
Given \cref{bfkv_man}, we can recall the two most common cases: 
$-1$--Hamiltonian manifolds, called BV manifolds, and $0$--Hamiltonian manifolds, 
also known as BFV manifolds.
\end{remark}
In applications to field theories on manifolds with
or corners, it is often too restrictive to require the condition
\begin{align*}
  \iota_{\mathcal Q}\varpi = \delta\mathcal S
\end{align*}
as in \cref{bfkv_man}. After integrating by parts, the action functional
$\mathcal S$ may fail to admit a Hamiltonian vector field, as it is in presence of a stratified structure. Moreover, in the most constructions, the form $\varpi$ is no often found to be degenerate. It is therefore convenient to relax these conditions, leading to the following definition.

\begin{definition}[Relaxed BF$^k$V manifold]\label{rel_bfkv_man}
A relaxed BF$^k$V manifold is a quadruple $(\mathcal F,\mathcal Q,\varpi,\mathcal S)$
such that:
\begin{enumerate}\setlength{\itemsep}{2pt}
  \item $\mathcal F$ is a $\mathbb Z$-graded supermanifold;
  \item $\varpi\in\Omega^2(\mathcal F)_{k-1}$ is a closed $2$-form of degree $k-1$
  on $\mathcal F$ (not necessarily non--degenerate);
  \item $\mathcal S\in C^\infty(\mathcal F)_k$ is a function of degree $k$;
  \item $\mathcal Q\in\mathfrak X(\mathcal F)_1$ is a vector field of degree $1$
  satisfying
  \begin{align*}
    [\mathcal Q,\mathcal Q]= 2 \mathcal{Q}^2= 0.
  \end{align*}
\end{enumerate}
In particular, we do not require that $\varpi$ be symplectic nor that
$\mathcal Q$ be the Hamiltonian vector field of $\mathcal S$; when these
additional conditions hold, one recovers a BF$^k$V manifold in the sense of
\cref{bfkv_man}.
\end{definition}

\begin{remark}
We introduce the
“check” $1$-form of degree $k$
\begin{align*}
  \check\alpha \coloneqq  \iota_{\mathcal Q}\varpi - \delta\mathcal S,
\end{align*}
which measures the failure of the Hamiltonian condition
$\iota_{\mathcal Q}\varpi = \delta\mathcal S$ for a BF$^k$V manifold. We then
denote by $\check\varpi$ its differential,
\begin{align*}
  \check\varpi \coloneqq  \delta\check\alpha = - L_{\mathcal Q}\varpi.
\end{align*} 
Notice that, since $\mathcal Q$ is odd, Cartan's formula takes the form
\begin{align*}
  L_{\mathcal Q} = \iota_{\mathcal Q}\delta - \delta \iota_{\mathcal Q}.
\end{align*}
The equation $[\mathcal Q,\mathcal Q] = 0$ then immediately implies
\begin{align*}
  L_{\mathcal Q}\check\varpi = 0.
\end{align*}
The $2$--form
\begin{align*}
  \check\varpi = \delta\check\alpha
\end{align*}
will play the role of the pre--symplectic
(in general degenerate) form induced by higher--codimension corner terms.
\end{remark}
\begin{remark}
    Notice that one has that $\mathcal{Q}$ is now Hamiltonian with respect to $\check{\varpi}$, and that the Hamiltonian $\mathcal{S}$ can be obtained by computing $\iota_{\mathcal{Q}}\iota_{\mathcal{Q}}\varpi$. We refer to \cite{CMW} for more details on this construction. 
\end{remark}
The idea of this construction is to relate the (relaxed) BF$^k$V structures associated with different corner components by transporting the corresponding data from one stratum to another. In the presence of a stratified structure 
as in \cref{def:strat_man}, our construction is compatible with the codimension 
of the strata: the integer $k$ in BF$^k$V corresponds to the codimension of the 
stratum. This reflects the idea of having a relaxed BF$^k$V manifold over each corner component of codimension $k$. In this context, BV is the theory built on the bulk (codimension $0$) and BFV is the one living on the boundary (codimension $1$) presented in \cref{grad_symp_poiss}.

\subsection{Generalized geometry}\label{subsec: gen geom}

Firstly introduced in \cite{Hitchin2003} and later refined in \cite{Gualtieri2004} in the complex case, generalized geometry stems from the study of the differential geometry of $TM\oplus T^*M$, for a smooth manifold $M$. Replacing $T$ with $T\oplus T^*$ allows to very neatly unify complex, symplectic and Poisson geometry within a single structure, paving the way for more general cases. In this paper, we will follow a more algebraic perspective, rather than focusing on the geometric structures of the bundle $TM\oplus T^*M$ (e.g. generalized complex, Kähler and metric structures). Indeed, $TM\oplus T^*M$ can be endowed with the algebraic structure of a Courant algebroid (the standard Courant algebroid), whose definition and properties are displayed in the following section.

\begin{definition}[Courant algebroid\footnote{See \cite{meinrenken2025}, \cite{Roytenberg2002} and references therein for more details.}]\label{def_cour_alg}
A \emph{Courant algebroid} over a smooth manifold $M$ is a vector bundle $E \to M$ equipped with a non-degenerate symmetric bilinear form $\langle \cdot, \cdot \rangle$, a bilinear bracket $[\cdot,\cdot] \colon  \Gamma(E) \times \Gamma(E) \to \Gamma(E)$, and an anchor bundle map $\rho \colon E \to TM$, satisfying the following conditions for all $e_1,e_2,e_3 \in \Gamma(E)$ and $f \in C^\infty(M)$:
\begin{enumerate}[label=(\roman*)]
  \item $[e_1, [e_2, e_3]] = [[e_1, e_2], e_3] + [e_2, [e_1, e_3]]$,
  \item $\rho([e_1,e_2]) = \{\rho(e_1),\rho(e_2)\}$,
  \item $[e_1, f e_2] = f [e_1,e_2] + (\rho(e_1)f) e_2$,
  \item $\langle e_1, [e_2, e_3] +[e_3, e_2]\rangle=\rho(e_1)\langle e_2, e_3 \rangle$,
  \item $\rho(e_1)\langle e_2, e_3 \rangle = \langle [e_1, e_2], e_3 \rangle + \langle e_2, [e_1, e_3] \rangle$,
\end{enumerate}
where $\{\cdot,\cdot\}$ is the Lie bracket of vector fields.
\end{definition}
\begin{remark}\label{rem:cour_dorf}
    From the properties above, it follows that $[e_1,e_2]+[e_2,e_1]=\mathcal{D}\langle e_1, e_2 \rangle$, where $\mathcal{D}\colon C^\infty(M) \to \Gamma(E)$ is a differential operator defined as $\mathcal D= \rho^*\mathrm{d}$, with $\rho^*\colon T^*M\to E$ the co-anchor map\footnote{a.k.a. transposed anchor map.} of $\rho$ defined by $\langle \rho^*(\alpha), e \rangle=\alpha(\rho(e))$, with $\alpha\in\Gamma(T^*M)$ and $e\in\Gamma(E)$, and $\mathrm{d}$ the de Rham differential operator. In other words, the bracket $[\cdot,\cdot]$ is not antisymmetric in general, and its lack of antisymmetry is measured by $\mathcal{D}$ and the pairing. Notice that it would be possible to modify the axioms of \cref{def_cour_alg} in order to account for an antisymmetric bracket. In this case, it would no longer satisfy the Jacobi-type identity of property (i), and the differential operator $\mathcal D$ would then measure its failure. The non--antisymmetric bracket is called the Dorfman bracket, whereas its antisymmetric counterpart is known as the Courant bracket.
\end{remark}
\begin{remark}
Notice the difference with a Lie algebroid. A Lie algebroid is a vector bundle $A \to M$ equipped with a Lie bracket $[\cdot,\cdot]$ on its space of sections and an anchor map $\rho\colon  A \to TM$ satisfying the Leibniz rule
\begin{equation}
[a, f b] = f [a,b] + (\rho(a)f)\, b,
\end{equation}
for all $a,b \in \Gamma(A)$ and $f \in C^\infty(M)$. In this setting, the bracket is skew-symmetric and satisfies both the Leibniz and the Jacobi identities.
\end{remark}

\begin{definition}
    Let $M$ be a smooth manifold. The standard Courant algebroid over $M$ is the vector bundle $TM\oplus T^*M$ endowed with the following structures
    \begin{align*}
        \langle X+\alpha, Y+\beta \rangle &= \iota_X \beta + \iota_Y\alpha,\\
        \rho(X+\alpha)&=X,\\
        [X+\alpha,Y+\beta]_\mathrm{D}&=[X,Y] + \mathrm{L}_X\beta - \iota_Y \de\alpha,
    \end{align*}
defined for all $X+\alpha, Y+\beta\in  \Gamma(TM\oplus T^*M)$, where $[X,Y]$ is the Lie bracket of vector fields and $[-,-]_\mathrm{D}$ denotes the Dorfman bracket of \cref{def_cour_alg}.
\end{definition}

\begin{remark}
    In this case, the differential operator mentioned in \cref{rem:cour_dorf} is just the de Rham differential, indeed
    \begin{align}\label{eq: failure of skew symmetry of Dorf}
     [X+\alpha,Y+\beta]_\mathrm{D} - [Y+\beta,X+\alpha]_\mathrm{D} = d <X+\alpha,Y+\beta>.
    \end{align}
    In accordance with \cref{rem:cour_dorf}, notice that one can also define the bundle by means of the Courant bracket, which would read
    \begin{align*}
        [X+\alpha,Y+\beta]_\mathrm{C}\coloneqq &[X+\alpha,Y+\beta]_\mathrm{D} - \de\langle\alpha,Y\rangle\\
        = &[X,Y] + \mathrm{L}_X\beta - \iota_Y \de\alpha - \de\iota_Y\alpha\\
        = & [X,Y] + \mathrm{L}_X\beta - \mathrm{L}_Y\alpha.
    \end{align*}
    This bracket is antisymmetric but does not satisfy the Jacobi identity. 
\end{remark}

From now on, we will drop the subscripts and automatically identify $[-,-]$ with the Dorfman bracket. 
It is however worth mentioning that the Dorfman bracket is not the unique bracket satisfying the properties of \cref{def_cour_alg}. Indeed, for any closed 3-form $H\in\Omega_{\mathrm{cl}}^3(M) $, we can define the twisted Dorfman bracket as
    \begin{equation*}
        [X+\alpha,Y+\beta]_H\coloneqq [X+\alpha,Y+\beta] + \iota_X\iota_Y H.
    \end{equation*}

Another interesting property of the standard Courant algebroid is that it fits in the following exact sequence
    \begin{equation*}
        0\longrightarrow T^*M \overset{\rho^*}{\longrightarrow}TM\oplus T^*M\overset{\rho}{\longrightarrow}TM\longrightarrow 0.
    \end{equation*}
\begin{definition}[$B$-transformation]
    Let $B$ be a smooth 2-form on $M$, viewed as a linear map $TM\to T^*M$ via $X\mapsto \iota_X B$. A $B$-transformation is an invertible bundle map of the kind
        \begin{equation*}
            e^B\coloneqq \begin{pmatrix}
                1 & 0 \\
                B & 1
            \end{pmatrix}\ \colon \ X+\alpha \longmapsto X+\alpha +\iota_X B,
        \end{equation*}
    where each entry of the matrix is a map from $T^{(*)}\to T^{(*)} $ depending on the position. 
\end{definition}

It is a known result (\cite{Gualtieri2004}) that automorphisms of the standard Courant algebroid are realized as $B$-transformations,\footnote{To be precise, by compositions of diffeomorphisms and $B$-transformations.} for closed $B\in\Omega^2_{\mathrm{cl}}$. Indeed one observes that
    \begin{align*}
        [e^B(X+\alpha),e^B(Y+\beta)]=[X+\alpha,Y+\beta] + \iota_X \iota_Y \de B.
    \end{align*}

\subsubsection{Dirac Structures}
We have already observed in \cref{rem:cour_dorf} and in \cref{eq: failure of skew symmetry of Dorf} that the failure of the Courant algebroid to be a Lie algebroid is ultimately related to exact terms involving $<-,->$. Restoring such property amounts to finding a subbundle $D\subset T\oplus T^*$ which is closed under the Dorfman bracket (i.e. involutive) and isotropic (i.e. $\langle\Gamma(D),\Gamma(D)\rangle=0$).

\begin{definition}[Dirac structure]
A \emph{Dirac structure} on a manifold $M$ is a subbundle $D \subseteq (T \oplus T^*)M$ whose space of sections $\Gamma(D)$ is closed under the bracket $[\cdot,\cdot]$ and that is \emph{maximally isotropic} with respect to the bilinear form $\langle \cdot, \cdot \rangle$, i.e. $\langle X+\alpha, Y+\beta \rangle=0$ for all $X+\alpha,Y+\beta\in \Gamma(D)$ and $D$ has maximal rank among all isotropic subbundles of $(T \oplus T^*)M$.
\end{definition}

The two simplest examples of Dirac structures are just $TM$ and $T^*M$. Indeed, the Dorfmann bracket  reduces to the Lie bracket of vector fields on the former and vanishes on the latter. Being half dimensional with respect to $T\oplus T^*$, they are maximal, and they are trivially isotropic.  

By ``twisting'' these two extremal cases, we obtain the following examples. 

\begin{example}[Dirac structure induced by a pre-symplectic form]
    Let $\varpi\in\Omega^2(M)$ be a pre-symplectic form, then the subbundle
        \begin{equation*}
            D_\varpi\coloneq e^\varpi(TM)=\{ X+\iota_X \varpi \ | \ X\in\Gamma(TM)\}
        \end{equation*}
    is a Dirac structure. Indeed the maximal isotropic subspace $e^\varpi(TM)$ is involutive if{f} $\de\varpi=0$, as we have seen for the case of $B$-transformations. 
\end{example}

\begin{example}[Dirac structure induced by a Poisson bivector field]\label{ex: Dirac structure induced by a Poisson bivector}
    Let $\pi\in\Gamma(\wedge^2 TM)$ be a Poisson bivector field. Then the subbundle 
        \begin{equation*}
            D_\pi\coloneq e^\pi(T^*M)=\{ \iota_\alpha\pi+ \alpha\ | \ \alpha\in\Gamma(T^*M)\}
        \end{equation*}
    is a Dirac structure. Conversely, if $D$ is Dirac and $D=e^{\pi}(T^*M)$ for some bivector $\pi$, then $\pi$ is Poisson. 
    
    To prove such statement, we start by assuming $\pi$ is Poisson. It is enough to consider the case where $\alpha=\de f$. For any $f,g\in\mathcal{C^1}(M)$ we have 
        \begin{equation*}
            \{f,g\}=\pi(\de f,\de g), \quad \text{and} \quad \iota_{\de f}\pi= X_f,
        \end{equation*}
    where $X_f$ is regarded as the Hamiltonian vector field of $f$. Involutivity is observed by noticing
        \begin{align*}
            [X_f + \de f, X_g + \de g]&= [X_f,X_g] + \de \iota_{X_f}\de g \\
            &= X_{\{f,g\}} + \de\{f,g\}.
        \end{align*}
     
    Conversely, assuming $D=e^{\pi}(T^*M)$ is Dirac, we see that the Jacobi identity is satisfied if{f} the Jacobiator Jac$(A,B,C)\coloneq [A,[B,C]] + [B,[C,A]] + [C,[A,B]]$ vanishes. In our case we see
        \begin{align*}
            &\mathrm{Jac}(X_f+\de f,X_g+\de g,X_h+\de h)=\\
            &= \iota_{\de\left(\{f,\{g,h\}\} + \{g,\{h,f\}\} +\{h,\{f,g\}\}\right)}\pi + \de\left(\{f,\{g,h\}\} + \{g,\{h,f\}\} +\{h,\{f,g\}\}\right) \\
            &=0 \quad \Leftrightarrow \quad \{f,\{g,h\}\} + \{g,\{h,f\}\} +\{h,\{f,g\}\}=0,
        \end{align*}
    which is equivalent to asking that $[\pi,\pi]=0$.\footnote{In principle, the above equation would only imply that $[\pi,\pi]$ is constant, but one can show that the only solution is provided by $[\pi,\pi]=0.$}
\end{example}

\section{Dirac corner structures in classical field theories}\label{sec: corner ft ft}
\subsection{Classical Lagrangian field theories on manifolds with corners}\label{sec: corner ft}
In the previous section we introduced a plethora of geometrical objects, whose application to field theory will be eviscerated in the following one. In particular, one can associate symplectic data to the boundary of spacetime, while on codimension 2 strata the relevant objects fit in the language of generalized geometry, where a Dirac structure naturally emerges as the graph of a Poisson structure.
As we shall see, reduction might be needed in both cases. 

Furthermore, we have that bulk, boundary and corner data can also in principle be nicely incorporated within the BF$^k$V formalism, as we will see in some simple applications. 

We start by considering  a spacetime manifold $M$ with corners. A field theory on $M$ is specified by a space of fields $F_M$, usually defined to be the space of sections $F_M=\Gamma(M,F)$ of a  vector bundle $F$, and by an action functional $S_M\in\mathcal{C^1}(F_M)$. In general, one can construct the variational bicomplex \cite{Anderson, Delgado} of local forms on $M\times F_M$ by pulling back the bicomplex of differential forms $\Omega^{\bullet,\bullet}(J^\infty F)$ along the evaluation map
    \[
    \mathrm{ev}\ \circ \ (id\times j^\infty) \ \colon \ M\times F_M \overset{}{\longrightarrow} M\times \Gamma(M,J^\infty F) \longrightarrow J^\infty F,
    \]
obtaining the following bicomplex
\[\begin{tikzcd}
	{\Omega^{\bullet,\bullet}_\mathrm{loc}(M\times F_M)} && {\Omega^{\bullet+1,\bullet}_\mathrm{loc}(M\times F_M)} \\
	\\
	{\Omega^{\bullet,\bullet+1}_\mathrm{loc}(M\times F_M)} && {\Omega^{\bullet+1,\bullet+1}_\mathrm{loc}(M\times F_M)}
	\arrow["\de", from=1-1, to=1-3]
	\arrow["\delta"', from=1-1, to=3-1]
	\arrow["\lrcorner"{anchor=center, pos=0.125}, draw=none, from=1-1, to=3-3]
	\arrow["\delta", from=1-3, to=3-3]
	\arrow["\de", from=3-1, to=3-3]
\end{tikzcd}\]
with horizontal differential d given by the de-Rham differential on $\Omega^\bullet(M)$ and the vertical differential $\delta$ given by the variational differential on $\Omega^\bullet(F_M)$. Notice that, being $D=\de+\delta$ the total differential, one obtains $\de\delta+\delta \de=0$ from $D^2=0$.

\begin{remark}
    For a $D$-dimensional field theory, the Lagrangian is a local form $L\in\Omega_\mathrm{loc}^{D,0}(M\times F_M) $ such that, when integrated along $M$, it produces a 0-form on $F_M$ which is indeed the action functional $S_M\coloneqq\int_M L$. 
\end{remark}
\subsubsection{The boundary structure}
If we assume that $M $ has a boundary $\Sigma\coloneqq \partial M$, then the varation of $S_M$ produces a boundary term after integration by parts, due to Stokes' theorem. In particular,
we have
    \[
    \delta S_M= el_M - \alpha_M,
    \]
where $el_M\in\Omega^{1}( F_M)$ is the Euler-Lagrange 1-form, containing the equations of motion of the theory, and $\alpha_M=\int_{\partial M} a_M$, for some $a_M\in\Omega_\mathrm{loc}^{D-1,1}(M\times F_M)$. We can then define the following two-form on $F_M$
    \[\varpi_M\coloneqq \delta \alpha_M,\]
which is automatically closed, hence pre-symplectic. We immediately notice that all the vector fields on $F_M$ that preserve the values of the fields (and their jets) at the boundary are in the kernel of $\varpi_M$, which is defined as the integral of  $\delta a_M\in \Omega_\mathrm{loc}^{D-1,2}(M\times F_M)$ on the codimension 1 surface $\partial M=\Sigma$. Assuming that such vector fields form a regular involutive distribution, we can quotient with respect to their flows and obtain, as the leaf space, the space of pre-boundary fields $\tilde{F}_\Sigma $, defining $\tilde{\pi}_\Sigma\colon F_M\to \tilde{F}_\Sigma $ to be the quotient map, and $\tilde{\alpha}_\Sigma $ and $\tilde{\varpi}_\Sigma$ such that
    \[
    \alpha_M=\tilde{\pi}^*_\Sigma(\tilde{\alpha}_\Sigma) \qquad \mathrm{and}\qquad \varpi_M=\tilde{\pi}^*_\Sigma(\tilde{\varpi}_\Sigma).
    \]
In general, it is possible that $\tilde{\varpi}_\Sigma$ is still degenerate, which is the case for gravity, and in which case a further reduction is needed. Also in this case one has to assume that $\Ker{\tilde{\varpi}_\Sigma}$ defines a regular involutive distribution, such that the leaf space with respect to the characteristic foliation of $\tilde{\varpi}_\Sigma$ is smooth. We denote such space by $F_\Sigma$ and we call it the 'geometric phase space'. Again, we can define the surjective submersion $\pi_\Sigma\colon F_M\to F_\Sigma$ such that
    \[
    \delta S_M= el_M - \pi^*_\Sigma\alpha_\Sigma \qquad \mathrm{and}\qquad \varpi_M=\pi_\Sigma^*(\varpi_\Sigma).
    \]  

It follows that $(F_\Sigma, \varpi_\Sigma)$ defines a symplectic manifold. 
However, this space does not yet represent the true {physical phase space} of the theory, 
as the latter corresponds to the set of {Cauchy data}. 
\begin{remark}\label{rem: tang ev eq}
To clarify this point, recall that the Euler--Lagrange equations naturally separate into two types: 
the {evolution equations}, which involve derivatives normal to the boundary, 
and the {constraint equations}, which depend only on derivatives tangent to the boundary. 
The constraints must be imposed on the space of pre--boundary fields, 
which typically enlarges the kernel of the pre--symplectic form. 
Performing the corresponding symplectic reduction yields the \emph{reduced phase space}.

\end{remark}
More conceptually, consider the cylindrical manifold 
$M_\epsilon \coloneqq  \Sigma \times [0, \epsilon]$ for some small $\epsilon > 0$. 
Its boundary decomposes as 
\[
\partial M = (\Sigma \times \{0\}) \sqcup (\Sigma \times \{\epsilon\}).
\]
The space $\tilde{L}_{M_\epsilon}$ can then be viewed as a relation%
\footnote{That is, a subset of $A \times B$ for two sets $A$ and $B$.} 
between the pre--boundary field spaces 
$\tilde{F}_\Sigma \simeq \tilde{F}_{\Sigma \times \{0\}}$ and 
$\tilde{F}_{\Sigma \times \{\epsilon\}}$.

The space of {Cauchy data}, denoted $\tilde{C}_\Sigma$, 
consists of those pre--boundary configurations on $\Sigma$ that can be extended 
to full solutions of the Euler--Lagrange equations within a small cylindrical neighborhood of $\Sigma$, i.e.
\[
\tilde{C}_\Sigma \coloneqq 
\left\{
c \in \tilde{F}_\Sigma \simeq \tilde{F}_{\Sigma \times \{0\}} 
\;\middle|\;
\exists\, u \in \tilde{F}_{\Sigma \times \{\epsilon\}} 
\text{ such that } (c,u) \in \tilde{L}_{M_\epsilon}
\right\}.
\]

The induced two--form $\tilde{\varpi}^C_\Sigma$ is in general degenerate on $\tilde{C}_\Sigma$, 
and taking its quotient yields the {reduced phase space} $\underline{C}_\Sigma$. 
Although this space is often non-smooth or singular, in field theory one is primarily interested 
in the algebra of functions defined on it, namely the algebra of {physical observables}. 
Under suitable assumptions, this algebra can be obtained cohomologically 
through the BFV formalism.

\begin{remark}
    In principle one could work with the space of boundary fields $F_\Sigma$ instead of $\tilde F_\Sigma$ and define $C_\Sigma\subset F_\Sigma$, however, one first needs to make sure that the constraints defining $C_\Sigma$ are invariant under the action of the vector fields inside $\Ker{\tilde{\varpi}_\Sigma}$. For Palatini--Cartan gravity in $D=4$ we will see that this is not true in general, however it is possible to use the non-invariant constraint to cleverly fix a representative of the fields in the quotient space $F_\Sigma$, such that the residual constraints on it are indeed invariant under the action of $\Ker{\tilde{\varpi}_\Sigma}$. For more details on the general construction, we refer to \cite{C23}.
\end{remark}

\subsubsection{The corner structure}\label{sec: classical corner}

If one assumes that $M$ has a corner $\Gamma$, the boundary structure of the
field theory is enriched with extra data. We consider the case in which $C_\Sigma$
is coisotropic and described by a family of constraints
$\{\phi_i\}_{i\in I}\subset C^\infty(F_\Sigma)$, possibly dependent, which are first
class:
\begin{align}
        \{\phi_i,\phi_j\}
        &=
        f_{ij}^{k}\phi_k,
        &
        f_{ij}^{k}
        &\in
        C^\infty(F_\Sigma).
        \label{eq:corner-first-class-abstract}
\end{align}
We write $\approx$ for equality modulo the ideal generated by the constraints.

Assume first that $\partial\Sigma=\emptyset$. If $\mathbb X_i$ denotes the Hamiltonian
vector field of $\phi_i$, then
\begin{align}
        \iota_{\mathbb X_i}\varpi_\Sigma
        &=
        \delta\phi_i.
        \label{eq:ham-vf-no-corner}
\end{align}
Hence, on shell,
\begin{align}
        [\mathbb X_i,\mathbb X_j]
        &\approx
        f_{ij}^{k}\mathbb X_k .
        \label{eq:ham-vf-on-shell-closure}
\end{align}
Equivalently, setting
\begin{align*}
        \phi_{ij}
        &:=
        \{\phi_i,\phi_j\},
        &
        \mathbb X_{ij}
        &:=
        [\mathbb X_i,\mathbb X_j],
\end{align*}
one has
\begin{align}
        \iota_{\mathbb X_{ij}}\varpi_\Sigma
        &=
        \delta\phi_{ij}.
        \label{eq: ham vf invol}
\end{align}

We now allow the boundary to have a boundary,
\begin{align*}
        \Gamma
        &=
        \partial\Sigma
        \neq
        \emptyset .
\end{align*}
Then \cref{eq:ham-vf-no-corner} is replaced by
\begin{align}
        \delta\phi_i
        &=
        \iota_{\mathbb X_i}\varpi_\Sigma
        -
        \mathcal X_i^\Sigma,
        \label{eq:corner-ham-equation}
\end{align}
where $\mathcal X_i^\Sigma\in\Omega^1(F_\Sigma)$ is a residual one-form supported at
the corner.\\
The splitting in \cref{eq:corner-ham-equation} is not canonical. Indeed, for a corner
functional $B_i\in C^\infty(\tilde F_\Gamma)$, one may set
\begin{align}
        \widehat\phi_i
        &:=
        \phi_i+\tilde\pi_\Gamma^*B_i,
        &
        \widehat{\mathcal X}_i
        &:=
        \tilde{\mathcal X}_i-\delta B_i.
        \label{eq:corner-exact-ambiguity}
\end{align}
Then
\begin{align*}
        \delta\widehat\phi_i
        &=
        \iota_{\mathbb X_i}\varpi_\Sigma
        -
        \tilde\pi_\Gamma^*\widehat{\mathcal X}_i .
\end{align*}
Thus, adding an exact corner term to the constraint changes the representative of the
residual one-form by an exact one-form.

In what follows we impose the following simplifying assumption.

\begin{assumption}\label{ass:undifferentiated-multipliers}
When the constraints are written by pairing equations with Lagrange multipliers, these
multipliers are not differentiated in the functionals $\phi_i$.
\end{assumption}

At this point, one can repeat the previous reduction argument. Any vector field on
$F_\Sigma$ which preserves the values of the fields, and of their jets, on $\Gamma$
leaves the residual forms $\mathcal X_i^\Sigma$ invariant. Quotienting by the
corresponding characteristic foliation defines the space of pre-corner fields
$\tilde F_\Gamma$, together with a quotient map
\begin{align*}
        \tilde\pi_\Gamma
        &:F_\Sigma\longrightarrow \tilde F_\Gamma
\end{align*}
such that
\begin{align}
        \delta\phi_i
        &=
        \iota_{\mathbb X_i}\varpi_\Sigma
        -
        \tilde\pi_\Gamma^*(\tilde{\mathcal X}_i),
        &
        \tilde{\mathcal X}_i
        &\in
        \Omega^1(\tilde F_\Gamma).
        \label{eq:corner-residual-reduced}
\end{align}
The Hamiltonian vector fields $\mathbb X_i$ are projectable along $\tilde\pi_\Gamma$;
we denote their push-forwards by
\begin{align*}
        \tilde{\mathbb X}_i
        &\in
        \mathfrak X(\tilde F_\Gamma).
\end{align*}
Therefore each boundary constraint determines a generalized section
\begin{align*}
        \tilde{\mathbb X}_i+\tilde{\mathcal X}_i
        &\in
        \Gamma((T\oplus T^*)\tilde F_\Gamma).
\end{align*}

As in \cref{subsec: gen geom}, $(T\oplus T^*)\tilde F_\Gamma$ is the standard Courant
algebroid on $\tilde F_\Gamma$, with pairing, Dorfman bracket and anchor
\begin{align*}
        \langle \mathbb X+\mathcal X,\mathbb Y+\mathcal Y\rangle
        &=
        \iota_{\mathbb X}\mathcal Y
        +
        \iota_{\mathbb Y}\mathcal X,
        \\
        [\mathbb X+\mathcal X,\mathbb Y+\mathcal Y]
        &=
        [\mathbb X,\mathbb Y]
        +
        \mathrm L_{\mathbb X}\mathcal Y
        -
        \iota_{\mathbb Y}\delta\mathcal X,
        \\
        \rho(\mathbb X+\mathcal X)
        &=
        \mathbb X .
\end{align*}

\begin{remark}
In the next sections, in order to readily obtain a BFV description of the boundary
theory, we will define the constraints $\phi_i$ by means of odd Lagrange
multipliers.\footnote{Identified with the ghosts of the theory.} Thus both the
Hamiltonian vector fields and the residual one-forms are odd. With this convention, the
pairing and the Dorfman bracket are read with the corresponding graded signs:
\begin{align*}
        \langle \mathbb X+\mathcal X,\mathbb Y+\mathcal Y\rangle
        &=
        \iota_{\mathbb X}\mathcal Y
        +
        (-1)^{xy}
        \iota_{\mathbb Y}\mathcal X,
        \\
        [\mathbb X+\mathcal X,\mathbb Y+\mathcal Y]
        &=
        [\mathbb X,\mathbb Y]
        +
        \mathrm L_{\mathbb X}\mathcal Y
        -
        (-1)^{xy}
        \iota_{\mathbb Y}\delta\mathcal X,
\end{align*}
where $x$ and $y$ are the parities of $\mathcal X$ and $\mathcal Y$, respectively,
and
\begin{align*}
        \mathrm L_{\mathbb X}
        &:=
        [\iota_{\mathbb X},\delta]
        =
        \iota_{\mathbb X}\delta
        -
        (-1)^x\delta\iota_{\mathbb X}.
\end{align*}
\end{remark}

Let
\begin{align*}
        \tilde D
        &:=
        \operatorname{span}_{C^\infty(\tilde F_\Gamma)}
        \{\tilde{\mathbb X}_i+\tilde{\mathcal X}_i\}_{i\in I}
        \subset
        \Gamma((T\oplus T^*)\tilde F_\Gamma).
\end{align*}
The following two criteria express isotropy and involutivity in terms of the original
boundary constraints.

\begin{lemma}\label{lem:corner-isotropy-criterion}
The distribution $\tilde D$ is isotropic if and only if
\begin{align}
        \mathbb X_i(\phi_j)+\mathbb X_j(\phi_i)
        &=
        0
        \qquad
        \forall i,j.
        \label{eq:corner-isotropy-criterion}
\end{align}
\end{lemma}

\begin{proof}
Using \cref{eq:corner-residual-reduced}, we have
\begin{align*}
        \iota_{\mathbb X_j}\iota_{\mathbb X_i}\varpi_\Sigma
        &=
        \mathbb X_j(\phi_i)
        +
        \iota_{\tilde{\mathbb X}_j}\tilde{\mathcal X}_i,
        \\
        \iota_{\mathbb X_i}\iota_{\mathbb X_j}\varpi_\Sigma
        &=
        \mathbb X_i(\phi_j)
        +
        \iota_{\tilde{\mathbb X}_i}\tilde{\mathcal X}_j .
\end{align*}
Adding the two identities and using the skew-symmetry of $\varpi_\Sigma$ gives
\begin{align*}
        0
        &=
        \mathbb X_j(\phi_i)
        +
        \mathbb X_i(\phi_j)
        +
        \iota_{\tilde{\mathbb X}_j}\tilde{\mathcal X}_i
        +
        \iota_{\tilde{\mathbb X}_i}\tilde{\mathcal X}_j
        \\
        &=
        \mathbb X_j(\phi_i)
        +
        \mathbb X_i(\phi_j)
        +
        \left\langle
        \tilde{\mathbb X}_i+\tilde{\mathcal X}_i,
        \tilde{\mathbb X}_j+\tilde{\mathcal X}_j
        \right\rangle .
\end{align*}
Thus the Courant pairing vanishes precisely when
\cref{eq:corner-isotropy-criterion} holds.
\end{proof}

\begin{lemma}\label{lem:corner-involutivity-criterion}
Assume that
\begin{align}
        \mathbb X_i(\phi_j)
        &\approx
        f_{ij}^{k}\phi_k .
        \label{eq:corner-on-shell-constraint-action}
\end{align}
Then, $\tilde D$ is involutive on shell:
\begin{align}
        [
        \tilde{\mathbb X}_i+\tilde{\mathcal X}_i,
        \tilde{\mathbb X}_j+\tilde{\mathcal X}_j
        ]
        &\approx
        f_{ij}^{k}
        \left(
        \tilde{\mathbb X}_k+\tilde{\mathcal X}_k
        \right).
        \label{eq:corner-on-shell-dorfman-closure}
\end{align}
\end{lemma}

\begin{proof}
From \cref{eq:corner-ham-equation},
\begin{align*}
        \mathcal X_i^\Sigma
        &=
        \iota_{\mathbb X_i}\varpi_\Sigma
        -
        \delta\phi_i .
\end{align*}
Since $\delta\varpi_\Sigma=0$, one has
\begin{align*}
        \delta\mathcal X_i^\Sigma
        &=
        \mathrm L_{\mathbb X_i}\varpi_\Sigma .
\end{align*}
Therefore
\begin{align*}
        \mathrm L_{\mathbb X_i}\mathcal X_j^\Sigma
        &=
        \mathrm L_{\mathbb X_i}\iota_{\mathbb X_j}\varpi_\Sigma
        -
        \delta\bigl(\mathbb X_i(\phi_j)\bigr),
        \\
        \iota_{\mathbb X_j}\delta\mathcal X_i^\Sigma
        &=
        \iota_{\mathbb X_j}\mathrm L_{\mathbb X_i}\varpi_\Sigma .
\end{align*}
Taking the difference gives
\begin{align}
        \mathrm L_{\mathbb X_i}\mathcal X_j^\Sigma
        -
        \iota_{\mathbb X_j}\delta\mathcal X_i^\Sigma
        &=
        \iota_{[\mathbb X_i,\mathbb X_j]}\varpi_\Sigma
        -
        \delta\bigl(\mathbb X_i(\phi_j)\bigr).
        \label{eq:corner-dorfman-oneform-identity}
\end{align}
Using \cref{eq:corner-on-shell-constraint-action} and
\cref{eq:ham-vf-on-shell-closure}, we obtain, on shell,
\begin{align*}
        \delta\bigl(\mathbb X_i(\phi_j)\bigr)
        &\approx
        f_{ij}^{k}\delta\phi_k,
        \\
        \iota_{[\mathbb X_i,\mathbb X_j]}\varpi_\Sigma
        &\approx
        f_{ij}^{k}\iota_{\mathbb X_k}\varpi_\Sigma
        =
        f_{ij}^{k}
        \left(
        \delta\phi_k+\mathcal X_k^\Sigma
        \right).
\end{align*}
Substitution in \cref{eq:corner-dorfman-oneform-identity} yields
\begin{align*}
        \mathrm L_{\mathbb X_i}\mathcal X_j^\Sigma
        -
        \iota_{\mathbb X_j}\delta\mathcal X_i^\Sigma
        &\approx
        f_{ij}^{k}\mathcal X_k^\Sigma .
\end{align*}
Together with
\begin{align*}
        [\mathbb X_i,\mathbb X_j]
        &\approx
        f_{ij}^{k}\mathbb X_k,
\end{align*}
this is precisely the on-shell closure of the corresponding Dorfman bracket. Since the
objects involved are projectable to the pre-corner field space, the identity descends to
$\tilde F_\Gamma$.
\end{proof}

The condition \cref{eq:corner-on-shell-dorfman-closure} is not automatic on the full
space $\tilde F_\Gamma$; in general, the corner fields must be restricted by the
constraints inherited from the boundary equations. Such constraints are identified by
considering the infinitesimal cylinder
\begin{align*}
        \Sigma_\epsilon
        &:=
        \Gamma\times[0,\epsilon],
        &
        \epsilon
        &>
        0.
\end{align*}
On $\Sigma_\epsilon$, let
\begin{align*}
        C_{\Sigma_\epsilon}
        &\subset
        F_{\Sigma_\epsilon}
\end{align*}
be the coisotropic submanifold of boundary fields. We define $F_\Gamma$ as the set of
field configurations at $\Gamma\times\{0\}$ of fields in $C_{\Sigma_\epsilon}$. In
other words, $F_\Gamma$ is the space of corner fields that extend to boundary
conditions\footnote{That is, to Cauchy data in $C_{\Sigma_\epsilon}$.} in an
infinitesimal neighbourhood of $\Gamma$.

By abuse of notation, we still denote by
$\tilde{\mathbb X}_i+\tilde{\mathcal X}_i$ the sections of
$(T\oplus T^*)F_\Gamma$ induced by the corresponding sections on $\tilde F_\Gamma$.
We then impose the following assumption.

\begin{assumption}\label{ass:pre-dirac-corner}
The distribution
\begin{align*}
        \tilde D
        &:=
        \operatorname{span}_{C^\infty(F_\Gamma)}
        \{
        \tilde{\mathbb X}_i+\tilde{\mathcal X}_i
        \}_{i\in I}
        \subset
        \Gamma((T\oplus T^*)F_\Gamma)
\end{align*}
is involutive and isotropic with respect to the Dorfman bracket and the Courant pairing.
We call such data a pre-Dirac structure.
\end{assumption}

We will see in Section~\ref{sec: predirac GR} that this assumption is satisfied for
Palatini--Cartan gravity.

It remains to identify the reduction which removes the obstruction to maximality. Let
\begin{align*}
        \operatorname{pr}_{T^*}
        &:
        (T\oplus T^*)F_\Gamma
        \longrightarrow
        T^*F_\Gamma
\end{align*}
be the projection onto the cotangent component. We define the kernel of $\tilde D$ by
\begin{align}
        K
        &:=
        \Bigl\{
        \mathbb Y\in\mathfrak X(F_\Gamma)
        \;:\;
        \iota_{\mathbb Y}\mathcal X=0
        \quad
        \forall\,
        \mathcal X\in
        \Gamma(\operatorname{pr}_{T^*}\tilde D)
        \Bigr\}.
        \label{eq:corner-kernel-general}
\end{align}
Equivalently, since $\operatorname{pr}_{T^*}\tilde D$ is generated by the residual
one-forms,
\begin{align*}
        K
        &=
        \Bigl\{
        \mathbb Y\in\mathfrak X(F_\Gamma)
        \;:\;
        \iota_{\mathbb Y}\tilde{\mathcal X}_i=0
        \quad
        \forall i
        \Bigr\}.
\end{align*}

\begin{lemma}\label{lem:kernel-involutive-criterion}
The distribution $K$ is involutive if and only if
\begin{align}
        \iota_{\mathbb Y}\iota_{\mathbb Y'}\delta\mathcal X
        &=
        0
        \qquad
        \forall\,\mathbb Y,\mathbb Y'\in K,
        \quad
        \forall\,
        \mathcal X\in\Gamma(\operatorname{pr}_{T^*}\tilde D).
        \label{eq:corner-kernel-involutive-criterion}
\end{align}
\end{lemma}

\begin{proof}
Let $\mathbb Y,\mathbb Y'\in K$. Then
$[\mathbb Y,\mathbb Y']\in K$ if and only if
\begin{align*}
        \iota_{[\mathbb Y,\mathbb Y']}\mathcal X
        &=
        0
        \qquad
        \forall\,
        \mathcal X\in\Gamma(\operatorname{pr}_{T^*}\tilde D).
\end{align*}
By Cartan calculus,
\begin{align*}
        \iota_{[\mathbb Y,\mathbb Y']}\mathcal X
        &=
        [\mathrm L_{\mathbb Y},\iota_{\mathbb Y'}]\mathcal X
        \\
        &=
        \mathrm L_{\mathbb Y}(\iota_{\mathbb Y'}\mathcal X)
        -
        \iota_{\mathbb Y'}\mathrm L_{\mathbb Y}\mathcal X .
\end{align*}
The first term vanishes since $\mathbb Y'\in K$. Moreover,
\begin{align*}
        \mathrm L_{\mathbb Y}\mathcal X
        &=
        \iota_{\mathbb Y}\delta\mathcal X
        +
        \delta(\iota_{\mathbb Y}\mathcal X)
        =
        \iota_{\mathbb Y}\delta\mathcal X,
\end{align*}
because $\mathbb Y\in K$. Therefore
\begin{align*}
        \iota_{[\mathbb Y,\mathbb Y']}\mathcal X
        &=
        -
        \iota_{\mathbb Y'}\iota_{\mathbb Y}\delta\mathcal X.
\end{align*}
This proves the claim.
\end{proof}

\begin{definition}\label{def:partial-kernel}
A partial kernel of $\tilde D$ is a regular involutive subdistribution
\begin{align*}
        \widehat K
        &\subset
        K .
\end{align*}
\end{definition}

The distribution used above to pass from $F_\Sigma$ to $\tilde F_\Gamma$, consisting
of vector fields which preserve the values of fields and jets at the corner, is an example
of such a partial kernel at the pre-corner level.

Another canonical example is
\begin{align}
        K_{\tilde D}
        &:=
        \Bigl\{
        \mathbb Y\in\mathfrak X(F_\Gamma)
        \;:\;
        \mathbb Y+0\in\Gamma(\tilde D)
        \Bigr\}.
        \label{eq:corner-KD}
\end{align}

\begin{lemma}\label{lem:KD-partial-kernel}
The distribution $K_{\tilde D}$ is a partial kernel.
\end{lemma}

\begin{proof}
If $\mathbb Y\in K_{\tilde D}$, then $\mathbb Y+0\in\Gamma(\tilde D)$. Since
$\tilde D$ is isotropic, for every $\mathbb X+\mathcal X\in\Gamma(\tilde D)$ one has
\begin{align*}
        0
        &=
        \langle
        \mathbb X+\mathcal X,
        \mathbb Y+0
        \rangle
        =
        \iota_{\mathbb Y}\mathcal X .
\end{align*}
Thus $\mathbb Y\in K$. Moreover, if
$\mathbb Y,\mathbb Y'\in K_{\tilde D}$, then
\begin{align*}
        [\mathbb Y+0,\mathbb Y'+0]
        &=
        [\mathbb Y,\mathbb Y']+0 .
\end{align*}
Since $\tilde D$ is involutive, $[\mathbb Y,\mathbb Y']+0\in\Gamma(\tilde D)$, hence
$[\mathbb Y,\mathbb Y']\in K_{\tilde D}$.
\end{proof}

\begin{remark}
For $\mathbb X+\mathcal X\in\Gamma(\tilde D)$ and
$\mathbb Y\in K_{\tilde D}$, the Dorfman bracket gives
\begin{align}
        [
        \mathbb X+\mathcal X,
        \mathbb Y+0
        ]
        &=
        [\mathbb X,\mathbb Y]
        -
        \iota_{\mathbb Y}\delta\mathcal X.
        \label{eq:projectability-KD}
\end{align}
Thus
\begin{align*}
        \iota_{\mathbb Y}\delta\mathcal X
        &=
        0
        \qquad
        \Longleftrightarrow
        \qquad
        [\mathbb X,\mathbb Y]\in K_{\tilde D}.
\end{align*}
This is the condition for $\mathbb X$ to be projectable along $K_{\tilde D}$.
\end{remark}

A stronger partial kernel is
\begin{align}
        K_f
        &:=
        \Bigl\{
        \mathbb Y\in K
        \;:\;
        \iota_{\mathbb Y}\delta\mathcal X=0
        \quad
        \forall\,
        \mathcal X\in\Gamma(\operatorname{pr}_{T^*}\tilde D)
        \Bigr\}
        \notag
        \\
        &=
        \Bigl\{
        \mathbb Y\in\mathfrak X(F_\Gamma)
        \;:\;
        \iota_{\mathbb Y}\mathcal X=0
        \text{ and }
        \iota_{\mathbb Y}\delta\mathcal X=0
        \quad
        \forall\,
        \mathcal X\in\Gamma(\operatorname{pr}_{T^*}\tilde D)
        \Bigr\}.
        \label{eq:corner-Kf}
\end{align}

We shall use the following equivalence relation between generalized distributions. Two
pre-Dirac distributions are considered equivalent if either they are related by a
diffeomorphism or by a $B$-transformation, or one is obtained from the other by
reduction along a regular partial kernel. Thus, if
\begin{align*}
        \widehat F
        &=
        F_\Gamma/\widehat K
\end{align*}
for a regular partial kernel $\widehat K\subset K$, and if the sections of $\tilde D$
are projectable, the reduced distribution $\widehat D$ on $\widehat F$ is equivalent
to $\tilde D$. In the special case $\widehat K=K_{\tilde D}$, this is a
Morita-type equivalence.

In particular, when $K$ itself is regular and involutive, one can take
\begin{align*}
        P_\Gamma
        &:=
        F_\Gamma/K .
\end{align*}
The induced distribution is
\begin{align}
        D
        &:=
        \operatorname{span}_{C^\infty(P_\Gamma)}
        \{
        \mathbb X_i+\mathcal X_i
        \}_{i\in I}
        \subset
        \Gamma((T\oplus T^*)P_\Gamma).
        \label{eq:reduced-dirac-distribution-general}
\end{align}
If involutivity is preserved and $D$ is maximally isotropic, then $D$ is a Dirac
structure on $P_\Gamma$. In the case in which its cotangent projection is all
$T^*P_\Gamma$, it is the graph of a Poisson bivector
\begin{align*}
        \Pi
        &\in
        \Gamma(\wedge^2TP_\Gamma).
\end{align*}

\begin{definition}\label{def:good-field-theory-corner}
We say that the field theory is good on $\Gamma$ if its pre-Dirac corner structure is
equivalent, in the above sense, to a Poisson structure on a reduced corner field space.
\end{definition}

To summarize the discussion, it is useful to keep in mind the following diagram:
\begin{center}
\begin{tikzcd}
{\tilde F_\Sigma}
\\
{F_\Sigma}
&&
{F_{\Sigma_\epsilon}}
&&
{\tilde F_\Gamma}
\\
{C_\Sigma}
&&
{C_{\Sigma_\epsilon}}
&&
{F_\Gamma}
\\
&&&&
{P_\Gamma}
\arrow["{/\ker\tilde\varpi_\Sigma}"', from=1-1, to=2-1]
\arrow["{\tilde\pi_\Gamma}", curve={height=-30pt}, from=2-1, to=2-5]
\arrow["{\tilde\pi_\Gamma^0}"', from=2-3, to=2-5]
\arrow["{\iota_{C_\Sigma}}", hook', from=3-1, to=2-1]
\arrow["{\iota_\Sigma}"', hook', from=3-3, to=2-3]
\arrow["{\pi_\Gamma}", from=3-3, to=3-5]
\arrow["{}"', curve={height=18pt}, from=3-3, to=4-5]
\arrow["{\iota_\Gamma}", hook, from=3-5, to=2-5]
\arrow["{/K}"', from=3-5, to=4-5]
\end{tikzcd}
\end{center}

\begin{remark}
Given a Poisson structure $\Pi\in\Gamma(\wedge^2TP_\Gamma)$, one automatically obtains
a BF${}^2$V structure on $T^*[1]P_\Gamma$, since $\Pi$ can equivalently be
regarded as a degree-two functional
\begin{align*}
        \mathcal S_\Gamma
        &\in
        C^\infty(T^*[1]P_\Gamma)
\end{align*}
satisfying the classical master equation
\begin{align*}
        \{\mathcal S_\Gamma,\mathcal S_\Gamma\}
        &=
        0,
\end{align*}
which is equivalent to
\begin{align*}
        [\Pi,\Pi]
        &=
        0.
\end{align*}
\end{remark}

\begin{remark}
In some cases, such as $BF$ theory, ghost-for-ghost fields affect the corner structure.
The procedure described above may still provide a Dirac structure given by the graph of a
Poisson bivector, but this does not necessarily capture the full homological nature of
the theory. Reducibility data require a higher Poisson structure, or $P_\infty$
structure, which in the cases considered here is concentrated in degrees one and two.

In other words, the ordinary Poisson bivector $\Pi$ is replaced by
\begin{align*}
        \Pi
        &=
        \Pi_1+\Pi_2,
\end{align*}
where $\Pi_1$ is a degree-one vector field and $\Pi_2$ is a bivector. These data satisfy
the Maurer--Cartan equation
\begin{align*}
        [\Pi_1+\Pi_2,\Pi_1+\Pi_2]
        &=
        0,
\end{align*}
where $[\cdot,\cdot]$ denotes the Schouten bracket. Equivalently,
\begin{align*}
        [\Pi_1,\Pi_1]
        &=
        0,
        &
        [\Pi_1,\Pi_2]
        &=
        0,
        &
        [\Pi_2,\Pi_2]
        &=
        0.
\end{align*}
The first identity says that $\Pi_1$ is cohomological, the third is the Poisson
condition for $\Pi_2$, and the second says that $\Pi_2$ is invariant under the
differential generated by $\Pi_1$. Thus the passage from a Poisson bivector to
$\Pi_1+\Pi_2$ incorporates residual reducibility directly into the cohomological
corner description. In practice, this amounts to a Dirac structure with support; see
\cite{BCCS,CC2022cor}.
\end{remark}
    
\subsection{Examples}
We now provide a couple of straightforward examples of field theories presenting corner structures that can be described within the BF$^2$V formalism. In the following, we let $G$ be a semi-simple Lie group with Lie algebra $\mathfrak{g}$ and $M$ be a $D=d+1$ manifold.
\subsubsection{Yang--Mills theory}
    We work in the so-called first order formalism, where only first order derivatives appear in the Lagrangian. The space of fields of the theory is given by $F^{\mathrm{YM}}_M\coloneqq\Omega^1(M,\mathfrak{g})\oplus \Omega^{d-1}(M,\mathfrak{g})\ni(A,B) $, and the action reads
    \[
    S_M^{\mathrm{YM}}=\int_M \Tr(B \wedge F_A) - \frac{1}{2} \Tr (B \wedge * B),
    \]
    where $*$ is the Hodge dual. From now on we will omit the $\wedge$ symbol and the $\Tr$, as they will be assumed in the relevant contexts.

    The variation of the action produces EoMs and a boundary term
    \[
    \delta S_M^{\mathrm{YM}}=\int_M  \mathrm{d}_A B \delta A + (B-*F_A)\delta B
    -\int_{\partial M} B \delta A.
    \]
    Setting $\Sigma=\partial M$, we define $F^{\mathrm{YM}}_\Sigma=\Omega^1(\Sigma,\mathfrak{g})\oplus\Omega^{d-1}(\Sigma,\mathfrak{g})\ni(A,B) $ as the space of boundary fields, on which we have the symplectic form
        \[
        \varpi^{\mathrm{YM}}_\Sigma=\int_\Sigma \delta B \delta A.
        \]
    The phase space  $C^{\mathrm{YM}}_\Sigma$ is defined as the zero locus of the following constraint, which turns out to be first class
        \[
        L_c=\int_\Sigma c \mathrm{d}_AB, \qquad \{L_c,L_c\}=\frac{1}{2}L_{[c,c]}, \qquad \text{for }c\in\Omega^0[1](\Sigma,\mathfrak{g}),
        \]
    having chosen an odd Lagrange multiplier $c$. The Hamiltonian vector field $\mathbb L_c \coloneqq \int_\Sigma \mathbb L_A \frac{\delta}{\delta A} + \mathbb L_B \frac{\delta}{\delta B}$ of $L_c$ gives the infinitesimal gauge transformations as $\mathbb L_A=\mathrm{d}_A c $ and $\mathbb L_B= [c,B] $. %The BFV description is carried out straightforwardly by defining the BFV action on $\mathcal{F}^{\mathrm{YM}}_\Sigma\coloneqq F^{\mathrm{YM}}_\Sigma \times T^*\Omega^0[1](\Sigma,\mathfrak{g})\ni(A,B,c,c^\dag)$ by twisting the Koszul and Chevalley--Eilenberg differentials, obtaining
    %\[
    %\mathcal{S}^{\mathrm{YM}}_\Sigma=\int c \mathrm{d}_A B +\frac{1}{2}c^\dag [c,c], \qquad \quad \varpi^\mathrm{YM}_\Sigma=\int_\Sigma \delta A \delta B + \delta c \delta c^\dag.
    %\]

    In the presence of corners, we have an extra term in the variation of $L_c$, i.e.
        \[
        \delta L_c = \iota_{\mathbb L_c}\varpi^{\mathrm{YM}}_\Sigma - \int_\Gamma c\delta B.
        \]
    Defining $\mathcal{J}_c\coloneqq \int_\Gamma c\delta B$ as a functional on the space of pre-corner fields $\tilde{F}^{\mathrm{YM}}_\Gamma\coloneqq \Omega^1(\Gamma,\mathfrak{g})\oplus\Omega^{d-1}(\Gamma,\mathfrak{g}) $, we can furthermore define 
        \[
        \mathbb J_c = \int_\Gamma \mathrm{d}_A c\frac{\delta}{\delta A} + [c,B]\frac{\delta}{\delta B} \in \mathfrak{X}(\tilde{F}^{\mathrm{YM}}_\Gamma)
        \]
    as the pushforward of $\mathbb L_c$ on $\tilde{F}^{\mathrm{YM}}_\Gamma$ via the projection map $\tilde{\pi}_\Gamma\colon F^{\mathrm{YM}}_\Sigma\to \tilde{F}_\Gamma^{\mathrm{YM}}$. Then the distribution $\tilde{D}^{\mathrm{YM}}_\Gamma\coloneqq \mathrm{span}(\mathbb J_c + \mathcal{J}_c) $ is involutive and isotropic, as can easily be checked by noticing
        \begin{align*}
            &\langle \mathbb J_c + \mathcal{J}_c, \mathbb J_{c'} + \mathcal{J}_{c'}\rangle = \mathcal{J}_{[c,c']} - \mathcal{J}_{[c',c]}=0,\\
            &[\mathbb J_c + \mathcal{J}_c,\mathbb J_c + \mathcal{J}_c]= \frac{1}{2}\mathbb J_{[c,c]} + \frac{1}{2}\mathcal{J}_{[c,c]}.
        \end{align*}
    \begin{remark}
        Following the previous section, we would have to work on the restriction to the corners of the space of Cauchy data. However, in the absence of transversal jet fields, we see that the constraint $\mathrm{d}_AB=0$ does not extend to the $d-1$-dimensional corner, it being a $d-$form. Therefore, in the notation we provided above, we can identify $F_\Gamma^{\mathrm{YM}}$ and $\tilde F_\Gamma^{\mathrm{YM}}$.
    \end{remark}
    In order to achieve the maximal isotropicity property, we quotient out with respect to the distribution defined by the kernel of $\mathcal{J}_c$, given by
        \[
        \Ker{\mathcal{J}_c}=\bigg\{ \int_\Gamma \mathbb X_A \frac{\delta}{\delta A} + \mathbb X_B \frac{\delta}{\delta B} \in\mathfrak{X}(\tilde{F}^{\mathrm{YM}}_\Gamma) \ \big\vert \ \mathbb X_B = 0   \bigg\} \simeq \Omega^1(\Gamma,\mathfrak{g}).
        \]
    Therefore we obtain that the space of classical corner fields is simply given by $P^{\mathrm{YM}}_\Gamma= \tilde{F}^{\mathrm{YM}}_\Gamma/\Omega^1(\Gamma,\mathfrak{g})\simeq \Omega^{d-1}(\Gamma,\mathfrak{g}) $, on which one can define the Poisson bivector field
        \[
        \pi^{\mathrm{YM}}_\Gamma=\int_\Gamma \frac{1}{2}B\left[ \frac{\delta}{\delta B},\frac{\delta}{\delta B} \right] \quad \text{such that } D^{\mathrm{YM}}_\Gamma=\mathrm{Graph}(\pi^{\mathrm{YM}}_\Gamma),
        \]
        such that $\iota_{\mathcal{J}_c}\pi=\mathbb J_c$, 
    proving that $D^{\mathrm{YM}}_\Gamma$ is a Dirac structure.

    Notice also that the Poisson bivector field $\pi^{\mathrm{YM}}_\Gamma$ is the same as a degree-2 function on $T^*[1] P^{\mathrm{YM}}_\Gamma$ satisfying the CME. In particular, one has
        \begin{align}\label{eq: YM BF^2 action}
            \mathcal{S}^{\mathrm{YM}}_\Gamma=\int_\Gamma \frac{1}{2}B [c,c],
        \end{align} 
    where $(B,c)\in\mathcal{F}^{\mathrm{YM}}_\Gamma=\Omega^{d-1}(\Gamma,\mathfrak{g})\oplus \Omega^0[1](\Gamma,\mathfrak{g}) $, after having identified $\mathfrak{g}^*\simeq \mathfrak{g}$ and $(\Omega^{d-1}(\Gamma))^*\simeq \Omega^{0}(\Gamma)$. Furthermore, on $\mathcal{F}^{\mathrm{YM}}_\Gamma$ one can define the canonical symplectic form
        \begin{equation}\label{eq: YM BF^2 symp}
            \varpi^{\mathrm{YM}}_\Gamma=\int_\Gamma \delta B \delta c.
        \end{equation}
    
    %\begin{remark}
    %    Notice that $\mathcal{S}^{\mathrm{YM}}_\Gamma$ is the BF$^2$V extension of $\mathcal{S}^{\mathrm{YM}}_\Sigma$. Indeed one sees
    %        \[
    %        \delta \mathcal{S}^{\mathrm{YM}}_\Sigma = \iota_{Q^{\mathrm{YM}}_\Sigma}\varpi^{\mathrm{YM}}_\Sigma - \pi^*(\vartheta^{\mathrm{YM}}_\Gamma),
    %        \]
    %    such that one finds \cref{eq: YM BF^2 action} and \cref{eq: YM BF^2 symp} as 
    %        \begin{align*}
    %            &\delta\vartheta^{\mathrm{YM}}_\Gamma= \varpi^{\mathrm{YM}}_\Gamma, \\
    %           &\iota_{Q^{\mathrm{YM}}_\Sigma}\iota_{Q^{\mathrm{YM}}_\Sigma}\varpi^{\mathrm{YM}}_\Sigma=2\pi^*(\mathcal{S}^{\mathrm{YM}}_\Gamma).
    %        \end{align*}
    %\end{remark}
    \begin{remark}
        For the sake of completeness, we remark that the same result could be achieved by keeping the transversal jets at all orders in the space of pre-corner fields. In the end, however, such jets are all quotiented out to obtain the maximality property. Ww will see in the next chapters that, for example in the case of PC gravity, it is more convenient to provide an initial description which takes the transversal jets into consideration.  
    \end{remark}
    \begin{remark}
        It is also interesting to notice that the polarization found in this example, obtained from studying the classical corner structure of YM theory, is not the unique choice in the context of the BF$^2$V formalism. Indeed, it was shown in \cite{CC2022cor} that different choices of polarizations lead to the arising of non-trivial $P_\infty$ structures, while preserving the underlying BF$^2$V data. 
    \end{remark}
\subsubsection{$BF$ theory}\label{example: BF theory}
With the same assumptions as in the previous example, the space of fields is the same as in first-order Yang--Mills theory, i.e. a connection 1-form $A$, modeled over the space $\Omega^1(M,\mathfrak{g})$ and a Lie algebra-valued $(d-1)$--form $B\in\Omega^{d-1}(M,\mathfrak{g})$. We denote such space by
$F_M^{{BF}}\coloneqq \Omega^1(M,\mathfrak{g})\oplus \Omega^{d-1}(M,\mathfrak{g})$. The action reads
    \[
    S_M^{{BF}}=\int_M B F_A + \Lambda P(B),
    \]
where $\Lambda$ is reminiscent of the cosmological constant in gravity and $P$ is an invariant polynomial of degree $k$ such that $k(d-1)=2$.\footnote{$P$ is trivial for any dimension where $d\neq 2,3$.} The variation of the action gives
    \[
    \delta S_M^{\mathrm{BF}}=\int_M \delta B (F_A + \Lambda P'(B)) + \delta A \mathrm{d}_A B - \int_\Sigma B \delta A.
    \]
Again, the space of boundary fields is given by $F^{{BF}}_\Sigma=\Omega^1(\Sigma,\mathfrak{g})\oplus\Omega^{d-1}(\Sigma,\mathfrak{g})$, with symplectic form
    \begin{equation*}
        \varpi_\Sigma^\mathrm{BF}=\int_\Sigma \delta A\delta B.
    \end{equation*}
    
For the sake of convenience, we now specify the analysis of $BF$ theory to the case $D=4$. In this setting, there are two constraints on $F^{{BF}}_\Sigma$ of the form\footnote{Here we choose $P(B)=\frac{1}{2}B^2$}
    \[
    L_c=\int_\Sigma c \mathrm{d}_A B, \qquad \mathrm{and} \qquad M_\tau =\int_\Sigma \tau(F_A+\Lambda B),
    \]
where $c\in\Omega^0[1](\Sigma,\mathfrak{g})$ and $\tau\in\Omega^1[1](\Sigma,\mathfrak{g})$ will be indentified with the ghost fields. As expected, such constraints  are in involution, which makes the phase space $C_\Sigma^{BF}$ a coisotropic submanifold. Specifically, one has
    \begin{align*}
        \mathbb L_c =\int_\Sigma \mathrm{d}_A c\frac{\delta}{\delta A} + [c,B]\frac{\delta}{\delta B}, \qquad \qquad \mathbb M_\tau=\int_\Sigma \Lambda \tau \frac{\delta}{\delta A}+ \mathrm{d}_A \tau \frac{\delta}{\delta B},
    \end{align*}
    and
    \begin{align*}
        & \{L_c,L_c\}=\frac{1}{2}L_{[c,c]} & \{L_c,M_\tau\}=-L_{[c,\tau]} &&\{M_\tau,M_\tau\}=0.
    \end{align*}

When corner terms are not neglected, we obtain
    \begin{align*}
        \delta M_\tau = \iota_{\mathbb M_\tau}\varpi ^{BF}_\Sigma - \intc  \tau \delta A, \qquad \mathrm{and}\qquad \delta L_c = \iota_{\mathbb L_c} \varpi^{BF}_\Sigma - \intc c \delta B.
    \end{align*}
Defining $\mathcal{J}_c \coloneqq \intc c \delta B$ and $\mathcal{V}_\tau\coloneqq \tau\delta A$ as one-forms on the space of corner fields $\tilde F^{BF}_\Gamma\coloneqq \Omega^1(\Gamma,\mathfrak{g})\oplus \Omega^2(\Gamma,\mathfrak{g}) \ni (A,B)$, we can furthermore push the Hamiltonian vector fields forward to $F^{BF}_\Sigma$ along the projection $\tilde \pi^\Gamma \colon F^{BF}_\Sigma \to\tilde F^{BF}_\Gamma$, defining 
    \[
    \mathbb J_c \coloneqq \pi^\Gamma_* \mathbb L_c, \qquad \mathrm{and }\qquad \mathbb V_\tau \coloneqq \pi^\Gamma_* \mathbb M_\tau.
    \]

It follows that the distribution $\tilde D^{BF}\coloneqq \mathrm{span}_{\mathcal{C^1}(\tilde F^{BF}_\Gamma)}(\mathbb J_c + \mathcal{J}_c, \mathbb V_\tau + \mathcal{V}_\tau)  $ is isotropic and involutive, as 
    \begin{align*}
        &\langle \mathbb J_c + \mathcal{J}_c, \mathbb J_{c'} + \mathcal{J}_{c'}\rangle =0, & \langle \mathbb J_c + \mathcal{J}_c, \mathbb V_\tau + \mathcal{V}_{\tau}\rangle =0, && \langle \mathbb V_\tau + \mathcal{V}_\tau, \mathbb V_{\tau'} + \mathcal{V}_{\tau'}\rangle =0\\
        &[ \mathbb J_c + \mathcal{J}_c, \mathbb J_{c} + \mathcal{J}_{c}] =\frac{1}{2}\left( \mathbb J_{[c,c]} + \mathcal{J}{[c,c]}\right), & [ \mathbb J_c + \mathcal{J}_c, \mathbb V_{\tau} + \mathcal{V}_{\tau}] =- \mathbb J_{[c,\tau]} - \mathcal{J}_{[c,\tau]}, &&[ \mathbb V_\tau + \mathcal{V}_\tau, \mathbb V_\tau + \mathcal{V}_\tau] =0.
    \end{align*}
Furthermore, $\tilde D^{BF}$ is automatically the graph of the following Poisson structure
    \begin{align*}
        \pi_2^{BF}=\intc \frac{1}{2} B\left[\frac{\delta}{\delta B},\frac{\delta}{\delta B}\right] + \frac{\delta}{\delta A} \mathrm{d}_A \frac{\delta}{\delta B} + \Lambda\frac{\delta}{\delta A}\frac{\delta}{\delta A}.
    \end{align*}

\begin{remark}
It is interesting to notice that in the case of $BF$ one can in principle obtain a Dirac structure without the need of restricting to the submanifold $F_\Gamma^{BF}$ defined as those corner field configurations that extend to boundary conditions in a neighborhood of $\Gamma$ (or, in other words, the restriction of the Cauchy data to the corners). Furthermore, it is not even necessary to reduce our space of (pre-)corner fields, as it already admits the existence of a Poisson structure. This is confirmed by the fact that the kernel of the 1-forms over $\tilde{F}_\Gamma^{BF}$ generated by $\mathcal{J}_c$ and $\mathcal{V}_\tau$ is trivial, indeed
    \[
    \Ker{\mathcal{J}_c+\mathcal{V}_\tau} = \left\{\mathbb X= \intc \mathbb X_A \frac{\delta}{\delta A} + \mathbb X_B \frac{\delta}{\delta B} \ \bigg| \ \iota_{\mathbb X}\mathcal{J}_c=\iota_{\mathbb X}\mathcal{V}_\tau=0\right \}=\{0\}.
    \]
\end{remark}
However,  we have not yet taken into account the presence of a ghost for ghost field. Indeed, considering the particular case where $\tau=\mathrm{d}_A \phi$, we have, on the boundary
    \[
    M_{\mathrm{d}_A\phi} = \int_\Sigma \mathrm{d}_A\phi (F_A +\Lambda B) =\int_\Sigma \Lambda\phi \mathrm{d}_A B = L_{\Lambda \phi},
    \]
which tells us that the gauge symmetry generated by the ghost $\tau$ is reducible. To do so, one needs to introduce another ghost, in fact a ghost for ghost field $\phi\in\Omega^0[2](\Sigma,\mathfrak{g}) $, and introduce gauge transformations for $c$ and $\tau$ as follows
    \begin{align*}
        c\longmapsto c -\Lambda \phi, \qquad \qquad \tau \longmapsto \tau + \mathrm{d}_A \phi.
    \end{align*}
Furthermore, one needs to assume that $\phi$ transforms as $\phi \mapsto \phi + [c,\phi].$\footnote{
In order to implement this in the BFV formalism, we define $\mathcal{F}^{BF}_\Sigma=F^{BF}_\Sigma \times T^*(\Omega^0[1](\Sigma,\mathfrak{g})\oplus\Omega^1[1](\Sigma,\mathfrak{g})\oplus\Omega^0[2](\Sigma,\mathfrak{g})) $, with symplectic form and action given by 
\begin{align*}
    \varpi_{\Sigma}^{BF}=\int_\Sigma& \delta A \delta B + \delta A^\dag \delta c + \delta B^\dag \delta \tau + \delta \tau^\dag \delta\phi;\\
    \mathcal{S}^{BF}_\Sigma=\int_\Sigma & \frac{1}{2}A^\dag [c,c] + c\mathrm{d}_AB + \tau(F_A + \Lambda B) + \phi ([c,\tau^\dag]- \mathrm{d}_A B^\dag ) +  \Lambda(B\tau - A^\dag \phi) .
\end{align*}
}
This is where classical generalised geometry fails to include such contributions. Specifically, at the corner level, one needs to introduce a $P_\infty$ structure. Such case is best analysed in the BF$^2$V formalism. The space of BF$^2$V fields is $\mathcal{F}^{BF}_\Gamma=\tilde F^{BF}_\Gamma\times (\Omega^0[1](\Gamma,\mathfrak{g}) \oplus \Omega^0[2](\Gamma,\mathfrak{g}) \oplus T^*(\Omega^1[1](\Gamma,\mathfrak{g}))) $, which can be regrouped into superfields
    \begin{align*}
        \tilde{A}\coloneqq c+  A +  B^\dag, \qquad \mathrm{and }\qquad \tilde{B}\coloneqq \phi + \tau + B,
    \end{align*}
with symplectic form and action
    \begin{align*}
         \varpi_{\Gamma}^{BF}=\intc &\delta\tilde{A}\delta\tilde{B}\\
         \mathcal{S}_\Gamma^{BF}=\intc& \tilde{B}F_{\tilde{A}} +\frac{1}{2}\Lambda \tilde{B}\tilde{B}\\
         =\intc & \frac{1}{2}B[c,c]+ \tau \mathrm{d}_A c + \phi (F_A + [c,B^\dag])+\Lambda\left(\frac{1}{2}\tau^2 + B\phi\right).
    \end{align*}
\begin{remark}
    Notice that, if we set $B^\dag=0$ and $\phi=0$, we are in the case above, where the action $\mathcal{S}_\Gamma^{BF}$ is equivalent to the Poisson structure $\pi_2^{BF}$.
\end{remark}    
One way to interpret the full $\mathcal{S}_\Gamma^{BF}$ is as a functional over the tangent space (shifted by one) to the space $\mathcal{N}\coloneqq \tilde F^{BF}_\Gamma \oplus \Omega^2[-1](\Gamma,\mathfrak{g})\ni (A,B,B^\dag) $.\footnote{In particular, the conjugate fields to $A$, $B$ and $B^\dag$ are identified with their dual, respectively given by $\tau$, $c$ and $\phi$.}  In this case, $\mathcal{S}_\Gamma^{BF}$ is not of homogeneous degree 2 in the fiber of $T[1]\mathcal{N}$, but there is a contribution which is only linear in the fiber, hence defining a vector field on $\mathcal{N}$. Indeed, $\mathcal{S}_\Gamma^{BF}$ corresponds to a $P_\infty$ structure of the kind $\pi=\pi_1+\pi_2$, with 
    \begin{align*}
        \pi_1=& \intc (F_A+\Lambda B)\frac{\delta}{\delta B^\dag},\\
        \pi_2=&\pi_2^{BF}+\intc B^\dag\left[\frac{\delta}{\delta B^\dag},\frac{\delta}{\delta B}\right].
    \end{align*}
The CME guarantees that $\pi_1$ is cohomological, $\pi_2$ is a Poisson structure which is invariant under $\pi_1$. Furthermore, the degree-0 cohomology of $\pi_1$ amounts to the quotient of the classical functionals on $\tilde F^{BF}_\Gamma$ by the ideal generated by the constraint $F_A+\Lambda B$. In other words, it imposes the restriction to the Poisson submanifold $F_{\Gamma}^{BF}\coloneqq \{ (A,B)\in \tilde F^{BF}_\Gamma \ | \ F_A+\Lambda B=0 \}$. 

\begin{remark}
    Notice that, considering the space $F_{\Gamma}^{BF}$ exactly amounts to restricting ourselves to the field configurations on the corner which extend to Cauchy data on a neighborhood of the corner. Indeed, in the absence of the transversal jets, only the constraint $F_A+\Lambda B=0$ survives, since $\mathrm{d}_AB=0$ is a three-form, which trivially vanishes on the corners. What we have described is then known as a Dirac structure with support, whenever one considers the full homological description provided by the $P_\infty$ structure $\pi=\pi_1+\pi_2$. 
\end{remark}

In the remainder of the paper we will see how, for Palatini--Cartan gravity, both the reduction step and the restriction to a constrained set are necessary to fully describe the Dirac structure at the corners.

\section{Review of PC gravity on manifolds with boundary}\label{sec:pc_boundary}
In this section we provide a quick summary of the results in \cite{CCS2020}. 

\subsection{The Palatini--Cartan theory}
General Relativity is traditionally formulated (\cite{einstein1916}) using a (pseudo)Riemannian metric and its associated Levi-Civita connection. This leads to Einstein’s field equations, expressed in terms of the Ricci tensor $ R_{\mu\nu} $, which depends on the Christoffel symbols $ \Gamma_{\mu\nu}^\lambda $ and, consequently, on the metric $ g_{\mu\nu} $. The Ricci scalar $ R $ and the metric $ g_{\mu\nu} $ further appear in the Einstein--Hilbert action functional:  
\begin{equation}\label{eq:EHaction}
S=\int R\sqrt{-g}\,\de^4x.
\end{equation}  

Alternatively, one can consider the same action \cref{eq:EHaction}, but treat the connection $ \nabla $ as an independent variable, so that $ S = S[g, \nabla] $. This approach allows $ \Gamma $ to represent the coefficients of a generic affine connection, which may, in general, have torsion. If we impose metric compatibility, the variational principle ensures the torsion-free condition (and vice versa). This formulation is known as the Palatini formalism (\cite{Palatini1919}).

By means of the coframe formalism we can cast the theory into a different setting. Namely, we are able to express the affine and the metric structure of the theory through a coframe field (a.k.a.\ a vielbein) $e$ and a principal connection $\omega$, letting general relativity resemble, from a geometrical perspective, an SU$(N)$ gauge theory. This approach is known as the coframe or Palatini--Cartan formulation.

We start by considering a principal $SO(D-1,1)$ bundle $P$ over a $D$-dimensional manifold $M$. Furthermore, we fix a $D$-dimensional vector space $V$ with a reference metric $\eta$ with Lorentzian signature. Then, considering $\rho$ to be the fundamental representation of the Lorentz group, we define the associated vector bundle
    \[
    \mathcal{V}\coloneqq P \times_\rho V,
    \]
which takes the name of ``Minkowski bundle'' (a.k.a.\ ``fake tangent bundle''). In order to define a metric on $M$, which will ultimately be the dynamical object of the theory, we resort to coframes (or vielbein), defined to be non-degenerate linear isomorphisms of the kind
    \[
    e \ \colon \ TM\longrightarrow \mathcal{V}. 
    \]
One can then define a metric $g$ on $M$ as $g=e^*(\eta)$ or, equivalently, as $g_{\mu\nu}=e_\mu^a e_\nu^b \eta_{ab} $. The other fundamental dynamical field is a principal connection $\omega\in\mathcal{A}_M$, defined on $M$. In general, given a reference connection $\omega_0$, we have
    \[
    \omega-\omega_0\in  \Omega^1(M,\mathfrak{so}(D-1,1)) .
    \]
Since the Lie algebra of $SO(D-1,1)$ can be identified with $D\times D$ antisymmetric matrices, we model the space of connections by $\Omega^1(M,\wedge^2\mathcal{V})$.

From now on, we will adopt the following notation
    \[
    \Omega^{i,j}\coloneqq \Omega^i(M,\wedge^j \mathcal{V}).
    \]
Then we are ready to give the definition of PC theory. 
\begin{definition}\label{PC_def}
    The classical Palatini--Cartan theory is the assignment of the pair $(F_M,  S_M)$ to every pseudo Riemannian D-dimensional manifold admitting a Lorentzian structure,\footnote{Note that any particular choice of the Lorentzian structure on $V$ is immaterial, since a change in $V$ would just isomorphically reflect to the space of fields without changing $S_{M}$.} with space of fields given by
    \begin{align*}
    F_M=\Omega^1_{\mathrm{n.d.}}(M,\mathcal{V})\oplus \mathcal A_M\ni (e,\omega),
    \end{align*}
    and action functional (omitting the wedge product symbol)
    \begin{align}\label{eq:actionPCT}
    S_M=\int_{M}\frac1{(D-2)!}e^{D-2} F_\omega+\frac1{D!}\Lambda e^D,
    \end{align}
    where $\Lambda \in\mathbb R$ represents the cosmological constant and $e^k\coloneqq e\wedge \cdots \wedge e$. 
\end{definition}
\begin{remark}

    The integrand in \cref{eq:actionPCT} (i.e., the Lagrangian) is a form in $\Omega_M^{D,D}=\Omega^D(M,\wedge^D \mathcal{V})$, which can canonically be identified with a density \cite{CCS2020}. For simplicity, in the following we assume $M$ to be orientable. A choice of orientation is reflected by a
  ``trace’’ map $\Tr\colon\Lambda^D\to\mathbb{R}$, which applied to the Lagrangian density yields a top form. We
may then view the action as the integral of this top form on $M$ with the given orientation.
\end{remark}

The Euler--Lagrange equations coming from the action principle $\delta  S_M=0$ are, respectively for the variations in $e$ and $\omega$ 
\begin{align}
\frac1{(N-3)!}e^{D-3}F_\omega-\frac1{(D-1)!}\Lambda e^{D-1}&=0\label{eq:PCe}\\[7pt]
e^{D-3}\de_\omega e&= 0,\label{eq:PComega}
\end{align}
which, in $D=4$, reduce to
\begin{align}
e F_\omega-\frac1{3!}\Lambda e^{3}&=0\label{eq:PCe4}\\[7pt]
e \de_\omega e&= 0.\label{eq:PComega4}
\end{align}
By injectivity of the map $e\wedge\cdot$ on $\Omega^{2,1}$ (see \cite{CCS2020}), \cref{eq:PComega4} is equivalent to\footnote{This generalizes to a generic $D$.}
\begin{equation}\label{e:tf}
\de_\omega e = 0,
\end{equation}
which is the torsion-free condition. Therefore, this is the equation that identifies the Levi-Civita connection for the metric induced by $e$.
It is a matter of computations to see that the Palatini--Cartan action is classically equivalent (on-shell) to the Einstein--Hilbert one. 

\subsection{Boundary structure of PC gravity}\label{sec: bdry struct of PC}
In this section, we will assume that our stratified structure, with respect of \cref{def:strat_man}, has only two strata, i.e. we will work on a manifold with boundary. Moreover, from now on, we will specify to the case $D=4$. The boundary structure is then obtained by applying the KT construction \cite{KT1979} outlined in the previous section. See \cite{CCS2020, C23}
\begin{remark}
    Given $i\colon:\Sigma\hookrightarrow M$ the boundary inclusion. We will use the following definition for bundle valued differential forms on the boundary
    \begin{align*}
        \Omega^{i,j}_\Sigma\coloneqq\Omega^i(\Sigma,\textstyle{\bigwedge^j} i^*V).
    \end{align*}
\end{remark}

We notice that the integration by parts in the variation of the action $S_M$ gives rise to a boundary term
\begin{align}\label{noetherform}
    \tilde\alpha_\Sigma=\intb\frac12e^{2}\delta \omega,
\end{align}
which we call the Noether $1$-form. Such object is defined on the space of pre-boundary fields $\tilde{ F}_\Sigma$, given by the restriction of the fields on the boundary, their transversal components and their transversal jets (to all orders).\footnote{This description differs from the description given in \cite{CCS2020}, where only the restriction of the fields to the boundary was taken into consideration.} Specifically, considering an infinitesimal cylindrical neighborhood $M_\epsilon\coloneqq [0,\epsilon]\times \Sigma$ of $\Sigma$, we have that a field $\Phi\in\Omega^{\bullet,\bullet}_{M_\epsilon}$ splits as
    \[
    \Phi=\phi + \phi_n dx^n,
    \]
where $x^n$ is the coordinate along the interval $[0,\epsilon]$, and 
$$\phi\in\mathcal{C^1}([0,\epsilon])\otimes \Omega^{\bullet,\bullet}_\Sigma, \qquad \qquad \phi_n dx^n\in\Omega^1([0,\epsilon])\otimes \Omega^{\bullet-1,\bullet}_\Sigma.$$ 
The space of pre-boundary fields is then obtained by considering the restriction of the fields, their transversal components and their jets at $\Sigma\times\{0\}.$ In practice, we have
    \begin{equation*}
        \tilde{F}_\Sigma\coloneqq (\Omega^{1,1}_{\Sigma,\mathrm{n.d.}}\oplus \Omega^{0,1}_\Sigma ) \oplus (\mathcal{A}_\Sigma \oplus \Omega^{0,2}_\Sigma) \oplus \mathfrak{J}_n(\Omega^{1,1}_{\Sigma}\oplus \mathcal{A}_\Sigma)\ni (e,e_n,\omega,\omega_n,\partial^k_n e, \partial^k_n \omega),
    \end{equation*}
where $\mathfrak{J}_n(\Omega^{1,1}_{\Sigma}\oplus \mathcal{A}_\Sigma)$ denotes the space of transversal jets (to all orders) to $e$ and $\omega$.
\begin{remark}
    Notice that we denote fields on the boundary with the exact same notation of the ones in the bulk. Furthermore, we assume the coframe field $e$ restricted to the boundary to give rise to a non-degenerate boundary metric $g^\Sigma$, which is therefore either space-like or time-like. 
\end{remark}

\begin{lemma}\label{lem_degsympform}
    The $2$-form on ${\widetilde{F}}_\Sigma$ defined via
    \begin{align}\label{degsympform}
        \tilde\varpi_\Sigma=\delta\tilde \alpha_\Sigma=\intb e\delta e \delta \omega
    \end{align}
is closed and degenerate. Furthermore, $\Ker{\tilde{\varpi}_\Sigma}$ defines an involutive regular distribution. 
\end{lemma}
\begin{proof}
    The two-form is closed because it is exact. The kernel of the pre-symplectic form is not empty, indeed, it splits as
        \[
        \Ker{\tilde\varpi_\Sigma}=\mathfrak{A}\oplus\mathfrak{X}\left( \Omega^{0,1}_\Sigma \oplus \Omega^{0,2}_\Sigma \oplus\mathfrak{J}_n(\Omega^{1,1}_{\Sigma}\oplus \mathcal{A}_\Sigma)\right),
        \]
    since all the vector fields $\mathbb Y$ tangent to the transversal fields and jets are automatically satisfy $\iota_{\mathbb Y}\tilde{\varpi_\Sigma}=0 $.  The distribution $\mathfrak{A}\subset \mathfrak{X}\left(\Omega^{1,1}_{\Sigma,\mathrm{n.d.}}\oplus\mathcal{A}_\Sigma\right) $ is such that, given
    \begin{align*}
        \mathbb X=\int_\Sigma \mathbb X_e\frac{\delta}{\delta e}+\mathbb X_\omega\frac{\delta}{\delta \omega}\in \mathfrak{X}\left(\Omega^{1,1}_{\Sigma,\mathrm{n.d.}}\oplus\mathcal{A}_\Sigma\right) ,
    \end{align*}
    we get
    \begin{align*}
       \iota_{\mathbb X} \tilde\varpi_\Sigma=\intb e(\mathbb X_\omega\delta e + \mathbb X_e\delta \omega)=0.
    \end{align*}
    Since we already mentioned that $e\wedge\cdot$ is injective on $\Omega^{1,1}_\Sigma$ but not on $\Omega^{1,2}_\Sigma$, it follows that $\mathbb X_e=0$ and $e\mathbb X_\omega=0$, namely the flow of $\mathbb X$ acts on $\omega$ as $\omega\mapsto \omega+v$ with $v\in\Omega^{1,2}_\Sigma$ such that $ev=0$. The distribution is clearly regular. It is involutive because it is the kernel of a closed 2-form.
\end{proof}

Before performing the quotient with respect to the above distribution, we give the definition of some maps. 
\begin{definition}\label{def_maps}
Considering $e\in\Omega_\Sigma^{1,1}$, we define the following:
    \begin{align*}
                W_k^{\Sigma, (i,j)}\colon \Omega_\Sigma^{i,j}  & \longrightarrow \Omega_\Sigma^{i+k,j+k} \\
                \alpha  & \longmapsto    e^k\wedge \alpha \nonumber
    \end{align*}  
    \begin{align*}
                    \varrho^{(i,j)} \colon \Omega_{\Sigma}^{i,j}  & \longrightarrow \Omega_{\Sigma}^{i+1,j-1} \\
                    \alpha & \longmapsto [e,\alpha] \nonumber
    \end{align*}

\end{definition}

Then the geometric phase space for the Palatini--Cartan theory is the symplectic manifold ($ F_{\Sigma}, \varpi_\Sigma)$ obtained by a symplectic reduction of the space of pre-boundary fields ${\widetilde{F}}_\Sigma$.
The symplectic form $\varpi$ is the unique $2$-form on ${{F}}_\Sigma$ such that $p^*\varpi_\Sigma = \tilde\varpi_\Sigma$, where $p\colon {\widetilde F}_\Sigma\to{F}_\Sigma$ is the canonical projection, and it is given by
    \begin{align}\label{sympl_form_geo_grav}
        \varpi_\Sigma=\int_{\Sigma}e\delta e \delta [\omega],
    \end{align}
    where
\begin{align}\label{eqcl}
    \omega^\prime\sim\omega\quad\Longleftrightarrow\quad\omega^\prime-\omega \in\mathrm{Ker}W_{1}^{\Sigma, (1,2)}.
\end{align}
We refer to this equivalence class as $\mathcal{A}_{\mathrm{red}}(\Sigma)$.

    As mentioned in \cref{rem: tang ev eq}, we interpret the field equations containing non-transversal derivatives to the boundary as local functionals (constraints) on the space of pre-boundary fields of the theory.

    \noindent Such constraints are given by
    \begin{align}
        \de_\omega e = 0, \qquad \qquad 
        eF_\omega + \frac{1}{3!}\Lambda e^3=0,
    \end{align}
    while the transversal equations are 
    \begin{align}
        &\de_\omega e_n + \de_{\omega_n }e  =0  \label{eq: transv torsion}\\
        & e_n F_\omega - e F_{\omega_n}=0 \label{eq: transversal Einstein},
    \end{align}
    which are evolution equations.\footnote{We will see in section \ref{sec: Bianchi} how they can be used to fix the representatives of the equivalence classes of the transversal jets of order one.}
    
    We notice immediately that the constraints are not well defined on $F_\Sigma$, as the torsion equation is not invariant under the action of $\Ker{\tilde{\varpi}_\Sigma}$.\footnote{Specifically, $d_\omega e=0$ is not invariant under $\omega\mapsto \omega+v$.}  However, given that $\Ker{\tilde{\varpi}_\Sigma}=\mathfrak{A}\oplus \mathfrak{E}$, with the identification $\mathfrak{E}\coloneqq\mathfrak{X}\left( \Omega^{0,1}_\Sigma \oplus \Omega^{0,2}_\Sigma \oplus\mathfrak{J}_n(\Omega^{1,1}_{\Sigma}\oplus \mathcal{A}_\Sigma)\right)$, we can perform a partial reduction $\tilde{F}_\Sigma/\mathfrak{E}$.

    \begin{lemma}
        $\tilde{F}_\Sigma/\mathfrak{E}$ is diffeomorphic to the space 
        \[
        \tilde{F}_\Sigma^n\coloneqq \{ ({e,\omega,\epsilon_n})\in \Omega^{1,1}_{\Sigma,\mathrm{n.d.}}\oplus\mathcal{A}_\Sigma\oplus\Omega^{0,1}_\Sigma \ | \ (e,\en)  \text{ forms a basis of $\mathcal{V}$ and $\delta \en=0$}  \}\simeq \Omega^{1,1}_{\Sigma,\mathrm{n.d.}}\oplus\mathcal{A}_\Sigma
        \]
        
    \end{lemma}
    \begin{proof}
        Clearly, quotienting by $\mathfrak{E}$ removes all the transversal fields and jets at all orders. In particular, since the only expressions containing $\omega_n, \partial^k_n e$ and $\partial^k_n\omega$ are evolution equations ---which we do not take into consideration when defining the phase space--- we can just discard them. 
        
        The field $e_n$ is different, since we require $e+e_n dx^n$ to define a coframe field on $M_{\Sigma_\epsilon}$, i.e. to be such that $(e,e_n)$ is a (local) basis for $\mathcal{V}$. 
        In the quotient space, we can then choose a representative $\epsilon_n\in [e_n]$ to be a nowhere vanishing element of $\Omega^{0,1}_\Sigma$ such that the above requirement is satisfied.

        In principle, such representative depends on $e$. However, for any $e\in\Omega^{1,1}_{\Sigma,\mathrm{n.d.}}$, there exists an open $\mathcal{U}\subset \Omega^{1,1}_{\Sigma,\mathrm{n.d.}}$ such that $\delta\en=0$ on $\mathcal{U}$. The situation is further simplified if $\Sigma$ is assumed to be space-like, as the choice of any time-like $\en$ defines a globally (with respect to the full $\Omega^{1,1}_{\Sigma,\mathrm{n.d.}}$) constant section. 
    \end{proof}
    \begin{remark}
        The expression $\frac{1}{3!}\epsilon_n e^3$ is a nowhere vanishing element of $\Omega^{3,4}_\Sigma$ and can be identified with a volume form on $\Sigma$, after using the identification $\wedge^4 \mathcal{V}\simeq \mathcal{C^1}(\Sigma)$.
    \end{remark}

    The partial reduction has not solved the problem of the torsion constraint $\de_\omega e =0$ not being invariant under the action of the residual vector fields in the kernel of $\tilde{\varpi}_\Sigma$, which is identified with the induced distribution on $\tilde{F}^n_\Sigma$ denoted (by abuse of notation) by $\mathfrak{U}\subset \mathfrak{X}(\tilde{F}^n_\Sigma)$. 
    However notice that, due to the injectivity of the map $W_e^{1,2}\colon\alpha\to e\wedge\alpha$,  $\de_\omega e=0$ is equivalent to $e \de_\omega e=0$  in the bulk. On the contrary, when moving to the boundary such injectivity property is lost, and therefore one needs to complement $e \de_\omega e=0$ ---which is invariant under the action of $\Ker{\tilde\varpi}_\Sigma$---\footnote{The simple computation is given by $e\de_{\omega + v} e = e \de_\omega e + e[v,e] = e \de_\omega e + [ev,e] = e \de_\omega e $, using $[e,e]=0$.} with an appropriate extra constraint, ensuring that $\de_\omega e=0$ is satisfied on the boundary.

The following lemma allows to rewrite $\de_\omega e=0$ as an invariant constraint and an extra equation, which we call structural constraint for reasons that will soon become apparent.

\begin{lemma}\label{strconstr_free_deg}
Assuming the boundary metric is non-degenerate, given $\alpha\in\Omega_\Sigma^{2,1}$, we have that
    \[
    \alpha=0 \qquad \Leftrightarrow \qquad \begin{cases}
        e \alpha =0\\[6pt]
        \en \alpha\in \Ima W_{1}^{\Sigma,(1,1)}\\[6pt]
    \end{cases}
    \]
which implies, setting $\alpha = \de_\omega e$, that
    \[
    \de_\omega e=0 \qquad \Leftrightarrow \qquad \begin{cases}
        e \de_\omega e =0\\[6pt]
        \en \de_\omega e = e\sigma\\[6pt]
    \end{cases}
    \]
    for some $\sigma\in\Omega^{1,1}_\Sigma$. 
\end{lemma}\
\begin{remark}\label{rem: brdy strct constr}
    A quick computation shows that the transversal component of $e\de_\omega e=0$ is given by  $$ e_n \de_\omega e = e(\de_\omega e_n + \de_{\omega_n }e)= e(\de_\omega e_n + \partial_n e + [\omega_n,e]),$$ which essentially reproduces the extra constraint, after setting $\sigma=\sigma(e,e_n,\omega,\omega_n, \partial_n e)$. Working on the quotient space $\tilde{F}^n_\Sigma$, the dependence on $\omega_n $ and $\partial_n e$ are lost and $e_n $ is replaced by $\en$.
\end{remark}

\begin{remark}
    If the boundary metric is degenerate, one needs an extra constrain to be imposed in order to obtain $\de_\omega e=0$ on the boundary. We refer to \cite{CCT21} for more details on this specific instance. In the following, we will assume $e\in\Omega^{1,1}_\Sigma$ to define a  non-degenerate boundary metric whenever unspecified.  
\end{remark}

The following theorem allows to make sense of the structural constraint as a way to fix a representative of $[\omega]\in\mathcal{A}_{\mathrm{red}}$. In other words, ensuring the equivalence of $\de_\omega e=0$ and $e \de_\omega e=0$ on the boundary, is enough to determine uniquely the representatives of the equivalence classes defined in \cref{eqcl}.

\begin{theorem}[\cite{CCS2020}]\label{uniquerep_free}
Given $\omega'\in\mathcal{A}_\Sigma$, there exists a unique decomposition 
    \[
    \omega'= \omega + v
    \]
    such that
    \begin{align}\label{eq: str constr bdry}
    ev=0, \quad \mathrm{and} \quad \epsilon_n \de_\omega e = e\sigma    
    \end{align}
    
    Furthermore, the space $$F_\Sigma^{n}\coloneqq\{ (e,\omega)\in\Omega^{1,1}_{\en}\oplus \Omega^{1,2}_\Sigma \ | \ \epsilon_n \de_\omega e = e\sigma \}\subset \tilde F_\Sigma$$ is diffeomorphic to $\tilde F^n_\Sigma/ \mathfrak{U}\simeq \tilde{F}_\Sigma/(\Ker{\tilde{\varpi}_\Sigma})\simeq F_\Sigma$. We use such diffeomorphism to endow $F_\Sigma^{n}$ with a symplectic form.
\end{theorem}

We now display the constraints of the theory, which we define as local functionals on $F_\Sigma$ by means of Lagrange multiplier. The vanishing locus $C_\Sigma$ of these functionals is then the phase space of the theory.
\begin{definition}\label{constraints_gravity}
Let $c\in\Omega_\Sigma^{0,2}[1]$, $\xi\in\mathfrak{X}(\Sigma)[1]$ and $\lambda\in C^\infty(\Sigma)[1]$. Then, we define the following functionals
    \begin{align*}
		L_c & = \int_\Sigma ce\de_\omega e \\[3pt]
		P_\xi & = \int_\Sigma \frac{1}{2}\iota_\xi(e^2)F_\omega+\iota_\xi(\omega-\omega_0)e\de_\omega e \\[3pt]
		H_\lambda & = \int_\Sigma \lambda \epsilon_n\Big(eF_\omega+\frac{\Lambda}{3!}e^3\Big)
	\end{align*}
We refer to these as the constraints of the Palatini--Cartan theory.
\end{definition}
\begin{remark}
 We have taken the Lagrange multipliers\footnote{In principle, one would need just one Lagrange multiplier $\mu\in\Omega^{0,1}_\Sigma[1] $ to impose the constraint $e F_\omega=0$; however, we can split $\mu=\iota_\xi e + \lambda \epsilon_n$. Indeed the vector field $\xi$ and the function $\lambda$ will respectively represent the ghost fields associated to diffeomorphisms tangent to $\Sigma$ and normal to $\Sigma$.}
 \footnote{
Also notice that the constraint $P_\xi$ has been complemented with the constraint $L_{\iota_\xi(\omega-\omega_0)}  $, which does not change the vanishing locus of the constraint set but will simplify the computation of the Hamltonian vector field associated to it.} as odd (shifted by one) fields as they will represent ghosts within the BFV formalism.
\end{remark}

We are now able to determine the algebra of the constraints of the theory.
\begin{theorem}[\cite{CCS2020}]\label{thm:Brackets_constraints}
The phase space $C_\Sigma$ is a coisotropic submanifold of $F_\Sigma$.
    In particular, the Poisson brackets of the constraints of \cref{constraints_gravity} read
    \begin{align*}
    &\begin{aligned}
        	&\{L_c, L_c\}  = - \frac{1}{2} L_{[c,c]} & \qquad\qquad
       	& \{P_{\xi}, P_{\xi}\}  =  \frac{1}{2}P_{[\xi, \xi]}- \frac{1}{2}L_{\iota_{\xi}\iota_{\xi}F_{\omega_0}} \\
       	& \{L_c, P_{\xi}\}  =  L_{\mathrm{L}_{\xi}^{\omega_0}c} &
       	& \{H_{\lambda},H_{\lambda}\} = 0 \\
        \end{aligned}\\
	&\{L_c,  H_{\lambda}\}
        = - P_{X^{(i)}} + L_{X^{(i)}(\omega - \omega_0)_i} - H_{X^{(n)}}\\        
        &\{P_{\xi},H_{\lambda}\} =  P_{Y^{(i)}} -L_{ Y^{(i)} (\omega - \omega_0)_i} +H_{ Y^{(n)}} 
    \end{align*}
    with $X=[c,\lambda \epsilon_n]$ and $Y=\mathrm{L}_\xi^{\omega_0}(\lambda \epsilon_n)$ and where the superscripts $(i)$ and $(n)$ describe their components with respect to $\{e_i,\epsilon_n\}$.
\end{theorem}

For the sake of completeness, we provide the expression of the Hamiltonian vector fields of the constraints:
\begin{align*}
&\mathbb{L}_e = [c,e] &  \mathbb{L}_\omega = \de_\omega c + \mathbb{V}_L\\
&\mathbb{P}_e = - \mathrm{L}_{\xi}^{\omega_0} e & \mathbb{P}_\omega = - \mathrm{L}_{\xi}^{\omega_0} (\omega-\omega_0) - \iota_ {\xi}F_{\omega_0} + \mathbb{V}_P\\
&\mathbb{H}_e= \de_\omega(\lambda \epsilon_n) + \lambda \sigma  & e \mathbb{H}_\omega =  \lambda \epsilon_n F_{\omega}+\frac{1}{2}\Lambda  \lambda \epsilon_n e^2.
\end{align*}
where we stress that the component $\mathbb X_\omega$ of the above vector fields is defined up to elements in $\mathrm{ker}(W_1^{\Sigma,(1,2)})$.\footnote{Indeed, one can think of the physical fields in the geometric phase space to be $\{e,e(\omega-\omega_0)\}$, where the new field $\Omega\coloneqq e(\omega-\omega_0)\in\Omega_\Sigma^{2,3} $ only depends on the equivalence class of $\omega-\omega_0\in\Omega^{1,2}_\Sigma/\Ker{W_e^{\Sigma(1,2)}}$.}

\section{PC gravity on manifolds with corners}\label{sec: corner PC}
In the following, for simplicity we will assume $\Lambda=0$, as that will have hardly any impact on the computations. The geometric set-up we are working with is that of a stratified manifold---more precisely, a manifold with corners as defined in \cref{def:strat_man} and specified in \cref{rem:two_strata}. From now on, we will refer to the codimension-$1$ and codimension-$2$ strata as, respectively, the boundary and the corner of the manifold.

We start by considering a neighborhood of the corner $\Sigma_{\epsilon}\coloneqq \Gamma\times[0,\epsilon]$, for an arbitrarily small $\epsilon>0$. Setting $\{x^m\}$ to be the coordinate along $[0,\epsilon]$, any field of the form ${\Phi}\in\Omega^{\bullet,\bullet}(\Sigma_\epsilon)$ is split as
\begin{align*}
    \Phi\mapsto\phi + \phi_m \de x^m,
\end{align*}
with $\phi\in \mathcal{C^1}([0,\epsilon])\otimes \Omega^{\bullet,\bullet}_\Gamma$ and $\phi_m\in \Omega^1([0,\epsilon])\otimes \Omega^{\bullet-1,\bullet}_\Gamma$ .

In this section, given $i\colon \Sigma \to M$ and $j\colon \Gamma\to \Sigma$, we will use the following notation:
    \begin{align*}
        &\Omega^{i,j}_\Gamma\coloneq\Omega^{i}(\Gamma, \bigwedge^j(j^*i^*\mathcal V)), &W_{k}^{\Gamma,(i,j)}\colon\Omega^{i,j}_\Gamma&\longrightarrow\Omega^{i+1,j+1}_\Gamma\\
        & &\alpha&\longmapsto e\wedge\alpha.
    \end{align*}
    
By considering the restriction at $\Gamma\times \{0\}$ of ${F}_{\Sigma_\epsilon} $, we obtain the space of pre-corner fields
    \begin{align*}
        \ \tilde{F}_\Gamma&\coloneqq \Omega_\Gamma^{1,1}\oplus \Omega_{\Gamma}^{0,1} \oplus\mathcal{A}_\Gamma \oplus \Omega_{\Gamma}^{0,2}\oplus \mathfrak{J}_m(\Omega_\Gamma^{1,1} \oplus\mathcal{A}_\Gamma )\ni (e,e_m,\omega,\omega_m,\partial_m^ke,\partial_m^k\omega),
    \end{align*}
where $e\in\Omega^{1,1}_\Gamma$ and $e_m\in\Omega^{0,1}_\Gamma$ are such that $\{e(\partial_1),e(\partial_2),e_m,\epsilon_n\}$ is a basis of $j^*i^*\mathcal V$\, where $\{\partial_1,\partial_2\}$ is a local basis of $T\Gamma$. Notice, however, that $e_m$ is not fixed, conversely to the case of $\epsilon_n$. Moreover, we denoted by $ \mathfrak{J}_m(\Omega_\Gamma^{1,1} \oplus\mathcal{A}_\Gamma )$ the space of transversal jets (to all orders) of $e $ and $\omega$. 

Furthermore, we need impose the constraints which are induced from the ones on the geometric phase $F_{\Sigma_\epsilon}$, i.e.
\begin{equation*}
    \en \de_\omega e = e \sigma, \qquad \text{and}\qquad \en(\de_{\omega_m} e - \de_\omega e_m)= e_m \sigma + e\sigma_m.
\end{equation*}
\begin{remark}
In \cref{sec: reduced corner space} we will see that the transversal jets only play a marginal role in the corner structure, as their value will be fixed by constraints, or set to zero identically. The reason to keep them into consideration is simply to be able to rigorously use the Leibniz rule with respect to the operator $\partial_m$.
\end{remark}
\begin{remark}
    The set $\tilde{F}_\Gamma$ can also be interpreted as the quotient of $\tilde F_{\Sigma_\epsilon}$ by the distribution of vector fields that preserve the values of the fields and their jets at $\Gamma\times\{0\}$, from which we obtain the quotient map
    \[
    \tilde{\pi}_\Gamma\colon {F}_{\Sigma_\epsilon}\longrightarrow \tilde{F}_\Gamma.
    \]
\end{remark}

 \subsection{Pre-Dirac corner structure of PC gravity}\label{sec: predirac GR}

Referring to the discussion of \cref{sec: classical corner}, we have that, in presence of corners, the defining equation of the Hamiltonian vector fields of the constraints $\phi_\alpha$ is modified as follows
    \begin{align}\label{eq: def corner 1-forms}
    \delta \phi_\alpha = \iota_{\mathbb X_\alpha} \varpi_\Sigma - \tilde{\pi}^*_\Gamma(\mathcal{X}_\alpha),    
    \end{align}

where $\phi_\alpha=\{L_c, P_\xi,H_\lambda\}$. Notice also that the ghost $\xi$ splits on the corner as $\xi\mapsto \zeta + \eta \partial_m$, with $\zeta\in\mathfrak{X}(\Gamma)[1]$ and $\eta\in\mathcal{C^1}(\Gamma)[1]$. Applying \cref{eq: def corner 1-forms}, we obtain four (1-shifted) 1-forms on $\tilde{F}_\Gamma$, given by
    \begin{align}
        \nonumber&\mathcal{J}_c = \int_\Gamma c e \delta e,  \\
        \nonumber& \mathcal{E}_\zeta = \int_\Gamma \iota_\zeta e e \delta\omega + \iota_\zeta (\omega-\omega_0) e \delta e,\\
        &\mathcal{K}_\eta =\int_\Gamma \eta e_m e \delta \omega + \eta (\omega-\omega_0)_m e \delta e, \label{eq: 1-forms}\\
        \nonumber& \mathcal{F}_\lambda=\int_\Gamma \lambda \epsilon_n e \delta\omega,
    \end{align}
to each of which we pair the corresponding odd vector field, given by the push forward of the Hamiltonian vector fields of the constraints $(\tilde{\pi}_\Gamma)_*(\mathbb{X}_\alpha)$. Explicitly, after some manipulations, we have
        \begin{align}
         \nonumber&\mathbb{J}_{c,e}=[c,e] \quad&&\mathbb{J}_{c,e_m}=[c,e_m]\\
         \nonumber&\mathbb{J}_{c,\omega}=\de_\omega c\quad&&\mathbb{J}_{c,\omega_m}=\de_{\omega_m} c\\
         \nonumber&\mathbb{E}_{\zeta,e}=-L_\zeta^{\omega_0}e \quad&&\mathbb{E}_{\zeta,e_m}=L_\zeta^{\omega_0}e_m+\iota_{\partial_m\zeta}e\\
         \nonumber&\mathbb{E}_{\zeta,\omega}=-\iota_\zeta F_{\omega_0} -L_\zeta^{\omega_0}(\omega-\omega_0)\quad&&\mathbb{E}_{\zeta,\omega_m}=L_\zeta^{\omega_0}\omega_m+\iota_{\partial_m\zeta}(\omega-\omega_0)-\iota_{\partial_m\zeta}\omega_0\\
        &\mathbb{F}_{\lambda,e}=\de_\omega(\lambda \epsilon_n)+\lambda\sigma \quad&&\mathbb{F}_{\lambda,e_m}=\de_{\omega_m}(\lambda \epsilon_n)+\lambda\sigma_m    \label{eq: v.f. corner}\\
         \nonumber&e\mathbb{F}_{\lambda,\omega}=\lambda \epsilon_n F_\omega\quad&&e\mathbb{F}_{\lambda,\omega_m}+e_m\mathbb{F}_{\lambda,\omega}=\lambda \epsilon_n(\partial_m \omega-\de_\omega\omega_m)\\
         \nonumber&\mathbb{K}_{\eta,e}=-\eta\de_{\omega^0_m}e+d\eta e_m\quad&&\mathbb{K}_{\eta,e_m}=\de_{\omega^0_m}(\eta e_m)\\
         \nonumber&\mathbb{K}_{\eta,\omega}=-\eta\de_{\omega^0_m}\omega+d\eta(\omega-\omega^0)_m+\eta\de\omega^0_m\quad&&\mathbb{K}_{\eta,\omega_m}=\de_{\omega^0_m}(\eta(\omega-\omega^0)_m).
    \end{align}

\begin{definition}
    Considering again the infinitesimal neighborhood $\Sigma_\epsilon=\Gamma\times[0,\epsilon]  $ of the corner, the \textbf{space of corner fields} $F_\Gamma$ is defined to be the restriction at $\Gamma\times \{0\}$ of the coisotropic submanifold $C_{\Sigma_\epsilon}\subset F_{\Sigma_\epsilon} $. In particular, $F_\Gamma$ can be realized as a submanifold of $\tilde{F}_\Gamma$ by imposing the following constraints, descending from the definition of $C_{\Sigma_\epsilon} $:
    \begin{align}
    &e_m\de_\omega e=e\de_\omega e_m-e\de_{\omega_m}e\label{add_constr_dom}\\
    &e_mF_\omega-eF_{\omega_m}=0\label{add_constr_F}.
    \end{align}
    Additionally, $F_\Gamma$ is subject to the structural constraints mentioned above
    \begin{align}
        &\epsilon_n \de_\omega e = e\sigma, \label{constr: corner omega}\\
        & \epsilon_n (\de_{\omega_m} e + \de_\omega e_m) = e_m \sigma + e \sigma_m  \label{constr: corner omega_m}
    \end{align}
which the restriction to the structural constraint \cref{eq: str constr bdry} in the boundary.
We therefore obtain a surjective submersion $\pi_\Gamma\colon F_\Sigma \to F_\Gamma$ by post-composing the projection $\tilde{\pi}_\Gamma$ with the restriction map $\tilde{F}_\Gamma\to F_\Gamma$. 
\end{definition}
\begin{remark}
    In the following, we will assume the metric $g^\Gamma\coloneqq e^*(\eta)$ induced on the corners is either time-like or space-like, i.e. such that it is non-degenerate.
\end{remark}

\begin{remark}
Working on $F_\Gamma$ is the same as working on-shell. Indeed, this can be understood by regrouping the constraints as $\epsilon_m F_\omega - e F_{\omega_m}=0$ and 
    \[
    \begin{cases}
        \epsilon_n \de_\omega e = e\sigma \\
        e_m \de_\omega e = e (\de_\omega e_m -\de_{\omega_m} e)\\
        \epsilon_n (\de_\omega e_m -\de_{\omega_m} e) + e_m \sigma = e \sigma_m,
    \end{cases}
    \]
    which, by applying Lemma \ref{lem: (2,1) constraint vanish} with $\alpha=\de_\omega e $ and $\rho = \de_\omega e_m -\de_{\omega_m} e$, uniquely implies
    \begin{align}
        &\de_\omega e=0\label{eq: constr omega corner}\\
        &\de_{\omega_m}e -\de_\omega e_m=0.\label{eq: constrm dme corner}
    \end{align}
    This is not a surprise, since on the boundary $\Sigma$, thanks to Lemma \ref{lem:Omega2,1_d4}, imposing $e \de_\omega e = 0$ and $\epsilon_n \de_\omega e = e\sigma$  implies $\de_\omega e=0$. The same condition then must descend to the corner, along with its transverse component.
\end{remark}
\begin{definition}
Letting $\mathcal{X}_\alpha=\{\mathcal{J}_c,\mathcal{E}_\zeta,\mathcal{K}_\eta,\mathcal{F}_\lambda\}\in\Omega^1(F_\Gamma)[1]$ and $\mathbb X_\alpha=\{ \mathbb J_c, \mathbb E_\zeta, \mathbb K_\eta, \mathbb F_\lambda\}\in\mathfrak{X}[1](F_\Gamma)$, we define $ {D}$ to be the distribution
    \[
     {D}\coloneqq \mathrm{span}_{\mathcal{C^1}(F_\Gamma)}\{\mathbb X_\alpha + \mathcal{X}_\alpha\}_{\alpha=\{c,\zeta,\eta,\lambda\}}\subset \Gamma((T\oplus T^*)[1]F_\Gamma)
    \]
\end{definition}

\begin{proposition}\label{prop: involutivity and isotropicity}
    The distribution $ {D}$ is involutive and isotropic, i.e.
        \[
        [ {D}, {D}]\subset  {D}\qquad \mathrm{and}\qquad \langle {D}, {D}\rangle=0,
        \]
    where $[\cdot,\cdot]$ is the Dorfmann bracket and $\langle\cdot,\cdot\rangle$ is the inner product on the standard Courant algebroid.
\end{proposition}
\begin{proof}
    We limit ourselves to show the isotropicity, leaving the involutivity for Appendix \ref{app:proof invol}. 
    First notice that, since we're working with odd quantities, any element $\mathbb X_\alpha + \mathcal{X}_\alpha$ is null-like, i.e. such that
        \[
        \langle \mathbb X_\alpha + \mathcal{X}_\alpha , \mathbb X_\alpha + \mathcal{X}_\alpha \rangle =0.
        \]
    For the sake of completeness, however, we consider couples of parameters of the same kind and show the isotropicity. In particular, we start with
        \begin{align*}
            \langle \mathbb J_c + \mathcal{J}_c, \mathbb J_{c'} + \mathcal{J}_{c'}\rangle &= \intc c' e [c,e] - c e [c',e] = \intc c'\left[c,\frac{e^2}{2}\right] - c \left[c',\frac{e^2}{2}\right]=0,
        \end{align*}
    having used graded Leibniz. 
        \begin{align*}
            \langle \mathbb E_\zeta + \mathcal{E}_\zeta, \mathbb E_{\zeta'} + \mathcal{E}_{\zeta'}\rangle &= \intc \iota_{\zeta'}\frac{e^2}{2} (-\iota_{\zeta} F_{\omega_0} - \mathrm{L}_{\zeta}^{\omega_0}(\omega-\omega_0)) - \iota_{\zeta'}(\omega-\omega_0) \mathrm{L}^{\omega_0}_{\zeta}\frac{e^2}{2}\\
            &\phantom{=\intc} - \iota_{\zeta}\frac{e^2}{2} (-\iota_{\zeta'} F_{\omega_0} - \mathrm{L}_{\zeta'}^{\omega_0}(\omega-\omega_0)) + \iota_{\zeta}(\omega-\omega_0) \mathrm{L}^{\omega_0}_{\zeta'}\frac{e^2}{2}\\
            &=\intc \frac{e^2}{2}\left(\iota_{\zeta'}\iota_\zeta F_{\omega_0} + \iota_{\zeta'}\mathrm{L}_{\zeta}^{\omega_0}(\omega-\omega_0) - \mathrm{L}_{\zeta}^{\omega_0}\iota_{\zeta'}(\omega-\omega_0) \right)\\
            &\phantom{=\intc} -\frac{e^2}{2}\left(\iota_\zeta\iota_{\zeta'} F_{\omega_0} + \iota_{\zeta}\mathrm{L}_{\zeta'}^{\omega_0}(\omega-\omega_0) - \mathrm{L}_{\zeta'}^{\omega_0}\iota_{\zeta}(\omega-\omega_0) \right)\\
            &=\intc \frac{e^2}{2} \left([\iota_{\zeta'},\mathrm{L}_{\zeta}^{\omega_0}](\omega-\omega_0) - [\iota_{\zeta},\mathrm{L}_{\zeta'}^{\omega_0}](\omega-\omega_0)  \right)\\
            &=\intc \frac{e^2}{2} \left(\iota_{[\zeta',\zeta]}(\omega-\omega_0) - \iota_{[\zeta,\zeta']}(\omega-\omega_0)  \right)=0,
        \end{align*}
    having used that $\iota_{\zeta}\iota_{\zeta'}=\iota_{\zeta'}\iota_{\zeta}$ and $[\zeta,\zeta']=[\zeta',\zeta]$.
        \begin{align*}
            \langle \mathbb K_\eta + \mathcal{K}_\eta, \mathbb K_{\eta'} + \mathcal{K}_{\eta'}\rangle &= \intc \eta' \emm e (-\eta\de_{\omega^0_m}\omega + \de\eta (\omega-\omega_0)_m + \eta \de\omega^0_m ) + \eta'(\omega-\omega_0)_m \de_\omega(\eta \emm e)\\
            &\phantom{=\intc}- \eta \emm e (-\eta'\de_{\omega^0_m}\omega + \de\eta' (\omega-\omega_0)_m + \eta' \de\omega^0_m ) - \eta(\omega-\omega_0)_m \de_\omega(\eta' \emm e)\\
            &=\intc \eta'\emm e \de\eta(\omega-\omega_0)m + \eta' (\omega-\omega_0)_m \de\eta \emm e - (\eta \leftrightarrow \eta')=0,
        \end{align*}
        \begin{align*}
            \langle \mathbb F_\lambda + \mathcal{F}_\lambda, \mathbb F_{\lambda'} + \mathcal F_{\lambda'}\rangle &= \intc  \lambda' \en \lambda \en F_{\omega} - \lambda \en \lambda' \en F_\omega=0.
        \end{align*}
    We are now left to show the isotropicity in the off-diagonal entries.
    \begin{align*}
        \langle \mathbb{J}_c+\mathcal{J}_c, \mathbb{E}_\zeta+\mathcal{E}_\zeta \rangle&=\iota_{\mathbb J_c}\intc\iota_\zeta \frac{e^2}2\delta\omega+\iota_\zeta(\omega-\omega_0) e\delta e-\iota_{\mathbb E_\zeta}\intc ce\delta e\\[3pt]
        &=\intc\iota_\zeta\frac{e^2}2 \de_\omega c+\iota_\zeta(\omega-\omega_0)e[c,e]+ceL_\zeta^{\omega_0}e\\[3pt]
        &=\intc-\frac{e^2}2\iota_\zeta \de_\omega c+\frac{e^2}2[\iota_\zeta(\omega-\omega_0),c]+\frac{e^2}2L_\zeta^{\omega_0}c\\[3pt]
        &=0,
    \end{align*}
    \begin{align*}
        \langle \mathbb{J}_c+\mathcal{J}_c, \mathbb{K}_\eta+\mathcal{K}_\eta \rangle&=\iota_{\mathbb J_c}\intc\eta e_me\delta\omega +\eta(\omega-\omega_0)_me\delta e-\iota_{\mathbb K_\eta}\intc ce\delta e\\[3pt]
        &=\intc \eta e_me\de_\omega c+\eta(\omega-\omega_0)_m[c,\frac{e^2}2]+ce(\eta \de_{\omega_m^0}e-\de\eta e_m)\\[3pt]
        &=\intc \de\eta e_m ec-\eta \de_\omega (e_m e)c+c[\eta(\omega-\omega_0)_m,\frac{e^2}2]\\
        &\phantom{=\intc}+c\eta \de_{\omega_m^0}\frac{e^2}2-ce\de\eta e_m\\[3pt]
        &=\intc c\eta \de_\omega(e_me)+c\eta e\de_{\omega_m}e\\[3pt]
        &=0,
    \end{align*}
    where we imposed \cref{eq: constr omega corner} and \cref{eq: constrm dme corner},
    \begin{align*}
        \langle \mathbb{J}_c+\mathcal{J}_c, \mathbb{F}_\lambda+\mathcal{F}_\lambda \rangle&=\iota_{\mathbb J_c}\intc\lambda \epsilon_n e\delta\omega-\iota_{\mathbb F_\lambda}\intc ce\delta e\\[3pt]
        &=\intc \lambda \epsilon_n e \de_\omega c-c\de_\omega(\lambda \epsilon_ne)\\[3pt]
        &=0,
    \end{align*}
    \begin{align*}
        \langle \mathbb{E}_\zeta+\mathcal{E}_\zeta, \mathbb{K}_\eta+\mathcal{K}_\eta \rangle&=\iota_{\mathbb E_\zeta}\intc\eta e_m e\delta\omega+\eta(\omega-\omega_0)_me\delta \iota_{\mathbb K_\eta}\intc\iota_\zeta\frac{e^2}2\delta\omega+\iota_\zeta(\omega-\omega_0)e\delta e\\[3pt]
        &=\intc\eta e_me(\iota_\zeta F_\omega+\de_\omega\iota_\zeta(\omega-\omega_0))-\eta(\omega-\omega_0)_mL_{\zeta}^{\omega_0}\frac{e^2}2
         -\iota_\zeta e eF_{\omega_m}\\
        &\phantom{=\intc}-\iota_\zeta e e\de_\omega(\eta(\omega-\omega_0)_m)+\iota_\zeta(\omega-\omega_0)e\eta\de_{\omega_m^0}e-\iota_\zeta(\omega-\omega_0)e\de\eta e_m\\[3pt]
        &=\intc \eta e_m\iota_\zeta e F_\omega -\de\eta e_m e\iota_\zeta (\omega-\omega_0) +\eta \de_\omega(e_m e)\iota_\zeta(\omega-\omega_0)\\
        &\phantom{=\intc}+\eta(\omega-\omega_0)_m\de_{\omega_0}\iota_\zeta\frac{e^2}2 -\iota_\zeta eeF_{\omega_m}-\eta \de_\omega\iota_\zeta\frac{e^2}2(\omega-\omega_0)_m\\
        &\phantom{=\intc}+\iota_\zeta(\omega-\omega_0)e\eta\de_{\omega_m^0}e-\iota_\zeta(\omega-\omega_0)e\de\eta e_m\\[3pt]
        &=\intc\eta e_m\iota_\zeta eF_\omega+\eta e\de_{\omega_m}e\iota_\zeta(\omega-\omega_0)-\iota_\zeta e eF_{\omega_m}\\
        &\phantom{=\intc}-\eta(\omega-\omega_0)_m[\omega-\omega_0,\iota_\zeta\frac{e^2}2]+\iota_\zeta(\omega-\omega_0)e\eta\de_{\omega_m^0}e\\[3pt]
        &=\intc \eta\iota_\zeta e(e_m F_\omega-eF_{\omega_m})-\eta e\de_{\omega_m}e\iota_\zeta(\omega-\omega_0)\\
        &\phantom{=\intc}+\eta\iota_\zeta(\omega-\omega_0)([(\omega-\omega_0)_m,\frac{e^2}2]-\de_{\omega_m^0}\frac{e^2}2)\\[3pt]
        &=0,
    \end{align*}
where we  used \cref{eq: constr omega corner}, \cref{eq: constrm dme corner} and \cref{add_constr_F},
        \begin{align*}
        \langle \mathbb{E}_\zeta+\mathcal{E}_\zeta, \mathbb{F}_\lambda+\mathcal{F}_\lambda \rangle&=\iota_{\mathbb E_\zeta}\intc\lambda \epsilon_n e\delta\omega-\iota_{\mathbb F_\lambda}\intc\iota_\zeta\frac{e^2}2\delta\omega+\iota_\zeta(\omega-\omega_0)e\delta e\\[3pt]
        &=\intc\lambda \epsilon_n e(-\iota_\zeta F_\omega+\de\omega\iota_\zeta(\omega-\omega_0))-\iota_\zeta e\lambda \epsilon_n F_\omega-\iota_\zeta(\omega-\omega_0)\de_\omega(\lambda \epsilon_n e)\\[3pt]
        &=0,
    \end{align*}
    and
    \begin{align*}
        \langle \mathbb{F}_\lambda+\mathcal{F}_\lambda, \mathbb{K}_\eta+\mathcal{K}_\eta \rangle&=\iota_{\mathbb F_\lambda}\intc\eta e_m e\delta\omega+\eta(\omega-\omega_0)_me\delta e-\iota_{\mathbb K_\eta}\intc\lambda \epsilon_n e\delta\omega\\[3pt]
        &=\intc\eta e_m\lambda \epsilon_n F_\omega+\eta(\omega-\omega_0)_m\de_\omega(\lambda \epsilon_ne)-\lambda \epsilon_ne\eta F_{\omega_m}-\lambda \epsilon_n e\de_\omega(\eta(\omega-\omega_0)_m)\\[3pt]
        &=\intc\eta\lambda \epsilon_n(e_mF_\omega-eF_{\omega_m})\\[3pt]
        &=0,
    \end{align*}
    where we implemented again \cref{add_constr_F}, which concludes the proof of the isotropy.

\end{proof}

\subsection{Maximality and reduction}
In order to obtain a genuine Dirac structure, which equates to imposing the maximality condition, we follow the prescription described in \cref{sec: classical corner}. Given that $F_\Gamma$ is the submanifold of $\tilde{F}_\Gamma
$ on which constraints \cref{eq: constr omega corner}, \cref{eq: constrm dme corner}  and  \cref{add_constr_F} are imposed, it is not easy to work with $TF_\Gamma$. Therefore, we quotient $\tilde F_\Gamma$ with respect to the distribution $\tilde{\mathfrak{W}}\coloneqq\Ker{\sum_\alpha \mathcal{X}_\alpha}\subset \Gamma(T\tilde{F}_\Gamma)$ and then apply the constraints on the quotient.\footnote{As we will see, this will not spoil the involutivity and isotropy property.}

A quick computation\footnote{To be precise, $\tilde{\mathfrak{W}}$ is given by those vector fields $\mathbb X\in\mathfrak{X}(F_\Gamma)$ such that
    \begin{align*}
        &\int_\Gamma c e \mathbb X_e=0 && \int_\Gamma \iota_\zeta e e \mathbb X_\omega + \iota_\zeta(\omega-\omega_0) e \mathbb{X}_e=0\\
        &\int_\Gamma \lambda \epsilon_n e \mathbb X_\omega =0 &&\int_\Gamma \eta e_m e \mathbb X_\omega + \eta (\omega-\omega_0)_m e \mathbb X_e=0.
    \end{align*}
    From which we obtain $e \mathbb X_e = 0$ and $ (\iota_\zeta e + \eta e_m +\lambda \epsilon_n) e \mathbb X_\omega= \mu e \mathbb X_\omega = 0 $, for all $\mu\in\Omega^{0,1}_\Gamma $, implying $e \mathbb X_\omega=0$.
} shows that $\tilde{\mathfrak{W}}$ is given by
    \[
    \left\{\int_\Gamma \mathbb X_e \frac{\delta}{\delta e} + \mathbb X_{e_m} \frac{\delta}{\delta e_m}+\mathbb X_\omega \frac{\delta}{\delta \omega}  + \mathbb X_{\omega_m} \frac{\delta}{\delta \omega_m} + \sum_{k\geq1}\left(\mathbb X_{\partial^k_m e} \frac{\delta}{\delta (\partial^k_m e)} + \mathbb X_{\partial^k_m \omega} \frac{\delta}{\delta( \partial^k_m \omega)} \right)\ \bigg\vert \ e\mathbb X_e=0, \ e\mathbb X_\omega=0 \right\}.
    \]
Notice that we can split $\tilde{\mathfrak{W}}\simeq\tilde{\mathcal{B}}\oplus\mathfrak{X}(\Omega_\Gamma^{0,1}\oplus  \Omega^{0,2}_\Gamma\oplus \mathfrak{J}_m(\og^{1,1}\oplus \mathcal{A}_\Gamma) )$
with
    \[
    \tilde{\mathcal{B}}=\left\{\mathbb X=\int_\Gamma \mathbb X_e \frac{\delta}{\delta e} +\mathbb X_\omega \frac{\delta}{\delta \omega}\ \bigg\vert \ e\mathbb X_e=0, \ e\mathbb X_\omega=0. \right\}.
    \]

\begin{remark}
    The above splitting will imply that, in the quotient, we can discard the transversal fields $e_m$ and $\omega_m$ and the jets to all orders $\partial^k_m e $ and $\partial^k_m \omega$ as dynamical fields, fixing their values in a way compatible with the constraints.
\end{remark}

\begin{proposition}
    The distribution $\tilde{\mathfrak{W}}$  is regular and involutive.
\end{proposition}
\begin{proof}
    The regularity condition is immediately verified by noticing that $\dim \tilde{\mathcal{B}}=\dim \ker(W_1^{\Gamma,(1,1)})+ \dim \ker(W_1^{\Gamma,(1,2)})$, which is constant.

    Regarding involutivity, we first notice that $\mathfrak{X}(\Omega_\Gamma^{0,1}\oplus \Omega^{0,2}_\Gamma \oplus \mathfrak{J}_m(\og^{1,1}\oplus \mathcal{A}_\Gamma))$ is involutive by construction, while $[\tilde{\mathcal{B}},\mathfrak{X}(\Omega_\Gamma^{0,1}\oplus \Omega^{0,2}_\Gamma\oplus \mathfrak{J}_m(\og^{1,1}\oplus \mathcal{A}_\Gamma) )]=0$ trivially. We are then left with showing that $\mathcal{\tilde{B}}$ is involutive. 
    To do so, consider the following identity
        \[ 
        \iota_{[\mathbb X, \mathbb Y]} \mathcal{X}_\alpha= \mathrm{L}_{\mathbb X}\iota_{\mathbb Y}\mathcal{X}_\alpha - \iota_{\mathbb Y} \mathrm{L}_{\mathbb X}\mathcal{X}_\alpha,
        \]
    where $\mathrm{L}_{\mathbb X} = \iota_{\mathbb X} \delta + \delta \iota_{\mathbb X} $. Letting  $\mathbb X, \mathbb Y\in\tilde{\mathcal{B}}$, i.e. such that for all $\alpha=c,\zeta,\eta,\lambda$, $\iota_{\mathbb X}\mathcal{X}_\alpha=\iota_{\mathbb Y}\mathcal{X}_\alpha=0$, we have
        \begin{align*}
            -\iota_{[\mathbb X, \mathbb Y]} \mathcal{X}_\alpha=& \iota_{\mathbb Y} \mathrm{L}_{\mathbb X}\mathcal{X}_\alpha\\
            =&  \iota_{\mathbb Y}  \iota_{\mathbb X} \delta\mathcal{X}_\alpha +  \iota_{\mathbb Y}  \delta  \iota_{\mathbb X}\mathcal{X}_\alpha\\
            =&\iota_{\mathbb Y}  \iota_{\mathbb X} \delta\mathcal{X}_\alpha . 
        \end{align*}
   For $\alpha=c,\zeta$, we notice that $\mathcal{J}_c$ and $\mathcal{E}_\zeta$ are $\delta$-exact, hence obtaining $\iota_{[\mathbb X, \mathbb Y]} \mathcal{J}_c=\iota_{[\mathbb X, \mathbb Y]} \mathcal{E}_\zeta=0$. We then just need to show $\iota_{[\mathbb X, \mathbb Y]} \mathcal{K}_\eta=\iota_{[\mathbb X, \mathbb Y]} \mathcal{F}_\lambda=0$. We have
    \begin{align*}
        \iota_{[\mathbb X, \mathbb Y]} \mathcal{K}_\eta=\int_\Gamma & \eta \mathbb X_{e_m} e \mathbb Y_\omega - \eta \mathbb Y_{e_m} e \mathbb X_\omega -\eta e_m (\mathbb X_e \mathbb Y_\omega - \mathbb Y_e \mathbb X_\omega) - \eta \mathbb X_{\omega_m} e \mathbb Y_e    + \eta \mathbb Y_{\omega_m} e \mathbb X_e    \\
        =\int_\Gamma & \eta e_me( \mathrm{L}_{\mathbb X}\mathrm{L}_{\mathbb Y}(\omega-\omega_0) - \mathrm{L}_{\mathbb Y}\mathrm{L}_{\mathbb X}(\omega-\omega_0))\\
        =\int_\Gamma &\eta e_m e \mathrm{L}_{[\mathbb X,\mathbb Y]}(\omega-\omega_0)\\
        =\int_\Gamma &\eta e_m e [\mathbb X,\mathbb Y]_\omega,
    \end{align*}
    having used $e\mathbb X_e= e\mathbb Y_e=0$, $e\mathbb X_\omega= e\mathbb Y_\omega=0$, the fact that $\mathbb X_\Phi= \mathrm{L}_{\mathbb X} \Phi$ for any $\Phi\in \tilde{F}_\Gamma$ and the identity
    \[
    \mathrm{L}_{\mathbb X}\mathrm{L}_{\mathbb Y} - \mathrm{L}_{\mathbb Y}\mathrm{L}_{\mathbb X}= \mathrm{L}_{[\mathbb X,\mathbb Y]}.
    \]
    It is a quick computation to see that 
    \[
     [\mathbb X,\mathbb Y]_\omega = \left(\mathbb X_e \frac{\delta \mathbb Y_\omega}{\delta e} + \mathbb X_\omega \frac{\delta \mathbb Y_\omega}{\delta \omega}\right) - \left(\mathbb Y_e \frac{\delta \mathbb X_\omega}{\delta e} + \mathbb Y_\omega \frac{\delta \mathbb X_\omega}{\delta \omega}\right),
    \]
    hence verifying $e[\mathbb X,\mathbb Y]_\omega=0$. Lastly, we are left with
    \begin{align*}
    \iota_{[\mathbb X,\mathbb Y]}\mathcal{F}_\lambda=\int_\Gamma \lambda \epsilon_n e [\mathbb X,\mathbb Y]_\omega=0,
    \end{align*}
    showing involutivity.
\end{proof}

\subsubsection{The reduced space of corner fields}\label{sec: reduced corner space}
The quotient of $\tilde{F}_\Gamma$ with respect to $\tilde{\mathfrak{W}}$ is given  by
    \[\tilde{P}_\Gamma\coloneqq \tilde{F}_\Gamma/\tilde{\mathfrak{W}}\simeq(\Omega^{1,1}_\Gamma\oplus \mathcal{A}_\Gamma)/\tilde{\mathcal{B}}.\]
This is achieved by fixing the transversal fields $e_m\in\og^{0,1}$ and $\omega_m\in\og^{0,2}$ and the transversal jets $\partial^k_m e \in\og^{1,1}$ and $\partial^k_m \omega\in\mathcal{A}_\Gamma$. 
We will see how some of these fields will be fixed by imposing the structural constraints on the reduced space of corner fields. For the moment, we choose a particular $e_m$ that is convenient for future calculations.
\begin{definition}
    We define $\epsilon_m\in\og^{0,1}$ to be the unique section of $\mathcal{V}$ which satisfies the following conditions
        \begin{equation}\label{eq: def e_m}
        \begin{split}
            & [e,\epsilon_m]=0\\
            & [\epsilon_n,\epsilon_m]=0\\
            & [\epsilon_m,\epsilon_m]=\nu,
        \end{split}
        \end{equation}
    with $\nu=\pm 1$ depending on the signature of $\Gamma$ and $\Sigma$.

\end{definition}

\begin{remark}
    The above definition is well given, as $\epsilon_m$ has four components which are point-wise uniquely determined by the above four equations. They amount to requesting $\epsilon_m$ to be linearly independent from $e$ and $\epsilon_n$, while also being normalized to $\pm 1$. It follows that $\epsilon_m\in\og^{0,1}$ is such that $\{e_i,\epsilon_m,\epsilon_n\}$ form a local basis of $\mathcal{V}$.
\end{remark}

\begin{proposition}
    We can fix $\epsilon_m\in\og^{0,1}$ as a representative of $[e_m]$ in $\tilde{P}_\Gamma$.
\end{proposition}
\begin{remark}
    It is important to notice that $\epsilon_m$ depends on $e$, and therefore $\delta \epsilon_m\neq 0.$
\end{remark}
\begin{remark}
    We can equivalently define $\emm$ by taking the hodge dual of $E\en$ (both in $\wedge^\bullet V$ and $\wedge^\bullet T^* M$). Indeed, assuming $\{v_a\}$ is an orthonormal basis of $V$, i.e. such that $[v_a,v_b]=\eta_{ab}$, one has $v_a v_b v_c v_d = \frac{1}{4!}\epsilon_{abcd} \mathrm{Vol}_V$, hence we can define 
        \begin{equation}
            \emm^a \coloneqq \frac{\varsigma}{2\cdot 4!}\epsilon^{ij}\eta^{ab}\epsilon_{bcdf} e^c_i e^d_j \en^f,  \label{def emm}
        \end{equation}
    where $\varsigma$ is a normalisation factor defined such that $E\emm \en= \mathrm{Vol}_\Gamma \mathrm{Vol}_V$.\footnote{With such definition, one automatically obtains that $[e,\emm]=0$ and $[\en,\emm]=0$.}
\end{remark}

We can now move on to describe the quotient $(\Omega^{1,1}_\Gamma\oplus \mathcal{A}_\Gamma)/\tilde{\mathcal{B}}$, starting from the coframes $[e]$.
Denoting the reduced space of corner vielbeins by $\Omega_{\Gamma,r}^{1,1}$, we have
    \[
    \Omega_{\Gamma,r}^{1,1}\coloneqq \{ [e]\subset \Omega_{\Gamma}^{1,1} \ \vert \ e' \sim e \ \mathrm{iff}  \ \exists \ t\geq 0 \ \mathrm{s.t.} \ e'=\Psi_{\mathbb X}^t(e), \text{ for some } \mathbb X\in \tilde{\mathfrak{W}} \ \},
    \]
where $\Psi_{\mathbb X}^t$ is the flow of $\mathbb X$ at time $t$. Then, on an open $\mathcal{U}_\Gamma\subset\Omega_{\Gamma,r}^{1,1} $, one can identify $\tilde{P}_\Gamma$ with the fiber bundle $\tilde{P}_\Gamma\longrightarrow \Omega_{\Gamma,r}^{1,1}$ such that locally
    \[
    \tilde{P}_\Gamma\big\vert_{\mathcal{U}_\Gamma}\simeq \mathcal{U}_\Gamma\times \mathcal{A}_{\Gamma,r},
    \]
where $\mathcal{A}_{\Gamma,r}\coloneqq \{ [\omega]\subset \mathcal{A}_\Gamma \ \vert \ \omega'\sim \omega \ \mathrm{iff} \ \omega' - \omega = v, \text{ s.t. } ev=0 \}$.

The difficulty of working with quotient spaces is evident. Therefore, it is convenient to find a more manageable description of such spaces. In particular, we can overcome the problem of dealing with $\Omega_{\Gamma,r}^{1,1}$ with two equivalent strategies: fix a representative of $[e]$ by working in a specific submanifold of $\Omega_{\Gamma}^{1,1}$ that is in 1-to-1 correspondence with $\Omega_{\Gamma,r}^{1,1} $, or establish a different bijection between $\Omega_{\Gamma,r}^{1,1}$ and another more convenient space. 

The latter option has been explored in \cite{CC2022cor} and is realized by considering the field redefinition $E\coloneqq \frac{e^2}{2}\in\Omega_\Gamma^{2,2} $. It is easy to prove that $E$ only depends on the equivalence class $[e]$. Indeed, considering the infinitesimal transformation $e\mapsto e+\alpha$, with $e\alpha =0$ and $\alpha^2=0$,\footnote{One can think of $\alpha$ as $\alpha=\mathbb X_e dt$, then $\alpha^2\propto dt^2=0$.} we see 
    \[
    E\longmapsto \frac{1}{2}(e+\alpha)^2= \frac{1}{2}e^2 + e\alpha + \frac{1}{2}\alpha^2=E.
    \]
It was established in \cite{CC2022cor} that any element $B\in\Omega^{2,2}_\Gamma$ is of the form $B=\frac{1}{2}a^2$ for some $a\in\Omega^{1,1}_\Gamma$ if and only if $\mathrm{Pf}(B)=0$, where Pf indicates the Pfaffian. Denoting by $\Omega^{2,2}_{\Gamma,0}$ the subset of $\Omega^{2,2}_{\Gamma}$ of elements with vanishing Pfaffian, one obtains the isomorphism 
    \begin{align*}
        \Omega^{1,1}_{\Gamma,r}\simeq\Omega^{2,2}_{\Gamma,0}.
    \end{align*}

\begin{remark}
    The non-degeneracy assumption of the vielbein $e$ is reformulated by the open condition
        \[
        \epsilon_m \epsilon_n E\neq 0.
        \]
\end{remark}
Alternatively, one can find a representative of $[e]$ thanks to the following.
\begin{theorem}\label{thm: fix e corn}
    Given $e^0\in\Omega^{1,1}_\Gamma$ as a reference vielbein, for all $\tilde{e}\in\Omega^{1,1}_\Gamma$ such that $ \epsilon_m\epsilon_n e(e^0)^2$  induces the same orientation as $\epsilon_m\epsilon_n \tilde{e}^2$,\footnote{In particular, one still requires $\{e_0,\epsilon_m,\epsilon_n\}$ to define a local basis of $\mathcal{V}$.} there exists a unique $e\in[\tilde{e}]$ such that
        \begin{align}\label{constr: fix e corner}
            \epsilon_n \epsilon_m e \in \Ima W_{e^0}^{\Gamma,(0,2)}.
        \end{align}
\end{theorem}
\begin{proof}
    By Lemma \ref{lem: decomp (1,3)}, we have that in general
        \[
        \epsilon_m \epsilon_n \tilde{e} = e^0 \alpha + \epsilon_m \epsilon_n \beta, \qquad \mathrm{s.t.}\qquad \begin{cases}
            \epsilon_m \epsilon_n \alpha=0\\
            e^0 \beta=0
        \end{cases}
        \]
    Therefore, defining $e\coloneqq \tilde{e}-\beta$ uniquely defines a representative of the equivalence class 
    \[
    [\tilde{e}]_0\coloneqq \{ e\in\Omega_\Gamma^{1,1} \ | \ e^0(\tilde{e}-e)=0\}.
    \]
    However, it is immediate to see that the two spaces $\Omega^{1,1}_{\Gamma,r_0} \coloneqq \{[\tilde{e}]_0\subset\Omega^{1,1}_\Gamma \} $ and $\Omega^{1,1}_{\Gamma,r}$ are isomophic. Therefore, fixing a representative of $[\tilde{e}]_0$ is equivalent to fixing a representative of $[{e}]$.
\end{proof}
\begin{remark}
    To have a more constructive perspective on the last theorem, one can alternatively notice that it is possible to write 
    \[e=\iota_\chi e^0 + \alpha_m\epsilon_m +\alpha_n \epsilon_n,\]
    with $\chi\in\Omega^{1,0}_\Gamma\otimes \mathfrak{X}(\Gamma)$ and $\alpha_m,\alpha_n\in\Omega^{0,0}_\Gamma$. Then the constraint $\epsilon_m \epsilon_n e = e^0 \alpha$ implies 
    \[\iota_{\chi}e^0 \epsilon_m\epsilon_n = e^0 \alpha,\]
    which, after retracing the steps in the proof of \ref{lem: decomp (1,3)}, is solved by $\iota_\chi e^0=f e^0$ for some $f\in\Omega^{0,0}_\Gamma $. Therefore, \cref{constr: fix e corner} imposes 3 equations, which ultimately fix the representative in the 5-dimensional space $\Omega^{1,1}_{\Gamma,r}$.  

    Additionally, from the proof of proposition 62 in \cite{CC2022cor}, we have that the isomorphism $\Omega_{\Gamma,r}^{1,1}\simeq \Omega_{\Gamma,0}^{2,2} $ is ensured by requiring that $E=\frac{1}{2}e^2$ is such that $e=fe^0 + \alpha_m \epsilon_m + \alpha_n \epsilon_n $, which is exactly what we find by imposing \cref{constr: fix e corner}.

\end{remark}

We are still left with the task of describing $\mathcal{A}_{\Gamma,r} $. Again, we can obtain an equivalent description by introducing a reference connection $\omega_0$,\footnote{This enables us to consider $\og^{1,2}$ as our model space for $\mathcal{A}_\Gamma$, as $(\omega-\omega_0)\in\og^{1,2}$.} and by considering the field redefinition
    \[\Omega\coloneqq e(\omega-\omega_0).\]
Indeed the map $W_1^{\Gamma,(1,2)}\colon (\omega-\omega_0)\mapsto e(\omega-\omega_0) $ is only sensitive to the equivalence class $[\omega-\omega_0]$. Furthermore, due to the surjectivity of $W_1^{\Gamma,(1,2)}$, we have that $W_1^{\Gamma,(1,2)}$ establishes an isomorphism
    \[
    \Omega_{\Gamma,r}^{1,2}\simeq \Omega_{\Gamma}^{2,3},
    \]
where $\Omega_{\Gamma,r}^{1,2}\coloneqq \{ [\omega-\omega_0]\subset\og^{1,2} \ | \  \omega\sim \omega' , \iff  \omega-\omega'\in\Ker{W_1^{\Gamma,(1,2)}}\}$.
\begin{remark}
    Fixing the representative of $[e]\in\Omega^{1,1}_{\Gamma,r}$ ensures that the maps $W_e^{\Gamma,(i,j)}$ are well defined on the quotient space $\tilde{P}_\Gamma$, guaranteeing that $\Omega$ is well-defined.
\end{remark}    

Equivalently, we can obtain a description of $\Omega_{\Gamma,r}^{1,2}$ by the following.
\begin{theorem}\label{thm: fix omega corner}
    For any $\omega'\in\mathcal{A}_\Gamma $, there exist unique $r\in\og^{0,1}$, $b\in\og^{1,0}$ and $v\in\og^{1,2}$ such that
    \[
    \omega'-\omega_0= v + er + \epsilon_m \epsilon_n b, \quad \mathrm{s.t.} \quad \begin{cases}
        [e,r]=\rog^{0,1} r=0\\
        ev=0
    \end{cases}
    \]
    Therefore, the constraint
    \begin{equation}\label{constr: fix omega corner}
        \omega-\omega_0= er + \epsilon_m\epsilon_nb
    \end{equation}
    uniquely fixes a representative of $[\omega-\omega_0]$.
\end{theorem}

\begin{proof}
We want to     prove that
    \[
    \og^{1,2}= \ker W_1^{\Gamma,(1,2)} \oplus e\ker (\rog^{0,1}) \oplus \epsilon_m\epsilon_n \og^{1,0}.
    \]
We start by checking that the dimensions of the subspaces add up to the dimension of $\og^{1,2}$. In particular, we know by \cite{CC2022cor} that $\dim (\ker (W_1^{\Gamma,(1,2)}))=8$, while it is easy to check that $[e,r]=0$ is a set of two independent equations, hence $\dim (\ker(\rog^{0,1}))=2$. One also trivially has $\dim\og^{1,0}=2$ . We observe
    \[
    \dim \og^{1,2}=12=\dim\ker W_1^{\Gamma,(1,2)} + \dim\ker\og^{0,1} + \dim\og^{1,0}.
    \]
Secondly, from \cref{thm: coframes prop} one sees that $W_1^{\Gamma,(0,1)} $ is injective. In the same way, a simple computations shows that $W_{mn}^{1,0}$ is also injective. 
We are left to show that $v+er+\epsilon_m\epsilon_nb=0$ implies $v=0$, $r=0$ and $b=0$. That is to say the three terms are linearly independent. 

By separately multiplying by $\epsilon_m$ and $\epsilon_n$ we obtain
    \[\epsilon_n v = \epsilon_n e r, \qquad \epsilon_m v = \epsilon_m e r,\]
Therefore, by defining $v_n=\epsilon_n r$ and $v_m=\epsilon_n r$, since $ev=0$, we can apply Lemma \ref{lem: (1,2) vanish} to show that $v=0$.

We are then left with $er=\epsilon_n\epsilon_m b$, from which we infer
    \[
    \begin{cases}
        e^2 \epsilon_n r =0\\
        e^2 \epsilon_m r=0\\
        e\epsilon_n\epsilon_m r=0
    \end{cases}
    \]
then, from \ref{lem: (0,1) vanish}, it follows $r=0$. Lastly, $\epsilon_n \epsilon_m b=0$ is uniquely solved by $b=0$ due to injectivity of $W_{mn}^{1,0}$.
\end{proof}

Given $\omega-\omega_0=er +\epsilon_m \epsilon_n b$, one would be tempted to fix 
\[
(\omega-\omega_0)_m=\epsilon_m (r +\epsilon_n b_m),
\]
however, the following choice will be more suitable for the computations to come.
\begin{definition}\label{fixomega_m}
    We fix the representative of $[\omega_m]$ by requiring
    \begin{equation}\label{eq: fix omega_m}
        (\omega-\omega_0)_m=0.
    \end{equation}
\end{definition}
%\begin{remark}
%    Notice that, as $\epsilon_m$, one has that $\delta(\omega-\omega_0)_m\neq 0$, as it depends both on $e$ and $(\omega-\omega_0)$. Moreover, it is also convenient to define
%    \begin{equation}
%        \Omega_m\coloneq e(\omega-\omega_0)_m + \epsilon_m(\omega-\omega_0)=e\epsilon_m r.
%    \end{equation}
%    with the property that $[e,\Omega_m]=0$.
%\end{remark}

Finally, we can define the reduced space of corner fields $P_\Gamma$ by imposing the constraints \cref{eq: constr omega corner}, \cref{eq: constrm dme corner} and \cref{add_constr_F}. 

\begin{proposition}
    The constraint \cref{eq: constrm dme corner} fixes the values of the transversal jet $\partial_m e$.
\end{proposition}
\begin{remark}\label{rem: partial m}
    In principle, one cannot integrate by parts with respect to $\partial_m$, as $\partial_m e$ and $\partial_m \omega$ are to be treated as fields. However, this operation becomes rigorous when considering the whole space of jets to all orders.
\end{remark}

\begin{proof}
   Wee see that \cref{eq: constrm dme corner} is such that 
    \[
    \partial_m e = \de_\omega \epsilon_m - [\omega_m,e],
    \]
    hence fixing $\partial_m e$. 
\end{proof}
\begin{remark}
    In the following subsection, we will see how to fix the transversal jet $\partial_m \omega$.
\end{remark}
Lastly, since the higher order transversal jets never appear in any computations, we simply fix them by requiring that the vanish in the quotient space $\tilde{F}_\Gamma/\tilde{\mathfrak{W}}$. The last constraint fixes some components of $\omega$, by requiring $\de_\omega e=0$, but not all of them.

\subsubsection{Bianchi identities and boundary fields revisited}\label{sec: Bianchi}

Bianchi identities in the bulk simply read
    \[
    \de_\omega F_\omega = 0 \qquad \text{and} \qquad \de_\omega^2(-)=[F_\omega,-].
    \]
The question arises on how such equations ``propagate'' to higher codimension strata.\footnote{In the following, we will denote bulk, boundary and corner fields with the same notation. They will be distinguishable by the context.} In order to answer the question fully, one needs to consider the transversal fields and jets to the boundary $\Sigma$.  In a way similar to the one presented above, it is possible to fix the values of these extra fields such that the resulting theory is equivalent to the one we presented in Section \ref{sec: bdry struct of PC}, when restricting to the coisotropic submanifold $C_\Sigma$. In particular, considering $(e_n,\omega_n)\in\Omega^{0,1}_\Sigma\oplus \Omega^{(0,2)}_\Sigma $  and the transversal jets $\partial ^k_n e $ and $\partial^i_n\omega$ of the fields at all orders, one can, for example, rewrite the condition $\de_\omega e =0$ as 
    \begin{align}
        e \de_\omega e = 0, \\
        \epsilon_n \de_\omega e = e (\de_{\omega_n} e - \de_\omega \epsilon_n) \label{eq: second bdry constr},
    \end{align}
having fixed the value of $e_n$ as $e_n =\epsilon_n$. As already mentioned in Remark \ref{rem: brdy strct constr}, notice here that \cref{eq: second bdry constr} is a reinterpretation of the structural constraint \cref{eq: str constr bdry} which includes the transversal jet $\partial_ne$. Furthermore, if we restrict ourselves to the submanifold where the full torsion constraint $\de_\omega e = 0$ is imposed, \cref{eq: second bdry constr} reduces to 
    \[
   \de_{\omega_n} e = \de_\omega \epsilon_n, 
    \]
which fixes the value of the transversal jet $\partial_n e$. 

On the other hand, when considering the transversal fields and jets, Einstein's equation in the bulk gives rise to an additional constraint descending to the boundary:
    \begin{equation}
        \en F_\omega = eF_{\omega_n} \label{constr omega_n}.
    \end{equation}
At first, one could think that this equation should somehow fix some components the curvature $F_\omega$, however, one can always find an element $\gamma\in\Omega_\Sigma^{1,2} $ such that $\en F_\omega=e\gamma$, as $W_{e}^{\Sigma,(1,2)}$ is surjective. We can then use \cref{constr omega_n} to fix some components of $F_{\omega_n}$ (and henceforth of $\partial_n \omega$), leaving out the ones which lie in the kernel of $W_e^{\Sigma,{1,2}}$. 

Taking the covariant derivative of the torsionless condition, the second Bianchi identity then simply implies $$[F_\omega, e]=0.$$ 
We stress that such equation holds only on the coisotropic submanifold which we called $C_\Sigma$, where $\de_\omega e =0$ has been imposed.

The exterior covariant derivative of \cref{eq: second bdry constr} yields 
    \begin{align*}
        \de_\omega \en \de_\omega e - \epsilon_n [F_\omega,e]=\de_\omega e (\de_{\omega_n}e - \de_\omega \en) + e (\de_{\omega_n}\de_\omega e + [F_{\omega_n},e] - [F_\omega,\en] ),
    \end{align*}
which, applying the on-shell condition and the second Bianchi identity, implies 
    \begin{equation}
        e([F_{\omega_n},e] - [F_\omega,\en]) =0. \label{eq: fake second bianchi}
    \end{equation}
This suggest considering the following equation, which we will refer to as $n$\textbf{-transversal Bianchi identity}:
    \begin{equation}
        [e,F_{\omega_n}]=[\epsilon_n,F_\omega]. \label{second constr omega_n}
    \end{equation}
\begin{proposition}\label{prop: fix partialnpmega}
    The $n$-transversal Bianchi identity \cref{second constr omega_n} and the $n$-transversal Einstein equation \cref{constr omega_n} uniquely fix $F_{\omega_n}$, hence the transversal jet $\partial_n \omega$ to the boundary. 
        
\end{proposition}
\begin{proof}
    We give some hints of proof here and leave the full one for the reader.% Appendix \ref{app: proof Bianchi}.
    
    We start by noticing that $e([F_{\omega_n},e] - [F_\omega,\en]) =0$ and $[F_{\omega_n},e] - [F_\omega,\en]=0$ are not equivalent. In particular, the latter comprises 12 equations while the former only 6. However, we argue that equation \cref{eq: fake second bianchi} holds trivially on $C_\Sigma$. Indeed, using the fact that $[e,e]=0$, \cref{constr omega_n} and the second Bianchi identity $[e,F_\omega]=0$, one can rewrite \cref{eq: fake second bianchi} as 
        \begin{align*}
            [\epsilon_n, eF_\omega]=0.
        \end{align*}
        Therefore, on $C_\Sigma$ \cref{second constr omega_n} only imposes 6 equations (lying in the kernel of $W_e^{\Sigma,(1,2)} $),  which, along with the 12 from \cref{constr omega_n}, uniquely fix the 18 components of $\partial_n \omega$.\footnote{Strictly speaking, the above equations fix $F_{\omega_n}$, but $F_{\omega_n}=-\partial_n\omega + \de_\omega(\omega_n)$.}
\end{proof}

Now we can repeat the same arguments for the corner fields. We summarise the result in the following table, setting $\sigma \coloneqq\de_{\omega_n} e - \de_\omega \epsilon_n$. 
% \usepackage{booktabs}
\begin{table}[H]\label{tab: bdry corner equations}
\centering
\begin{tabular}{@{}llll@{}}
\cmidrule(l){2-4}
                                 & $M$                                     & \hspace{18mm}$\Sigma$ & \hspace{22mm}$\Gamma$ \\ \cmidrule(l){2-4} \\
                                  & $\de_\omega e = 0$      & \makecell{$e \de_\omega e = 0,$ \\                                  
        $\epsilon_n \de_\omega e = e (\de_{\omega_n} e - \de_\omega \epsilon_n) $ }       &  \makecell{$\emm \de_\omega e = e(\de_{\omega_m} e - \de_\omega \epsilon_m),$ \\                                  
        $\epsilon_n \de_\omega e = e (\de_{\omega_n} e - \de_\omega \epsilon_n) $, \\
        $\en (\de_\omega \emm -\de_{\omega_m}e) = \emm \sigma + e\sigma_m $}     \\
        \\ \cmidrule(l){2-4} \\         
                                & $[e,F_\omega]=0 $ & \makecell{$[e,F_\omega]=0,$\\
                                    $[\en, F_\omega]=[e,F_{\omega_n}]$}        &  \makecell{ $[\en, F_\omega]=[e,F_{\omega_n}]$,\\ $[\emm,F_\omega]=[e,F_{\omega_m}],$\\ $[\en, F_{\omega_m}]=[\emm,F_{\omega_n}]$ }       \\
                                \\ \cmidrule(l){2-4} \\  
                                & $eF_\omega=0$        & \makecell{$e F_\omega =0$, \\ $\en F_\omega = e F_{\omega_n} $}        & \makecell{$\en F_\omega = e F_{\omega_n}$, \\ $\emm F_\omega = e F_{\omega_m}$, \\ $\en F_{\omega_m}=\emm F_{\omega_n} $ }        \\ 
                                \\\cmidrule(l){2-4} 
                                &                                         &          &         
\end{tabular}
\caption{Torsion, Einstein and second Bianchi in various codimensions} 
\end{table}

\begin{proposition}
    The transversal jet $\partial_m \omega$ is uniquely fixed (in the reduced space of corner fields $P_\Gamma$) by requiring 
    \begin{align*}
        & \emm F_\omega = e F_{\omega_m}, &\en F_{\omega_m}=\emm F_{\omega_n}, \\ 
        & [\emm,F_\omega]=[e,F_{\omega_m}], &[\en, F_{\omega_m}]=[\emm,F_{\omega_n}].
    \end{align*}
    
\end{proposition}
\begin{proof}
    The proof follows the same steps as the proof of Proposition \ref{prop: fix partialnpmega} on the manifold $\Gamma\times [0,\epsilon]\eqqcolon \Sigma_\epsilon$, by redefining
        \begin{align*}
            &e|_{\Sigma_\epsilon} = e + \emm \de x^m,\\
            & F_\omega|_{\Sigma_\epsilon}= F_\omega + F_{\omega_m}\de x^m,
        \end{align*}
    where $x^m$ is interpreted as the coordinate along the infinitesimal interval $[0,\epsilon]$.
\end{proof}

To summarize the above discussion, we will consider the following definition. 

\begin{theorem}
    The reduced space of corner fields is given by the submanifold $P_\Gamma$ of $\tilde{P}_\Gamma$, on which the constraint $\de_\omega e=0$ is imposed. Explicitly $P_\Gamma$ is diffeomorphic to the following space
    \begin{equation*}
        \{(E,\Omega,\epsilon_m,\omega_m, \partial_m e,\partial_m \omega) \ | \ \mathrm{s.t.} \ \cref{eq: def e_m}, \cref{constr: fix e corner}, \cref{constr: fix omega corner}, \cref{eq: fix omega_m} \text{ and 3$^{\mathrm{rd}}$ column of Table \ref{tab: bdry corner equations}} \},
    \end{equation*}
    while all higher transversal jets are identically set to zero. 
\end{theorem}

\subsection{Involutivity and isotropy of the induced (pre-)Dirac structure on $P_\Gamma$}\label{sec: inv_isotr}
Having defined the reduced space of corner fields $P_\Gamma$, we can induce vector fields $\mathbb X_\alpha$ and one forms $\mathcal{\chi}_\alpha$ from \cref{eq: v.f. corner} and \cref{eq: def corner 1-forms}. In particular, this amounts to rewriting the vector fields and one forms in terms of the dynamical fields $E$ and $\Omega$ in $P_\Gamma$. For example, for $\alpha=c$ we automatically have
    \begin{align*}
        \mathcal{J}_c=\int_\Gamma c \delta E.% \qquad \mathbb J_c=\int_\Gamma[c,E]\frac{\delta}{\delta E} + \left( [c,\Omega] + e\de_{\omega_0}c \right)\frac{\delta}{\delta \Omega}. 
    \end{align*}
For $\alpha= \zeta$,  we see
    \begin{align*}
        \mathcal{E}_\zeta&=\delta\int_\Gamma - e \iota_\zeta\Omega =\int ( \iota_\zeta e (\omega-\omega_0) + e \iota_\zeta(\omega-\omega_0) )\delta e + \iota_\zeta e \delta\Omega\\
        &=\int_\Gamma \iota_\zeta e \delta\Omega + \iota_\zeta (\omega-\omega_0) \delta E+\iota_\zeta e (\omega-\omega_0) \delta e.
    \end{align*}
Notice however how, by \cref{thm: fix omega corner}, we can rewrite $\iota_\zeta e (\omega-\omega_0) \delta e = \iota_\zeta e (e r+ \epsilon_m \epsilon_n b) \delta e $. We notice that $\iota_\zeta e \epsilon_m \epsilon_n b\in\og^{1,3}$, therefore we can apply lemma \ref{lem: decomp (1,3)} to see 
    \[
    \iota_\zeta e \epsilon_m \epsilon_n b= e\alpha_\zeta + \epsilon_m\epsilon_n \beta_\zeta,
    \]
for some uniquely defined $\alpha_\zeta \in\og^{0,2} $ and $\beta_\zeta\in\ker(W_e^{\Gamma(1,1)})$. We then see 
\begin{align*}
    \epsilon_m \epsilon_n \beta_\zeta \delta e&= \epsilon_m \epsilon_n \delta(e\beta_\zeta)-\epsilon_m \epsilon_n e \delta\beta_\zeta\\
    &=-\epsilon_m \epsilon_n e\left(\frac{\delta \beta_\zeta}{\delta e} \delta e + \frac{\delta \beta_\zeta}{\delta\omega}\delta\omega\right)\\
    &=-\epsilon_m \epsilon_n e\frac{\delta \beta_\zeta}{\delta e} \delta e - \epsilon_m \epsilon_n \left( \frac{\delta (e \beta_\zeta)}{\delta \omega} -\frac{\delta e }{\delta \omega}\beta_\zeta\right)\delta\omega\\
    &=-\epsilon_m \epsilon_n \frac{\delta \beta_\zeta}{\delta e}\delta E
\end{align*}
We obtain
    \begin{align*}
        &\mathcal{E}_\zeta =\int_\Gamma \iota_\zeta e \delta\Omega + \left(\iota_\zeta (\omega-\omega_0) + \iota_\zeta e r + \alpha_\zeta - \emm\en \frac{\delta\beta_\zeta}{\delta e}\right) \delta E.%,\\
        %&\mathbb E_\zeta =\int_\Gamma -\mathrm{L}_\zeta^{\omega_0}(E)\frac{\delta}{\delta E } -\left(e \iota_\zeta F_\omega -\mathrm{L}_\zeta^{\omega_0}(\Omega) \right)\frac{\delta}{\delta\Omega}.
    \end{align*}
Similarly, we have
    \begin{align*}
        & \mathcal{K}_\eta=\int \eta \epsilon_m \delta\Omega + \eta \epsilon_m r  \delta E.%\\
       % &\mathbb K_\eta = \int_\Gamma \left( \de_\omega(\eta \epsilon_m e) - \eta[(\omega-\omega_0),E] \right)\frac{\delta }{\delta E} + (\de\eta \Omega_m - \eta\partial_m \Omega)\frac{\delta}{\delta\Omega},
    \end{align*}
    Lastly, we get
    \begin{align*}
        &\mathcal{F}_\lambda=\int_\Gamma \lambda\epsilon_n \delta\Omega +\lambda \epsilon_n r \delta E.%\\
%        &\mathbb F_\lambda=\int_\Gamma \de_\omega(\lambda \epsilon_n e)\frac{\delta}{\delta E} + \left( \lambda\epsilon_n F_\omega + (\de_\omega(\lambda\epsilon_n)+\lambda\sigma)(\omega-\omega_0) \right)\frac{\delta}{\delta\Omega}
    \end{align*}

The vector fields $\mathbb X_\alpha$ also change according to the redefinitions of the fields. Specifically, using constraints again the transversal torsionless, Einstein and Bianchi identities, we obtain
    \begin{align}
         \nonumber&\mathbb{J}_{c,E}=[c,E] \quad&&\mathbb{J}_{c,\Omega}=e\de_{\omega_0} c+[c,\Omega]\\
         \nonumber&\mathbb{E}_{\zeta,E}=-\mathrm{L}_\zeta^{\omega_0}E \quad&&\mathbb{E}_{\zeta,\Omega}=-e \iota_\zeta F_{\omega} +\de_\omega\iota_\zeta \Omega -[\iota_\zeta(\omega-\omega_0),e](\omega-\omega_0)\\
        &\label{newdef ham v.f.} \mathbb{F}_{\lambda,E}=\de_\omega(\lambda \epsilon_ne)\quad&&\mathbb{F}_{\lambda,\Omega}=\lambda \epsilon_n F_\omega+ \de_\omega(\lambda \en)(\omega-\omega_0)\quad&&\\
         \nonumber&\mathbb{K}_{\eta,E}=\de_\omega(\eta \emm e) && \mathbb{K}_{\eta,\Omega}=-\de_\omega(\eta\Omega_m)-\eta\emm \de_\omega(\omega-\omega_0)+\eta \emm F_\omega
    \end{align}
\begin{remark}
    Regarding the vector fields along $\omega_m$, we have naively $\mathbb X_{\omega_m}=  \mathrm{L}_{\mathbb X}(\omega-\omega_0)_m=0$. However, we can define $\Omega_m\coloneqq \emm (\omega-\omega_0) + e(\omega-\omega_0)_m=\emm (\omega-\omega_0)$, and see that the vector fields along the transversal fields $\Omega_m $ and $\emm$ are 
    \begin{align*}
        &\mathbb{J}_{c,\emm}=[c,\emm] \quad &&\mathbb{J}_{c,\Omega_m}=e\de_{\omega^0_m} c +\emm\de_{\omega_0}e +[c,\Omega_m] \\ 
        &\mathbb{E}_{\zeta,e_m}=\mathrm{L}_\zeta^{\omega_0}e_m+\iota_{\partial_m\zeta}e \quad &&  \mathbb{E}_{\lambda,\Omega_m}=\de_\omega(\iota_\zeta \Omega_m))+\de_{\omega_m}(\iota_\zeta \Omega)-\iota_\zeta e F_{\omega_m} - [\iota_\zeta(\omega-\omega_0),\emm](\omega-\omega_0) \\
        &\mathbb{F}_{\lambda,e_m}=\de_{\omega_m}(\lambda \epsilon_n) \quad &&\mathbb{F}_{\lambda,\Omega_m}=\de_{\omega_m}(\lambda \en)(\omega-\omega_0) +\lambda\en F_{\omega_m}\\
        &\mathbb{K}_{\eta,\epsilon_m}=\de_{\omega^0_m}(\eta \epsilon_m) \quad && \mathbb{K}_{\eta,\Omega_m }=\eta \emm F_{\omega_m} +\de_{\omega_m}(\eta \Omega_m)  - \eta \emm\de_{\omega_m}(\omega-\omega_0).
    \end{align*}
\end{remark}

\begin{proposition}
    Let $ D$ be the distribution defined by 
    \begin{equation}\label{def D ridotta}
        D\coloneqq \mathrm{span}_{\mathcal{C^1}(P_\Gamma)}(\mathbb X_\alpha + \mathcal{X}_\alpha)\subset \Gamma((T\oplus T^*)[1]P_\Gamma), \qquad \mathrm{with} \ \alpha=c,\zeta,\lambda,\eta.
    \end{equation} 
    Then, $D$ is involutive and isotropic. 
\end{proposition}
\begin{proof}
    The proof follows, mutatis mutandis, from the proof of proposition \ref{prop: involutivity and isotropicity}. 
\end{proof}

\subsection{The Poisson structure on $P_\Gamma$}

\begin{theorem}\label{thm: Poisson struct}
    Let $ D$ be the distribution defined by \cref{def D ridotta}.
    Then $ D$ is a Dirac structure, obtained as the graph of the Poisson structure $(P_\Gamma, \pi)$, where $\pi\in \Gamma(\wedge^2T P_\Gamma)$ is the Poisson bivector field given by
    \begin{align}\label{Poisson bivector}
        \nonumber\pi&=\intc \frac12E\left[\pard{}{E},\pard{}{E}\right]-\pard{}{\Omega}\left(e\de_{\omega_0}\pard{}{E}+\left[\Omega,\pard{}{E}\right]\right)+\left[\pard{}{E},\epsilon_m\right]\pard{}{\epsilon_m}\\
        &\phantom{=\intc} +\pard{}{\Omega_m}\left(e\de_{\omega_m^0}\pard{}{E}+\left[\Omega_m,\pard{}{E}\right]+\epsilon_m\de_{\omega_0}\pard{}{E}\right)+\frac12(F_\omega+\Theta)\pard{}{\Omega}\pard{}{\Omega}\\
        \nonumber&\phantom{=\intc}+(F_{  \omega_m}+\Theta_m)\pard{}{\Omega}\pard{}{\Omega_m}+\left(\de_{\omega_m}\pard{}{\Omega}-\left[\pard{}{\Omega}r,\epsilon_m\right]\right)\pard{}{\epsilon_m}\\
        \nonumber&\phantom{=\intc}-W_e^{-1}\left[\pard{}{\Omega}\epsilon_m\epsilon_nb,\epsilon_m\right]\pard{}{\epsilon_m},
    \end{align}
\end{theorem}
where
\begin{align*}
    &\Theta=-[e,\omega-\omega_0]r-\de_\omega(\omega-\omega_0),
    &\Theta_m=- [\emm,\omega-\omega_0] r-\de_\omega(\omega-\omega_0)_m -\de_{\omega_m}(\omega-\omega_0),
\end{align*}
and having defined $\Omega_m\coloneqq \emm (\omega-\omega_0)$.
\begin{proof}
We work with representatives of the equivalence classes fixed by \cref{fixomega_m}, where it holds that $\Theta_m=- [\emm,\omega-\omega_0] r -\de_{\omega_m}(\omega-\omega_0)$. Since we have already shown that $D$ is both involutive and isotropic, if we can find a bivector $\pi\in\Gamma(\wedge^2 P_\Gamma)$ such that
    \[
    D=e^\pi(\Omega^1[1](P_\Gamma))=\{ \iota_{\mathcal{X}}\pi +\mathcal{X} \ | \ \mathcal{X}\in\Omega
    ^1[1](P_\Gamma)\},
    \]
we are guaranteed to have maximality as well, from which it follows that $D$ is Dirac. Moreover, by Example \ref{ex: Dirac structure induced by a Poisson bivector} we have, as a consequence of involutivity, that $[\pi,\pi]=0$, showing that $\pi $ is Poisson.

In order to show that $D=e^\pi(\Omega^1[1](P_\Gamma))$, we start by noticing that 1-shifted one--forms on $P_\Gamma$ are generated by 
    \[
    A\coloneqq \intc a \delta E, \qquad \mathrm{and} \qquad B\coloneqq \intc b\delta \Omega,
    \]
with $a\in\Gamma[1](P_\Gamma,\Omega^{0,2}_\Gamma)$ and $b\in\Gamma[1](P_\Gamma,\Omega^{0,1}_\Gamma) $. However, we can always decompose
    \begin{align*}
        b=\iota_\varsigma e + \nu \epsilon_n +\tau \epsilon_m,
    \end{align*}
with $\nu,\tau\in \mathcal{C^1}[1](P_\Gamma)$ and $\varsigma\in \Gamma[1](P_\Gamma,\mathfrak{X}(\Gamma)) $, obtaining
    \[
    B=\mathcal{E}_\varsigma - \mathcal{J}_{\iota_\varsigma(\omega-\omega_0)} - \mathcal{J}_{\iota_\varsigma e r}  -\mathcal{J}_{W_e^{-1}(\iota_\varsigma e\epsilon_n \epsilon_m b )} + \mathcal{F}_\nu + \mathcal{K}_\tau ,
    \]
which, after noticing that $A=\mathcal{J}_a$, implies that 
    \[
    \Omega^1(P_\Gamma)=\mathrm{span}_{\mathcal{C^1}(P_\Gamma)}( \mathcal{X}_\alpha) \qquad \mathrm{with} \ \alpha=a,\varsigma,\nu,\tau.
    \]
It follows that it suffices to check that $\pi(\mathcal{X}_\alpha)=\iota_{\mathcal{X}_\alpha}\pi =\mathbb X_\alpha$ for all $\alpha$'s.
We leave this for appendix \ref{app: proof poisson}.

\end{proof}

\section{The BF$^2$V structure of PC gravity}\label{sec: BF2V}

Given a poisson structure on a manifold $P$, a BF$^2$V structure is automatically obtained by promoting the Poisson bivector to a function on the 1-shifted tangent bundle of $P$. In the Palatini--Cartan case, such Poisson bivector depends, in principle, only on the fields $E$ and $\Omega$ in the reduced space of corner fields. However, the strict dependency on them is not manifest in \cref{Poisson bivector}, as it contains terms proportional to $\frac{\delta}{\delta\Omega_m}$ and $\frac{\delta}{\delta \epsilon_m}$, which have not been expressed explicitly in terms of the two dynamical fields $E$ and $\Omega$.

In order to obtain a BF$^2$V action of PC gravity, we set $\mathcal{F}_\Gamma\coloneqq T[1]P_{\Gamma}$. Since $P_\Gamma$ is decorated with the transversal fields $\Omega_m$ and $\emm$ ---which have been fixed to obtain $P_\Gamma$ as a submanifold of the original space of corner fields $F_\Gamma$---  the terms proportional to $\frac{\delta}{\delta\Omega_m}$ and $\frac{\delta}{\delta \epsilon_m}$ in the BF$^2$V action require introducing the corresponding conjugate fields (in the fiber of $T[1]$)  and fix them such that they can only depend on $E$, $\Omega$ and their conjugate fields.

\begin{remark}
    While $\Omega\in\og^{2,3} $ is not subject to any constraint, $E\in \og^{2,2}$ is required to have a vanishing Pfaffian. Such constraint will be reflected also on its conjugate field in $T[1]\og^{2,2}$.
\end{remark}

We introduce the following
    \[
    \mathcal{F}_\Gamma\coloneqq \{ \ (E,\Omega,\emm,\Omega_m,c,\mu,\nu^\dag,f^\dag) \ | \ c\in\Omega_{\Gamma,r}^{0,2}, \ \mu\in\og^{0,1} , \ f^\dag,\nu^\dag \text{ fixed}  \}.
    \]
We will shortly see how $f^\dag\in\og^{2,3}[1] $ and $\nu^\dag\in\og^{1,2} $ are fixed in terms of $(E,\Omega,c,\mu)$. First, we have to define $\Omega_{\Gamma,r}^{0,2}$ as the fiber of $T[1]\Omega_{\Gamma,r}^{1,1} $. In order to do so, we consider the symplectic  form on $T[1](\og^{1,1}\oplus \og^{2,2})$ 
    \begin{equation}\label{eq: presympl BF2V}
        \tilde{\varpi}_\Gamma\coloneqq \int_\Gamma \delta E \delta c + \delta \Omega \delta \mu .
    \end{equation}
After setting $E=\frac{e^2}{2}$, we notice that that $\tilde{\varpi}_\Gamma$ becomes pre-symplectic, indeed
    \[
    \Ker{\tilde{\varpi}_\Gamma}=\bigg\{ \int_\Gamma \mathbb{X}_c \frac{\delta}{\delta c} \in \mathfrak{X}(T[1](\og^{1,1}\oplus \og^{2,2})) \ | \ e \mathbb X_c =0 \bigg\}.
    \]
Therefore, in order to have a symplectic form on $T[1]P_\Gamma$, we need to consider only those elements $[c]$ living in the quotient $\Omega_{\Gamma,r}^{0,2}[1]\coloneqq \og^{2,2}[1]/\Ker{W_e^{\Gamma,(0,2)}} $. 

\begin{lemma}
    There exists a unique representative of $c\in[c]\in\Omega_{\Gamma,r}^{0,2}[1]$ such that 
    \begin{equation}
        \emm\en c =0. 
    \end{equation}
\end{lemma}
\begin{proof}
    We know, essentially by construction, that $\Ima(W_e^{\Gamma,(0,2)})[1]\simeq \og^{2,2}[1]/\Ker{W_e^{\Gamma,(0,2)}}$. Indeed, defining $C\coloneqq e c\in \Ima(W_e^{\Gamma,(0,2)})[1]$, have $C=e c'$ if and only if $c'-c \in \Ker{W_e^{\Gamma,(0,2)}}$. However, by Lemma \ref{lem: decomp (1,3)}, any $K\in \og^{1,3}$ splits uniquely as 
        \[
        K= e\alpha + \epsilon_m\en \beta \qquad \text{ s.t.}\quad  \emm \en \alpha =0 \quad \mathrm{and} \quad e\beta=0.
        \]
    This implies that for any $C\in\Ima(W_e^{\Gamma,(0,2)})[1]\subset \og^{1,3}[1] $ there exists a unique $c\in\og^{0,2}[1] $ such that 
        \[
        C= ec \qquad \mathrm{and} \qquad \emm\en c =0.
        \]
\end{proof}

Naively, the BF$^2V$ action reads
    \begin{align}
        \nonumber\mathcal{S}_\Gamma&=\intc \frac12E[c,c]-\mu(e\de_{\omega_0}c+[\Omega,c])+[c,\epsilon_m]f^\dag\\
        &\phantom{=\intc} +\nu^\dag(e\de_{\omega_m^0}c+[\Omega_m,c]+\epsilon_m\de_{\omega_0}c)+\frac12(F_\omega+\Theta)\mu^2\\
        \nonumber&\phantom{=\intc}+(F_{  \omega_m}+\Theta_m)\mu\nu^\dag+(\de_{\omega_m}\mu-[\mu r,\epsilon_m])f^\dag\\
        \nonumber&\phantom{=\intc}-W_e^{-1}[\mu\epsilon_m\epsilon_nb,\epsilon_m]f^\dag,
    \end{align}
while the symplectic form is
    \begin{equation}
        \varpi_\Gamma=\intc \delta E \delta c + \delta\Omega \delta \mu +\delta \Omega_m \delta\nu^\dag + \delta\emm \delta f^\dag.,
    \end{equation}
stressing that $\delta \Omega_m \delta\nu^\dag + \delta\emm \delta f^\dag$ is going to depend just on the other fields. 

The above expressions are quite cumbersome. Specifically, we would like to find a symplectomorphism that symplifies both the action and the symplectic form. We start by noticing that, due to our own choice, we have $\Omega_m = \emm (\omega-\omega_0)$. Hence, we obtain
    \begin{align*}
        \delta\nu^\dag\delta\Omega_m &= \delta\nu^\dag (\delta\emm (\omega-\omega_0) - \emm \delta(\omega-\omega_0))\\
        &=  \delta(\nu^\dag(\omega-\omega_0))\delta\emm -  \delta(\nu^\dag\emm)\delta\omega.
    \end{align*}
Similarly, we have
    \begin{align*}
        \delta\mu\delta\Omega&= \delta\mu (\delta e (\omega-\omega_0) + e \delta(\omega-\omega_0))\\
        &= - \delta(\mu(\omega-\omega_0)) \delta e + \delta M\delta\omega,
    \end{align*}
having defined $M\coloneqq e\mu$. The symplectic form is then
    \begin{align*}
        \varpi_\Gamma=\intc &\delta (ec - \mu(\omega-\omega_0)) + \delta(M-\nu^\dag \emm)\delta\omega + \delta(f^\dag + \nu^\dag(\omega-\omega_0))\delta\emm.
    \end{align*}
This suggests fixing $f^\dag= -\nu^\dag(\omega-\omega_0)$, while $\nu^\dag= e \gamma^\dag$, with $\gamma^\dag=[\emm,\mu]$. Furthermore, consider the field redefinition $\Psi\colon T[1]P_\Gamma \to  T[1]P_\Gamma $, such that it acts as the identity on all the fields except $\mu$ and $c$, where
    \begin{align*}
        &\Psi^*(\mu)= \mu + \emm \gamma^\dag\\
        &\Psi^*(c)= c - W_e^{-1}(\mu (\omega-\omega_0)).
    \end{align*}
\begin{remark}
    Notice that $\mu\in\og^{0,1}$, can be rewritten as $\mu = \iz e + \lambda \en + \eta\emm$. This implies that the ghosts are the conjugate momenta to the fields in the BF$^2$V structure of PC gravity. 
\end{remark}
\begin{proposition}
    With the above (re-)definition of the fields, the data
        \begin{equation}
            \varpi_\Gamma'=\Psi^*(\varpi_\Gamma)=\intc\delta c \delta E + \delta\mu e\delta\omega + \mu \delta e \delta\omega,
        \end{equation}
    and 
        \begin{equation}\label{eq: easy BF2V action}
            \mathcal{S}'_\Gamma=\Psi^*(\mathcal{S}_\Gamma)=\intc \frac{1}{2}E [c,c] + \mu e \de_\omega c + \frac{1}{2}\mu^2 F_\omega
        \end{equation}
    define a BF$^2$V structure on $\mathcal{F}_\Gamma$.
\end{proposition}
\begin{proof}
    Since the full computation of the pullback of $\mathcal{S}_\Gamma$ along $\Psi$ is quite involved, we prefer to omit it and to prove that the resulting data is that of a genuine BF$^2$V structure on $\mathcal{F}_\Gamma$ by means of the classical master equation. 

    We start by finding the Hamiltonian vector field $\mathbb Q_\Gamma$ of $\mathcal{S}'_\Gamma$, which, after some manipulations, is simply obtained as
        \begin{align}
            &\mathbb Q_{\Gamma,e}= [c,e] + \de_\omega \mu   & e \mathbb Q_{\Gamma,\omega} = e \de_\omega c + \mu F_\omega\\
            \nonumber&e\mathbb Q_{\Gamma,c} = \frac{1}{2}e [c,c] -\mu W_{e}^{-1}(\mu F_\omega) &\mathbb Q_{\Gamma,\mu}=[c,\mu].
        \end{align}
    We remark that $\mathbb Q_{\Gamma,e}$ is defined up to terms in $\ker(W_e^{\Gamma(1,1)})$. The same is true for $\mathbb Q_{\Gamma,c}$ and $\mathbb Q_{\Gamma,\omega}$. It is interesting to consider the expression $\mu F_\omega$. Indeed, using Einstein's equation and its transversal components, summarized in Table \ref{tab: bdry corner equations}, it can be rewritten as
        \begin{align*}
            \mu F_\omega = & (\iz e + \lambda \en + \eta \emm) F_\omega \\
            =& e(- \iz F_\omega + \lambda F_{\omega_n} + \eta F_{\omega_m}),
        \end{align*}
    hence 
        \begin{align*}
            &\mathbb Q_{\Gamma,e}= [c,e] + \de_\omega \mu   &  \mathbb Q_{\Gamma,\omega} =  \de_\omega c - \iz F_\omega + \lambda F_{\omega_n} + \eta F_{\omega_m}\\
            &\mathbb Q_{\Gamma,\mu}=[c,\mu] &e\mathbb Q_{\Gamma,c} = \frac{1}{2}e [c,c] +\mu ( \iz F_\omega - \lambda F_{\omega_n} - \eta F_{\omega_m}) .
        \end{align*}
    For the following computations, it is also convenient to compute $W_e^{-2}(\mu^2F_\omega)=W_e^{-1}(\mu W_e^{-1}(\mu F_\omega))$. We find, using the transversal Bianchi identities from table \ref{tab: bdry corner equations},
        \begin{align*}
            \mu ^2 F_\omega =& \iz e \iz e F_\omega + 2 \iz e (\lambda \en + \eta\emm) F_\omega + 2 \eta\lambda \en\emm F_\omega\\
            =& -\iz E \iz F_\omega - 2 e \lambda \en \iz F_\omega - 2 e \eta \emm \iz F_\omega\\
            =& E \iz \iz F_\omega-2e(\lambda\iz (eF_{\omega_n} ) + \eta \iz (e F_{\omega_m}))\\
            =&E \iz \iz F_\omega -2 E (\lambda \iz F_{\omega_n} + \eta \iz F_{\omega_m}).
        \end{align*}\
    having used $2\en\emm F_\omega= \en F_{\omega_m} - \emm F_{\omega_n}=0$. We can then infer
        \begin{equation}
            W_e^{-2}(\mu^2F_\omega)=\frac{1}{2}\iz \iz F_\omega - \lambda \iz F_{\omega_n} - \eta \iz F_{\omega_m}.
        \end{equation}
    The classical master equation, in absence of vertices, is equivalent to the cohomological property of the Hamiltonian vector field $\mathbb Q_\Gamma$. Hence, we just need to check $\mathbb Q_\Gamma^2=0$. We have
        \begin{align*}
            \mathbb Q^2_{\Gamma,\mu}= &\frac{1}{2}[[c,c],\mu] - [W_e^{-2}(\mu^2F_\omega),\mu] - e [c,[c,\mu]]\\
            =& - [W_e^{-2}(\mu^2F_\omega),\mu]\\
            =& [\iz e + \lambda \en + \eta \emm, \frac{1}{2}\iz \iz F_\omega - \lambda \iz F_{\omega_n} - \eta \iz F_{\omega_m} ]\\
            =& \frac{1}{2}\iz\iz\iz[e,F_\omega] - \frac{1}{2}\iz\iz [e,\lambda F_{\omega_n} - \eta F_{\omega_m}] \\
            &+ [\lambda\en , \frac{1}{2}\iz \iz F_\omega - \lambda \iz F_{\omega_n} - \eta \iz F_{\omega_m} ]\\
            &+[\eta\emm , \frac{1}{2}\iz \iz F_\omega - \lambda \iz F_{\omega_n} - \eta \iz F_{\omega_m} ]\\
            =&- \frac{1}{2}\iz\iz (\lambda [\en,F_\omega] +\eta [\emm,F_\omega] )\\
            &+ [\lambda\en , \frac{1}{2}\iz \iz F_\omega ] +\lambda\eta \iz[\en,F_{\omega_m}]\\
            &+[\eta\emm , \frac{1}{2}\iz \iz F_\omega ] - \lambda\eta \iz[\emm,F_{\omega_n}]\\
            =&0,
        \end{align*}
    having used the graded Jacobi identity and the transversal Bianchi identities. In a similar way, one can also show that 
        \begin{align}
            &[F_\omega,\mu]= [W_e^{-1}(\mu F_\omega),e] \label{id: Bianchi mu}\\
            &[W_e^{-1}(\mu F_\omega),\mu]=[W_e^{-2}(\mu^2 F_\omega),e].\label{id: Bianchi mu 2}
        \end{align}
        \begin{align*}
            e\mathbb Q_{\Gamma,\omega}^2=&\mathbb Q_\Gamma(e\mathbb   Q_{\Gamma\omega}) - (\mathbb Q_{\Gamma,e})(\mathbb Q_{\Gamma,\omega})\\
            =&  e [\de_\omega c,c] + e [W_e^{-1}(\mu F_\omega),c ] - \frac{1}{2} e\de_\omega[c,c]+ e \de_\omega(W_e^{-2}(\mu^2 F_\omega)  )+ [c,\mu] F_\omega \\
            & - \mu \de_\omega(\de_\omega c + W_e^{-1}(\mu F_\omega)) - [c,e]W_e^{-1}(\mu F_\omega) - \de_\omega\mu W_e^{-1}(\mu F_\omega)\\
            =& [\mu F_\omega,c ] + [c,e] W_e^{-1}(\mu F_\omega) - \de_\omega(\mu W_e^{-1}(\mu F_\omega)) + [c,\mu ]F_\omega\\
            &-\mu [F_\omega,c ] - \mu \de_\omega( W_e^{-1}(\mu F_\omega)) - [c,e] W_e^{-1}(\mu F_\omega) - \de_\omega\mu  W_e^{-1}(\mu F_\omega)\\
            =&0 ,
        \end{align*}        
        \begin{align*}
             \mathbb Q_{\Gamma,e}^2=& \frac{1}{2}[[c,c],e] - [ W_e^{-2}(\mu^2 F_\omega),e] - [c,[c,e]] - [c,\de_\omega\mu]\\
             &+ [\de_\omega c,\mu] + [ W_e^{-1}(\mu F_\omega),\mu] - \de_\omega[c,\mu]\\
             =&-[ W_e^{-2}(\mu^2 F_\omega),e]+ [ W_e^{-1}(\mu F_\omega),\mu]\\
             \overset{\cref{id: Bianchi mu 2}}{=}&0 ,
        \end{align*}

    Lastly, we need to check that the remaining constraint on the fields (which is not a ``structural constraint'' used to fix the values of the transversal fields and jets) is invariant under the action of $\mathbb Q_\Gamma$. One finds
        \begin{align*}
            \mathbb Q_\Gamma (\de_\omega e ) &= [\de_\omega c + W_e^{-1}(\mu F_\omega), e ] - \de_\omega([c,e] + \de_\omega\mu)\\
            &=[W_e^{-1}(\mu F_\omega), e ] - [F_\omega,\mu] + [c,\de_\omega e]\\
            &\overset{\cref{id: Bianchi mu}}{=} [c,\de_\omega e],
        \end{align*}
    which vanishes on-shell. Similarly, one can check the transversal Bianchi identities to be preserved by the action of $\mathbb Q_\Gamma$.
\end{proof}
\begin{remark}
    As one can check by quickly inspecting \cref{eq: easy BF2V action}, the BF$^2$V action contains quadratic terms in the classical fields $e$ and $\omega$. However, the wish is to be able to rewrite it---through yet another field redefinition---such that it only contains  linear terms in the classical fields. The BF$^2$V structure would then correspond to a linear (or an affine) Poisson structure on the original space of reduced corner fields $P_\Gamma$.
\end{remark}
\begin{proposition}
    Let $\Phi\colon T[1]P_\Gamma\to T[1]P_\Gamma $ be the diffeomorphism acting as the identity on all the fields except on $c$, where
    $$\Phi^*(c)=c-\iota_\zeta (\omega-\omega_0).$$
    Then, we have
        \begin{equation}
            \mathcal{S}^{PC}_\Gamma\coloneqq \Phi^*(S'_\Gamma)=\intc \frac{1}{2}E[c,c] + \mu e\de_{\omega_0}c +\frac{1}{2} \mu^2 F_{\omega_0}, 
        \end{equation}
    while the symplectic form becomes
        \begin{equation}
            \varpi^{PC}_\Gamma\coloneqq \Phi^*(\varpi'_\Gamma)=\intc \delta E \delta(c - \iz (\omega-\omega_0))+ \delta\mu e\delta\omega + \mu \delta e \delta\omega.
        \end{equation}
    \end{proposition}
 \begin{proof}
    The proof is just the result of a quite long but simple computation, which we omit for the sake of readability. 
 \end{proof}
 \begin{remark}
     The above BF$^2$V data correspond to the following affine Poisson structure
        \begin{equation}
            \pi^{PC}\coloneqq \intc \frac{1}{2}E\big[\frac{\delta }{\delta E},\frac{\delta}{\delta E}\big] + \frac{\delta}{\delta\Omega }e\de_{\omega_0}\frac{\delta}{\delta E} + \frac{1}{2}\frac{\delta}{\delta \Omega}\frac{\delta}{\delta \Omega} F_{\omega_0}.
        \end{equation}
    Furthermore, notice that such action closely resembles the one found in the $BF$ theory of Example \ref{example: BF theory}, after setting $F_{\omega_0}=0$. Specifically, one can recover the BF$^2$V structure of PC gravity as a particular $BF$ theory for the Lie algebra $\mathfrak{so}(3,1)$, restricted to the submanifold defined by the following field parametrisation
        \begin{align*}
            &c=c  & A=\omega  & &B = E,\\
            & \phi = \frac{1}{2}\mu^2 & \tau = e\mu  && B^\dag =0.
        \end{align*}
 \end{remark}
%\subsection{Quantization and outlook}
\section{Conclusion and Outlook}

In this paper, we investigated the corner structure of four-dimensional Palatini--Cartan gravity on manifolds with corners from a classical perspective. Starting from the boundary constraint algebra, we associated to the codimension-two strata a pre-Dirac structure on the space of corner fields. Using techniques from generalized geometry, we identified a natural reduction procedure that removes the degeneracies responsible for the failure of maximality and yields a genuine Dirac structure on a reduced corner phase space. We further showed that this Dirac structure can be identified with the graph of a Poisson bivector, thereby providing a geometric characterization of the reduced corner degrees of freedom.

The resulting Poisson geometry naturally gives rise to a strict BF$^2$V structure. In this way, the corner theory is obtained directly from the reduced classical geometry, in close analogy with the construction of the BFV theory for Palatini--Cartan gravity developed in \cite{CCS2020}. This should be contrasted with the BFV-induced construction of \cite{CC2022cor}, where the codimension-two theory is obtained by descending from the boundary BFV data and is naturally described by a local $P_\infty$-algebra. There, the induced presymplectic structure remains singular and the resulting theory is pre-BF$^2$V in nature. By performing the relevant reduction before introducing antifields, the construction presented here removes this obstruction and yields a genuine BF$^2$V theory associated with the reduced corner phase space.

A second outcome of our analysis is the identification of an equivalent affine Poisson description of the reduced corner theory. This reformulation reveals a close relation with the corner structure of four-dimensional $BF$ theory and places the gravitational corner degrees of freedom in a setting that is particularly well suited for quantization. Indeed, affine Poisson structures naturally determine central extensions of Lie algebras through the Kirillov--Kostant construction. Consequently, part of the quantization procedure is already encoded in the classical geometry itself: the affine term determines the extension data that govern the corresponding quantum corner algebra.

This observation suggests a possible connection with the quantization program developed for corners in four-dimensional $BF$ theory. In that setting, the affine Poisson structure gives rise to a centrally extended infinite-dimensional Lie algebra whose universal enveloping algebra defines the free corner algebra. The physical corner algebra is subsequently obtained by imposing the quantum counterparts of the classical constraints. While this construction successfully produces a rich class of representations of the free corner algebra, the passage to the physical algebra turns out to be considerably more subtle. In the non-abelian case, the constraint ideal acts non-trivially on the available modules, and the resulting representations descend only trivially to the physical corner algebra. Thus, the main challenge lies not in the construction of the free quantum corner algebra itself, but rather in understanding the physical quotient and its representation theory.

The affine Poisson structure obtained in the present work suggests that analogous questions can now be formulated for gravitational corners. In particular, one may hope to associate to the reduced corner theory a gravitational corner algebra obtained from the quantization of the corresponding central extension and to study its representations as candidate quantum corner state spaces. Whether non-trivial physical representations survive the implementation of the gravitational constraints remains a completely open question. Nevertheless, the framework developed here provides a concrete geometric and algebraic setting in which such problems can be systematically investigated.

Another natural direction concerns the relation between the BF$^2$V structure constructed here and the BFV theory of Palatini--Cartan gravity. Since our approach follows the philosophy of \cite{CCS2020}, it is natural to ask whether there exists a boundary BFV theory whose codimension-two descendant automatically reproduces the strict BF$^2$V structure obtained in this paper. Such a construction would provide a genuine BF$^{(k)}$V description of Palatini--Cartan gravity compatible with the reduced corner geometry.

A particularly intriguing possibility is that such a BFV theory may be related to the one constructed in \cite{CCS2020} through a suitable partial coisotropic reduction. If so, the reduced corner theory described here could emerge naturally from a reduced BFV structure already at codimension one. Clarifying the relation between these constructions would provide a better understanding of the interplay between reduction and the BF$^{(k)}$V descent procedure and could ultimately lead to a unified description of bulk, boundary, and corner structures in Palatini--Cartan gravity.

More broadly, the appearance of Dirac and affine Poisson geometry at codimension two suggests that generalized geometric structures may play a fundamental role in the organization of higher-codimensional data in gauge theories. It would therefore be interesting to investigate whether similar mechanisms arise in other gauge theories admitting BF$^{(k)}$V descriptions, and whether the affine Poisson structures obtained in this way admit a systematic quantization compatible with gluing and locality principles. We hope that the present work provides a first step in this direction.

\newpage
\appendix
\section{Useful results and notation}
The following section contains lemmas and identities useful throughout the paper. They are regrouped here for convenience. Considering inclusions $i\colon \Sigma\hookrightarrow M$ and $j\colon \Gamma \hookrightarrow\Sigma $, by abuse of notation we will denote by $\mathcal{V}$ also the induced vector bundles on the boudnary and corners $i^* \mathcal{V}$ and $(i\circ j)^*\mathcal{V}$. First, we recall the definitions
    \begin{align}
        \Omega^{i,j}\coloneqq \Omega^i(M,\wedge^j V) ,\qquad \Omega_\Sigma^{i,j}\coloneqq \Omega^i(\Sigma,\wedge^j V) ,\qquad \Omega^{i,j}_\Gamma\coloneqq \Omega^i(\Gamma,\wedge^j V) ,
    \end{align}
    and maps $W_k^{i,j}\coloneqq e^k\wedge -$, where $e$ is understood to be the vielbein on $M$, $\Sigma$ or $\Gamma$ depending on the context, 
    \begin{align*}
        W_k^{i,j}\ \colon \ & \Omega^{i,j}\longrightarrow \Omega^{i+k,j+k}  &W_k^{\Sigma(i,j)} \colon & \ \Omega_\Sigma^{i,j}\longrightarrow \Omega_\Sigma^{i+k,j+k}   & W_k^{\Gamma(i,j)} \colon & \ \Omega_\Gamma^{i,j}\longrightarrow \Omega_\Gamma^{i+k,j+k} .
    \end{align*}
We also define $\varrho^{i,j}_{(\Gamma,\Sigma)}\ \colon \ \Omega^{i,j}_{(\Gamma,\Sigma)} \to \Omega^{i+1,j-1}_{(\Gamma,\Sigma)} \ \colon \ \alpha \mapsto [e,\alpha]  $. 

Additionally, defining $\epsilon_m\in\Omega^{0,1}_\Gamma $ such that $\{e_i,\epsilon_m,\epsilon_n\}$ is a local basis of $\mathcal{V}$, with $e_i$ is such that $e=e_i \de x^i$ is the corner vielbein, we have the following map
    \begin{align}
        W_{mn}^{i,j}\colon  \Omega_\Gamma^{i,j}&\longrightarrow \Omega_\Gamma^{i,j+2}\\
        \alpha & \longmapsto \epsilon_m\epsilon_n \alpha.
    \end{align}
With the above definitions, we have the following theorem
\begin{theorem}\cite{Canepa:2024rib}\label{thm: coframes prop}
    Letting $l$ be the codimension of the manifold we're working on, i.e. $l(M)=0$, $l(\Sigma)=1$ and $l(\Gamma)=2$. Then $W_k^{l,(i,j)}$ is surjective if and only if
        \[i+j\geq  4-k\]
    and injective if and only if
    \[
    i+j\leq 4- l-k.
    \]
\end{theorem}

In the following, we will make extensive use of the following manifolds
    \begin{align*}
        &\Sigma_m\coloneq \Gamma\times[0,1] \ni (x,x^m), \qquad \text{and}\qquad
        M_{mn}\coloneqq \Gamma \times [0,1]^2\ni (x,x^m,x^n).
    \end{align*}

\begin{lemma}\cite{CCS2020}\label{lem:Omega2,1_d4}
Let $\alpha \in \Omega^{2,1}_\Sigma$. Then 
\begin{align}\label{ConditionforOmega21_d4}
\alpha=0 \qquad \Longleftrightarrow  \qquad \begin{cases}
e\alpha =0 \\
\epsilon_n\alpha \in \Ima W_{1}^{\Sigma, (1,1)}
\end{cases}
\end{align}
where the last condition implies $\epsilon_n\alpha = e\sigma$ for some $\sigma\in\Omega_\Sigma^{1,1} $.
\end{lemma}

\begin{lemma}\label{lem: (2,1) constraint vanish}
Let $\alpha \in \Omega^{2,1}_\Gamma$. Then 
\begin{align*}
\alpha=0 \qquad \Longleftrightarrow  \qquad \begin{cases}
\epsilon_n\alpha=e\sigma\\
\epsilon_m \alpha= e\rho \\
\epsilon_n \rho + \epsilon_m \sigma \in \Ima W_1^{\Gamma,(0,1)}
\end{cases}
\end{align*}
for some $\sigma,\rho\in\Omega_\Gamma^{1,1}$. Additionally, we also obtain $\sigma=\rho=0$
\end{lemma}
\begin{proof}
    Consider $M_{mn}\coloneqq \Gamma \times [0,1]^2\ni (x,x^m,x^n)$, and define $E\coloneqq e+ \epsilon_n \de x^n + \epsilon_m \de x^m$ vielbein on $M_{mn}$. Let $\Omega^{2,1}\ni A=\alpha + \sigma \de x^n + \rho \de x^m + \gamma \de x^m \de x^n$, with $\alpha\in \Omega^{2,1}_\Gamma$, $\sigma,\rho\in\Omega^{1,1}_\Gamma$ and $\gamma\in\Omega^{0,1}_\Gamma$. By Theorem \ref{thm: coframes prop} we have that 
    \[ A=0 \ \Leftrightarrow \ E\wedge A = 0.\]
    Expanding the latter equation, we obtain 4 equations
        \[
        \begin{cases}
            e\alpha =0 \\
            \epsilon_n \alpha = e\sigma\\
            \emm \alpha = e\rho\\
            \epsilon_n \rho + \epsilon_m \sigma = e \gamma.
        \end{cases}
        \]
    Notice that $e\alpha=0$ identically as $\Omega^3(\Gamma)=\{0\}$. Since the above system of equations is equivalent to $A=0$, this shows the lemma.
\end{proof}
\begin{lemma}\label{lem: (0,1) vanish}
Let $r \in \Omega^{0,1}_\Gamma$. Then 
\begin{align*}
r=0 \qquad \Longleftrightarrow  \qquad \begin{cases}
\epsilon_ne^2r=0\\
\epsilon_me^2r=0\\
\epsilon_n\epsilon_mer=0
\end{cases}
\end{align*}
\end{lemma}
\begin{proof}
    Considering $M_{mn}$ and $E\coloneqq e+ \epsilon_n \de x^n + \epsilon_m \de x^m$ as before, letting $\Omega^{0,1}\ni R=r$, we know from Theorem \ref{thm: coframes prop} that $E^3\wedge R=0$ if and only if  $R=0$. Expanding in each component of $E$, we obtain the desired system of equations. 
\end{proof}

\begin{lemma}\label{lem: (1,1) vanish}
Let $\beta\in \Omega^{1,1}_\Gamma$. Then 
\begin{align*}
\beta=0 \qquad \Longleftrightarrow  \qquad \begin{cases}
\epsilon_ne\beta=\frac{e^2}{2}\beta_n\\
\epsilon_me\beta=\frac{e^2}{2}\beta_m\\
\epsilon_n \epsilon_m \beta= e(\epsilon_n \beta_m + \epsilon_m \beta_m)
\end{cases}
\end{align*}
for some $\beta_n,\beta_m\in\Omega_\Gamma^{0,1}$.
\end{lemma}
\begin{proof}
    Consider again $E\coloneqq e+ \epsilon_n \de x^n + \epsilon_m \de x^m$ coframe on $M_{mn}$ and let $\Omega^{1,1}\ni B\coloneqq \beta + \beta_m\de x^m + \beta_n \de x^n $. From Theorem \ref{thm: coframes prop}, we know $E^2\wedge B =0$ if and only if $ B=0$, hence $\beta=0$ if and only if
        \[
        \begin{cases}
            e^2 \beta=0\\
            \epsilon_ne\beta=\frac{e^2}{2}\beta_n\\
\epsilon_me\beta=\frac{e^2}{2}\beta_m\\
\epsilon_n \epsilon_m \beta= e(\epsilon_n \beta_m + \epsilon_m \beta_m),
        \end{cases}
        \]
    where $e^2\beta=0$ identically. We therefore obtain the statement of the lemma. 
\end{proof}

\begin{lemma}\label{lem: (1,2) vanish}
Let $v \in \Omega^{1,2}_\Gamma$. Then 
\begin{align*}
v=0 \qquad \Longleftrightarrow  \qquad \begin{cases}
ev=0\\
\epsilon_nv=ev_n\\
\epsilon_m v= ev_m \\
\epsilon_n v_m - \epsilon_mv_n=0
\end{cases}
\end{align*}
for some $v_m,v_n\in\Omega_\Gamma^{0,2}$.
\end{lemma}
\begin{proof}
    Consider again $E\coloneqq e+ \epsilon_n \de x^n + \epsilon_m \de x^m$ coframe on $M_{mn}$ and let $\Omega^{1,2}\ni V= v + v_m \de x^m + v_n \de x^n $. By theorem \ref{thm: coframes prop} $E\wedge V=0$ if and only if $V=0$. Hence, expanding the former equation, 
    \begin{align*}
v=0 \qquad \Longleftrightarrow  \qquad \begin{cases}
ev=0\\
\epsilon_nv=ev_n\\
\epsilon_m v= ev_m \\
\epsilon_n v_m - \epsilon_mv_n=0
\end{cases}
\end{align*}
which shows the lemma.
\end{proof}

\begin{lemma}\label{lem: decomp (1,3)}
    For any $\Theta\in \Omega_\Gamma^{1,3}$, there exists a unique decomposition 
    \[
    \Theta= e\alpha + \epsilon_m \epsilon_n \beta, \quad \mathrm{with} \quad \alpha\in\Omega^{0,2}_\Gamma , \ \beta\in\Omega^{1,1} 
    \]
    such that
    \[
    \epsilon_n \epsilon_m \alpha=0 \qquad \mathrm{and} \qquad e\beta=0.
    \]
\end{lemma}
\begin{proof}
    We have to show that $\Omega^{1,3}_\Gamma=e\ker(W_{mn}^{0,2})\oplus \epsilon_m\epsilon_n \ker(W_1^{\Gamma,(1,1)} ) $.
    It was proved in \cite{CC2022cor} that $\dim(\ker(W_1^{\Gamma,(1,1)} ))=3 $.\footnote{By dimension, we mean the dimension of the fiber of $\Omega^{i,j}$ at a point $x\in\Gamma$.} One can also check that $\dim \ker(W_{mn}^{0,2})=5$, as $\epsilon_m\epsilon_n \alpha=0\in\Omega^{0,4}_\Gamma $ fixes 1 component of $\alpha$, and $\dim\Omega_\Gamma^{0,2}=6 $. It follows that $\dim \Omega_\Gamma^{1,3}=8=\dim \ker(W_{mn}^{0,2}) +\dim(\ker(W_1^{\Gamma,(1,1)} ))$. 

    Now we prove that 
    \[
    \Ima \left(W_1^{\Gamma,(0,2)}\vert_{\ker(W_{mn}^{0,2})} \right) \cap \Ima\left( W_{mn}^{1,1}\vert_{\ker(W_1^{\Gamma,(1,1)}) } \right)=\{0\}.
    \]
    Consider $\beta\in\ker(W_1^{\Gamma,(1,1)})$ and $\alpha\in \ker(W_{mn}^{0,2}$. By contradiction, assume
        \[
        e\alpha= \epsilon_m\epsilon_n \beta.
        \]
    In general, one can write $\alpha=\frac{1}{2}\chi^{ij}e_i e_j + \chi^{i}_m\epsilon_m e_i + \chi^{i}_n\epsilon_n e_i + \chi\epsilon_n\epsilon_m $, with
    \[
    \chi^{ij}=\epsilon^{ij}\chi^{12}, \ \mathrm{for}\ \chi^{12}\in\mathcal{C^1}(\Gamma), \quad \chi_m,\chi_n\in\mathfrak{X}(\Gamma), \quad \chi\in\mathcal{C^1}(\Gamma).
    \]
    From $\epsilon_n\epsilon_m \alpha=0$ it follows $\alpha^{ij}=0$, hence obtaining
    \[
    e\alpha=\iota_{\chi_m}\frac{e^2}{2} \epsilon_m + \iota_{\chi_n}\frac{e^2}{2} \epsilon_n + \chi e \epsilon_m \epsilon_n.
    \]
    Imposing $e\alpha=\epsilon_m \epsilon_n \beta$ implies
    \[
    \begin{cases}
        \chi_m=\chi_n=0\\
        \epsilon_m \epsilon_n\beta=e\epsilon_m \epsilon_n \chi.
    \end{cases}
    \]
    But now, since by assumption $g^\Gamma$ is non-degenerate, then $e^2\epsilon_m\epsilon_n=f\mathrm{Vol}_g$, for some non-vanishing $f\in\mathcal{C^1}(\Gamma)$. From $e\beta=0$, one can infer $\epsilon_m \epsilon_n e^2 \chi=0$, which is uniquely solved by $\chi=0$, proving $\alpha=0$. Additionally, the following equations are also satisfied
        \[
        \begin{cases}
        \epsilon_m e \beta=0\\
        \epsilon_n e \beta=0\\
        \epsilon_m \epsilon_n \beta=0
        \end{cases}
        \]
    implying, by Lemma \ref{lem: (1,1) vanish}, that $\beta=0$. 

    Lastly, we are only left with checking that $W_1^{\Gamma,(0,2)}|_{\ker W_{mn}^{0,2}} $ and $W_{mn}^{1,1}|_{\ker(W_1^{\Gamma,(1,1)})} $ are injective. This is ensured by noticing that, repeating the above proof, one has
        \[
        \begin{cases}
            e\alpha=0\\
            \epsilon_m\epsilon_n \alpha=0
        \end{cases}\quad \Rightarrow \ \alpha=0, \quad \mathrm{and}\quad \begin{cases}
            e\beta=0\\
            \epsilon_m\epsilon_n \beta=0
        \end{cases}
        \quad \Rightarrow \ \beta=0.
        \]
\end{proof}

\section{Long computations}

\subsection{Proof of Proposition \ref{prop: involutivity and isotropicity}}\label{app:proof invol}
\begin{proof}

First of all notice that the isotropy of $ {D}$ guarantees that $[-,-]$ is (graded)skew-symmetric. In particular, assuming $\mathbb X + \mathcal{X}, \mathbb Y+ \mathcal{Y}\in {D} $, one has\footnote{In the following, choosing whether to compute $[\mathbb{X}+\mathcal{X},\mathbb Y+\mathcal{Y}]$ or $[\mathbb Y+\mathcal{Y},\mathbb{X}+\mathcal{X}]$ is based on computational convenience. }
    \[
    [\mathbb X + \mathcal{X}, \mathbb Y + \mathcal{Y}] + [\mathbb Y + \mathcal{Y}, \mathbb X + \mathcal{Y}] = \delta\langle \mathbb X + \mathcal{X}, \mathbb Y + \mathcal{Y} \rangle = 0
    \]
Secondly,  because of the odd parity of $\mathbb X_\alpha + \mathcal{X}_\alpha$, unlike in the usual case one has
    \[
    [\mathbb X_\alpha + \mathcal{X}_\alpha,\mathbb X_\alpha + \mathcal{X}_\alpha]\neq 0.
    \]
We have
    \begin{align*}
        [\mathbb J_c+\mathcal{J}_c,\mathbb J_c+\mathcal{J}_c]&=-\frac12\mathbb J_{[c,c]}+\cancel{\frac12\iota_{\mathbb J_c}\delta\mathcal{J}_c}-\frac12\delta(\iota_{\mathbb J_c}\mathcal{J}_c)+\cancel{\frac12\iota_{\mathbb J_c}\delta\mathcal{J}_c}\\[3pt]
        &=-\frac12\mathbb J_{[c,c]}-\frac12\intc c[c,e\delta e]=-\frac12\mathbb J_{[c,c]}-\frac12\intc [c,c]e\delta e\\[3pt]
        &=-\frac12\mathbb J_{[c,c]}-\frac12\mathcal J_{[c,c]},
    \end{align*}
        where we used the exactness property of $\mathcal{J}_c$,
    \begin{align*}
        [\mathbb E_\zeta+\mathcal{E}_\zeta,\mathbb J_c+\mathcal{J}_c]&=\mathbb J_{L_\zeta^{\omega_0}c}-\delta(\iota_{\mathbb E_\zeta}\mathcal{J}_c)\\[3pt]
        &=\mathbb J_{L_\zeta^{\omega_0}c}+\delta\intc ceL_\zeta^{\omega_0}e\\[3pt]
        &=\mathbb J_{L_\zeta^{\omega_0}c}+\delta\intc \frac{e^2}2L_\zeta^{\omega_0}c\\[3pt]
        &=\mathbb J_{L_\zeta^{\omega_0}c}+\mathcal J_{L_\zeta^{\omega_0}c},
    \end{align*}
\begin{align*}
    [\mathbb J_c+\mathcal{J}_c,\mathbb K_\eta+\mathcal{K}_\eta]&=\mathbb{J}_{\eta\de_{\omega^0_m}c}+\iota_{\mathbb J_c}\delta\mathcal{K}_\eta-\delta(\iota_{\mathbb J_c}\mathcal{K}_\eta)\\[3pt]
    &=\mathbb{J}_{\eta\de_{\omega^0_m}c}+\iota_{\mathbb J_c}\intc -\eta \delta e_m e\delta\omega+\eta e_m \delta e\delta\omega-\eta e\delta\omega_m\delta e\\[2pt]
    &\phantom{=}-\delta(\intc \eta e_m e \de_\omega c-\eta e(\omega-\omega_0)_m[c,e])\\[3pt]
    &=\mathbb{J}_{\eta\de_{\omega^0_m}c}+\intc \eta e[c,e_m]\delta \omega -\eta \delta e_m e \de_\omega c+\eta e_m[c,e]\delta\omega\\[2pt]
    &\phantom{=}+\eta e_m\delta e\de_\omega c+\eta e\de_{\omega_m}c\delta e+\eta e \delta\omega_m[c,e]\\[2pt]
    &\phantom{=}+\eta\delta e_m e \de_\omega c-\eta e_m\delta e \de_\omega c-\eta e_m[\delta \omega,c]\\[2pt]
    &\phantom{=}+\eta(\omega-\omega_0)_m[c,e\delta e]-\eta\delta\omega_m[c,\frac{e^2}2]\\[3pt]
    &=\mathbb{J}_{\eta\de_{\omega^0_m}c}+\intc\eta\de_{\omega_m^0}c e\delta e\\[3pt]
    &=\mathbb{J}_{\eta\de_{\omega^0_m}c}+\mathcal{J}_{\eta\de_{\omega^0_m}c},
\end{align*}
\begin{align*}
    [\mathbb F_\lambda+\mathcal{F}_\lambda,\mathbb J_c+\mathcal{J}_c]&=-\mathbb E_{[c,\lambda \epsilon_n]^i}-\mathbb R_{[c,\lambda \epsilon_n]^m}-\mathbb F_{[c,\lambda \epsilon_n]^n}
    +\mathbb J_{[c,\lambda \epsilon_n]^i(\omega-\omega_0)_i}\\[2pt]
    &\phantom{=}+\mathbb J_{[c,\lambda \epsilon_n]^m(\omega-\omega_0)_m} -\delta(\iota_{\mathbb F_\lambda}\mathcal{J}_c)+\iota_{\mathbb J_c}\delta\mathcal{F}_\lambda\\[3pt]
    &=-\mathbb E_{[c,\lambda \epsilon_n]^i}-\mathbb R_{[c,\lambda \epsilon_n]^m}-\mathbb F_{[c,\lambda \epsilon_n]^n}+\mathbb J_{[c,\lambda \epsilon_n]^i(\omega-\omega_0)_i}\\[2pt]
    &\phantom{=}+\mathbb J_{[c,\lambda \epsilon_n]^m(\omega-\omega_0)_m}
    -\delta\intc ce(\de_\omega(\lambda \epsilon_n)+\lambda\sigma)+\iota_{\mathbb J_c}\intc\lambda \epsilon_n\delta e\delta\omega\\[3pt]
    &=-\mathbb E_{[c,\lambda \epsilon_n]^i}-\mathbb R_{[c,\lambda \epsilon_n]^m}-\mathbb F_{[c,\lambda \epsilon_n]^n}+\mathbb J_{[c,\lambda \epsilon_n]^i(\omega-\omega_0)_i}+\mathbb J_{[c,\lambda \epsilon_n]^m(\omega-\omega_0)_m}\\[2pt]
    &\phantom{=}-\delta\intc \de_\omega c(e\lambda \epsilon_n)-\intc\lambda \epsilon_n[c,e]\delta\omega-\lambda \epsilon_n \de_\omega c\delta e\\[3pt]
    &=-\mathbb E_{[c,\lambda \epsilon_n]^i}-\mathbb R_{[c,\lambda \epsilon_n]^m}-\mathbb F_{[c,\lambda \epsilon_n]^n} +\mathbb J_{[c,\lambda \epsilon_n]^i(\omega-\omega_0)_i}+\mathbb J_{[c,\lambda \epsilon_n]^m(\omega-\omega_0)_m}\\[2pt]
    &\phantom{=}-\intc [\delta\omega,c]\lambda \epsilon_n e+\lambda \epsilon_n \de_\omega c \delta e +\lambda \epsilon_n[c,e]\delta\omega+\lambda \epsilon_n \de_\omega c\delta e\\[3pt]
    &=-\mathbb E_{[c,\lambda \epsilon_n]^i}-\mathbb R_{[c,\lambda \epsilon_n]^m}-\mathbb F_{[c,\lambda \epsilon_n]^n} +\mathbb J_{[c,\lambda \epsilon_n]^i(\omega-\omega_0)_i}+\mathbb J_{[c,\lambda \epsilon_n]^m(\omega-\omega_0)_m}\\[2pt]
    &\phantom{=}-\mathcal E_{[c,\lambda \epsilon_n]^i}-\mathcal R_{[c,\lambda \epsilon_n]^m}-\mathcal F_{[c,\lambda \epsilon_n]^n}+\mathcal J_{[c,\lambda \epsilon_n]^i(\omega-\omega_0)_i}+\mathcal J_{[c,\lambda \epsilon_n]^m(\omega-\omega_0)_m},
\end{align*}

\begin{align*}
        [\mathbb E_\zeta+\mathcal{E}_\zeta,\mathbb E_\zeta+\mathcal{E}_\zeta]&=\frac12\mathbb E_{[\zeta,\zeta]}-\frac12\mathbb J_{\iota_\zeta\iota_\zeta F_{\omega_0}}-\frac12\delta(\iota_{\mathbb E_\zeta}\mathcal{E}_\zeta)\\[3pt]
        &=\frac12\mathbb E_{[\zeta,\zeta]}-\frac12\mathbb J_{\iota_\zeta\iota_\zeta F_{\omega_0}}+\frac12\delta\intc \iota_\zeta\frac{e^2}2 (\iota_\zeta F_{\omega_0}+L_\zeta^{\omega_0}(\omega-\omega_0))+\iota_\zeta(\omega-\omega_0)L_\zeta^{\omega_0}\frac{e^2}2\\[3pt]
        &=-\frac12\delta\intc\frac{e^2}2\iota_\zeta\iota_\zeta F_{\omega_0}+\frac{e^2}2\iota_\zeta L_\zeta^{\omega_0}(\omega-\omega_0)-\frac{e^2}2(\iota_{[\zeta,\zeta]}(\omega-\omega_0)+\iota_\zeta L_\zeta^{\omega_0}(\omega-\omega_0))\\[3pt]
        &=\frac12\mathbb E_{[\zeta,\zeta]}+\frac12\mathcal E_{[\zeta,\zeta]}-\frac12\mathbb J_{\iota_\zeta\iota_\zeta F_{\omega_0}}-\frac12\mathcal J_{\iota_\zeta\iota_\zeta F_{\omega_0}},
    \end{align*}
\begin{align*}
    [\mathbb F_\lambda+\mathcal{F}_\lambda,\mathbb F_\lambda+\mathcal{F}_\lambda]&=\iota_{\mathbb F_\lambda}\delta\mathcal F_\lambda-\frac12\delta(\iota_{\mathbb F_\lambda}\mathcal F_\lambda)\\[3pt]
    &=\iota_{\mathbb F_\lambda}\intc\lambda \epsilon_n\delta e\delta \omega-\frac12\delta\intc\lambda \epsilon_n\lambda \epsilon_n F_\omega\\[3pt]
    &=\intc\lambda \epsilon_n(\de_\omega(\lambda \epsilon_n)+\lambda\sigma)\delta\omega+\lambda \epsilon_n\mathbb F_{\lambda,\omega}\delta e\\[3pt]
    &=\intc\lambda \epsilon_n(\de_\omega\lambda \epsilon_n-\lambda \de_\omega \epsilon_n)\delta\omega\\[3pt]
    &=0,
\end{align*}      
where we used $\lambda^2=\epsilon_n^2=0$ and the fact that $\mathbb F_{\lambda,\omega}$ is proportional to $\lambda$,

\begin{align*}
    [\mathbb E_\zeta+\mathcal{E}_\zeta,\mathbb K_\eta+\mathcal{K}_\eta]&=\mathbb K_{\iota_\zeta \de\eta}+\mathbb J_{\eta\iota_\zeta F_{\omega^0_m}}+\mathbb E_{\eta\partial_m\zeta}+\iota_{\mathbb E_\eta}\delta\mathcal K_\eta-\delta(\iota_{\mathbb E_\zeta}\mathcal K_\eta)\\[3pt]
    &=\mathbb K_{\iota_\zeta \de\eta}+\mathbb J_{\eta\iota_\zeta F_{\omega^0_m}}+\mathbb E_{\eta\partial_m\zeta}\iota_{\mathbb E_\zeta+}\intc-\eta\delta e_m e\delta\omega+\eta e_m\delta e\delta\omega -\eta e\delta \omega_m\delta e\\[2pt]
    &\phantom{=}-\delta\intc\eta e_m e (-\iota_\zeta F_{\omega_0}-L^{\omega_0}_\zeta(\omega-\omega_0))+\eta e(\omega-\omega_0)_mL^{\omega_0}_\zeta e\\[3pt]
    &=\mathbb K_{\iota_\zeta \de\eta}+\mathbb J_{\eta\iota_\zeta F_{\omega^0_m}}+\mathbb E_{\eta\partial_m\zeta}+\intc\eta L^{\omega_0}_\zeta e_m e \delta\omega-\eta\delta e_m e (-\iota_\zeta F_{\omega_0}-L^{\omega_0}_\zeta(\omega-\omega_0))\\[2pt]
    &\phantom{=}-\eta e_mL^{\omega_0}_\zeta e\delta\omega+\eta e_m(-\iota_\zeta F_{\omega_0}-L^{\omega_0}_\zeta (\omega-\omega_0))\delta\omega-\eta e\delta\omega_mL^{\omega_0}_\zeta e\\[2pt]
    & \phantom{=}-\eta e(-L^{\omega_0}_\zeta \omega_m-\iota_{\partial_m\zeta}(\omega-\omega_0)+\iota_\zeta\partial_m\omega_0)\delta e\\[2pt]
    &\phantom{=}+\eta\delta e_m e(-\iota_\zeta F_{\omega_0}-L^{\omega_0}_\zeta(\omega-\omega_0))-\eta e_m\delta e(-\iota_\zeta F_{\omega_0}-L^{\omega_0}_\zeta(\omega-\omega_0))\\[2pt]
    &\phantom{=}+\eta(\omega-\omega_0)_mL^{\omega_0}_\zeta(e\delta e)-\eta\delta\omega_mL^{\omega_0}_\zeta(\frac{e^2}2)-\eta e_m eL^{\omega_0}_\zeta\delta\omega+\eta\iota_{\partial_m\zeta}\frac{e^2}2\delta\omega\\[3pt]
    &=\mathbb K_{\iota_\zeta \de\eta}+\mathbb J_{\eta\iota_\zeta F_{\omega^0_m}}+\mathbb E_{\eta\partial_m\zeta}+\intc \iota_\zeta \de\eta e_m e\delta\omega+\iota_\zeta d \eta e(\omega-\omega_0)_m\delta e\\[2pt]
    &\phantom{=}-\eta L^{\omega_0}_\zeta(\omega-\omega_0)_m e \delta e+\eta e (L^{\omega_0}_\zeta\omega_m+\iota_{\partial_m\zeta}(\omega-\omega_0)+\iota_\zeta\partial_m(\omega-\omega_0)\delta e\\[3pt]
    &=\mathbb K_{\iota_\zeta \de\eta}+\mathbb J_{\eta\iota_\zeta F_{\omega^0_m}}+\mathbb E_{\eta\partial_m\zeta}+\mathcal K_{\iota_\zeta \de\eta}+\eta\iota_{\partial_m\zeta}\frac{e^2}2\delta\omega\\[2pt]
    &\phantom{=}+\intc\eta e (\iota_\zeta \de\omega_m^0+\iota_{\partial_m\zeta}(\omega-\omega_0)+\iota_\zeta\partial_m\omega_0)\delta e\\[3pt]
    &=\mathbb K_{\iota_\zeta \de\eta}+\mathbb J_{\eta\iota_\zeta F_{\omega^0_m}}+\mathbb E_{\eta\partial_m\zeta}+\mathcal K_{\iota_\zeta \de\eta}+\mathcal J_{\eta\iota_\zeta F_{\omega^0_m}}+\intc\eta(\iota_{\partial_m\zeta}\frac{e^2}2\delta\omega+\iota_{\partial_m\zeta}(\omega-\omega_0)\frac{\delta e^2}2)\\[3pt]
    &=\mathbb K_{\iota_\zeta\de\eta}+\mathbb J_{\eta\iota_\zeta F_{\omega^0_m}}+\mathbb E_{\eta\partial_m\zeta}+\mathcal K_{\iota_\zeta\de\eta}+\mathcal J_{\eta\iota_\zeta F_{\omega^0_m}}+\mathcal E_{\eta\partial_m\zeta}
\end{align*}
\begin{align*}
    [\mathbb K_\eta+\mathcal{K}_\eta,\mathbb K_\eta+\mathcal{K}_\eta]&=\mathbb K_{\eta\partial_m\eta}+\iota_{\mathbb K_\eta}\delta\mathcal K_\eta-\frac12\delta(\iota_{\mathbb K_\eta}\mathcal K_\eta)\\[3pt]
        &=\mathbb K_{\eta\partial_m\eta}+\iota_{\mathbb K_\eta}\intc-\eta\delta e_m e\delta\omega+\eta e_m\delta e\delta \omega-\eta e\delta \omega_m\delta e\\[2pt]
        &\phantom{=}+\frac12\delta\intc\eta(ee_m\delta\omega+(\omega-\omega_0)_me\delta e\\[3pt]
        &=\mathbb K_{\eta\partial_m\eta}+\intc\eta\partial_m\eta e_m e \delta\omega+\eta\delta e_m e\de\eta(\omega-\omega_0)_m+\eta e_m\delta e\de\eta(\omega-\omega_0)_m\\[2pt]
        &\phantom{=}+\eta e\partial_m\eta(\omega-\omega_0)_m\delta e-\eta ed\eta e_m\delta\omega_m+\frac12\delta\intc\eta ee_md\eta(\omega-\omega_0)_m+\eta(\omega-\omega_0)_med\eta e_m\\[3pt]
        &=\mathbb K_{\eta\partial_m\eta}+\intc\eta\partial_m\eta(e_me\delta\omega+e(\omega-\omega_0)_m\delta e)+\eta\de\eta\delta(e_me(\omega-\omega_0)_m)\\[2pt]
        &\phantom{=}
        +\delta\intc\eta\de\eta e_me(\omega-\omega_0)_m\\[3pt]
        &=\mathbb K_{\eta\partial_m\eta}+\intc\eta\partial_m\eta(e_me\delta\omega+e(\omega-\omega_0)_m\delta e)\\[3pt]
        &=\mathbb K_{\eta\partial_m\eta}+\mathcal K_{\eta\partial_m\eta},        
\end{align*}

\begin{align*}
    [\mathbb E_\zeta+\mathcal{E}_\zeta,\mathbb F_\lambda+\mathcal{F}_\lambda]&=\mathbb{E}_{L_\zeta^{\omega_0}(\lambda \epsilon_n)^i}+\mathbb{F}_{L_\zeta^{\omega_0}(\lambda\epsilon_n)^n}+\mathbb{R}_{L_\zeta^{\omega_0}(\lambda \epsilon_n)^m}-\mathbb{J}_{L_\zeta^{\omega_0}(\lambda \epsilon_n)^i(\omega-\omega_0)_i}\\[2pt]
    &\phantom{=}-\mathbb{J}_{L_\zeta^{\omega_0}(\lambda \epsilon_n)^m(\omega-\omega_0)_m}+\iota_{\mathbb E_\zeta}\delta\mathcal F_\lambda-\delta(\iota_{\mathbb E_\zeta}\mathcal F_\lambda)\\
    &=\mathbb{E}_{L_\zeta^{\omega_0}(\lambda \epsilon_n)^i}+\mathbb{F}_{L_\zeta^{\omega_0}(\lambda \epsilon_n)^n}+\mathbb{R}_{L_\zeta^{\omega_0}(\lambda \epsilon_n)^m}-\mathbb{J}_{L_\zeta^{\omega_0}(\lambda \epsilon_n)^i(\omega-\omega_0)_i}-\mathbb{J}_{L_\zeta^{\omega_0}(\lambda \epsilon_n)^m(\omega-\omega_0)_m}\\
    &\phantom{=}+\iota_{\mathbb E_\zeta}\intc \lambda \epsilon_n\delta e\delta\omega+\delta\intc\lambda \epsilon_n e(L_\zeta^{\omega_0}(\omega-\omega_0)+\iota_\zeta F_{\omega_0})\\
    &=\mathbb{E}_{L_\zeta^{\omega_0}(\lambda \epsilon_n)^i}+\mathbb{F}_{L_\zeta^{\omega_0}(\lambda \epsilon_n)^n}+\mathbb{R}_{L_\zeta^{\omega_0}(\lambda \epsilon_n)^m}-\mathbb{J}_{L_\zeta^{\omega_0}(\lambda \epsilon_n)^i(\omega-\omega_0)_i}-\mathbb{J}_{L_\zeta^{\omega_0}(\lambda \epsilon_n)^m(\omega-\omega_0)_m}\\
    &\phantom{=}-\intc\lambda \epsilon_nL_\zeta^{\omega_0}e\delta\omega+\lambda \epsilon_n(\iota_\zeta F_{\omega_0}+L_\zeta^{\omega_0}(\omega-\omega_0))\delta e\\
    &\phantom{=}+\lambda \epsilon_n(-\iota_\zeta F_{\omega_0}-L_\zeta^{\omega_0}(\omega-\omega_0))\delta e+\lambda \epsilon_n eL_\zeta^{\omega_0}\delta\omega\\
    &=\mathbb{E}_{L_\zeta^{\omega_0}(\lambda \epsilon_n)^i}+\mathbb{F}_{L_\zeta^{\omega_0}(\lambda \epsilon_n)^n}+\mathbb{R}_{L_\zeta^{\omega_0}(\lambda \epsilon_n)^m}Z-\mathbb{J}_{L_\zeta^{\omega_0}(\lambda \epsilon_n)^i(\omega-\omega_0)_i}-\mathbb{J}_{L_\zeta^{\omega_0}(\lambda \epsilon_n)^m(\omega-\omega_0)_m}\\
    &\phantom{=}+\intc (L_\zeta^{\omega_0}(\lambda \epsilon_n)^ie_i+L_\zeta^{\omega_0}(\lambda \epsilon_n)^n\epsilon_n+L_\zeta^{\omega_0}(\lambda \epsilon_n)^me_m)e\delta\omega\\
    &=\mathbb{E}_{L_\zeta^{\omega_0}(\lambda \epsilon_n)^i}+\mathbb{F}_{L_\zeta^{\omega_0}(\lambda \epsilon_n)^n}+\mathbb{R}_{L_\zeta^{\omega_0}(\lambda \epsilon_n)^m}-\mathbb{J}_{L_\zeta^{\omega_0}(\lambda \epsilon_n)^i(\omega-\omega_0)_i}-\mathbb{J}_{L_\zeta^{\omega_0}(\lambda \epsilon_n)^m(\omega-\omega_0)_m}\\
    &\phantom{=}+\mathcal{E}_{L_\zeta^{\omega_0}(\lambda \epsilon_n)^i}+\mathcal{F}_{L_\zeta^{\omega_0}(\lambda \epsilon_n)^n}+\mathcal{R}_{L_\zeta^{\omega_0}(\lambda \epsilon_n)^m}-\mathcal{J}_{L_\zeta^{\omega_0}(\lambda \epsilon_n)^i(\omega-\omega_0)_i}-\mathcal{J}_{L_\zeta^{\omega_0}(\lambda \epsilon_n)^m(\omega-\omega_0)_m},
\end{align*}
and
\begin{align*}
    [\mathbb F_\lambda+\mathcal{F}_\lambda,\mathbb K_\eta+\mathcal{K}_\eta]&=\mathbb E_{\eta\partial_m(\lambda \epsilon_n)^i}+\mathbb F_{\eta\partial_m(\lambda \epsilon_n)^n}+\mathbb K_{\eta\partial_m(\lambda \epsilon_n)^m}-\mathbb J_{\eta\partial_m(\lambda \epsilon_n)^i(\omega-\omega_0)_i}\\[2pt]
    &\phantom{=}-\mathbb J_{\eta\partial_m(\lambda \epsilon_n)^m(\omega-\omega_0)_m}+\iota_{\mathbb F_\lambda}\delta\mathcal K_\eta-\delta(\iota_{\mathbb F_\lambda}\mathcal K_\eta)+\iota_{\mathbb K_\eta}\delta\mathcal F_\lambda\\[3pt]
    &=\mathbb E_{\eta\partial_m(\lambda \epsilon_n)^i}+\mathbb F_{\eta\partial_m(\lambda \epsilon_n)^n}+\mathbb K_{\eta\partial_m(\lambda \epsilon_n)^m}-\mathbb J_{\eta\partial_m(\lambda \epsilon_n)^i(\omega-\omega_0)_i}-\mathbb J_{\eta\partial_m(\lambda \epsilon_n)^m(\omega-\omega_0)_m}\\[2pt]
    &\phantom{=}+\iota_{\mathbb F_\lambda}\intc-\eta\delta ee_m\delta\omega-\eta\delta e_me\delta\omega-\eta e\delta\omega_m\delta e-\delta\intc\eta e_m\lambda \epsilon_n F_\omega\\[2pt]
    &\phantom{=}-\delta\intc-\eta(\omega-\omega_0)_m\de_\omega(\lambda \epsilon_ne)+\iota_{\mathbb K_\eta}\intc\lambda \epsilon_n\delta e\delta\omega\\[3pt]
    &=\mathbb E_{\eta\partial_m(\lambda \epsilon_n)^i}+\mathbb F_{\eta\partial_m(\lambda \epsilon_n)^n}+\mathbb K_{\eta\partial_m(\lambda \epsilon_n)^m}-\mathbb J_{\eta\partial_m(\lambda \epsilon_n)^i(\omega-\omega_0)_i}-\mathbb J_{\eta\partial_m(\lambda \epsilon_n)^m(\omega-\omega_0)_m}\\[2pt]
    &\phantom{=}+\intc\eta(\de_\omega(\lambda \epsilon_n)+\lambda\sigma)e_m\delta\omega+\eta e_m \mathbb F_\omega\delta +\eta(\de_{\omega_m}(\lambda \epsilon_n)+\lambda\sigma)e\delta\omega-\eta\delta e_m\lambda \epsilon_n F_\omega\\[2pt]
    &\phantom{=}+\eta e \mathbb F_\omega\delta e-\eta e(\de_\omega(\lambda \epsilon_n)+\lambda\sigma)\delta\omega_m+\intc\eta\delta e_m\lambda \epsilon_n F_\omega+\eta e_m\lambda \epsilon_n \de_\omega\delta\omega\\[2pt]
    &\phantom{=}-\eta\delta\omega_m \de_\omega(\lambda \epsilon_n e)-\eta(\omega-\omega_0)_m[\delta\omega,\lambda \epsilon_n e]+\eta(\omega-\omega_0)_m\de_\omega(\lambda \epsilon_n\delta e)\\[2pt]
    &\phantom{=}+\intc \lambda \epsilon_n(-\eta\de_{\omega_m^0}e-d\eta e_m)\delta\omega +\lambda \epsilon_n(-\eta\de_{\omega_m^0}\omega+d\eta(\omega-\omega_0)_m)\delta e\\[3pt]
    &=\mathbb E_{\eta\partial_m(\lambda \epsilon_n)^i}+\mathbb F_{\eta\partial_m(\lambda \epsilon_n)^n}+\mathbb K_{\eta\partial_m(\lambda \epsilon_n)^m} -\mathbb J_{\eta\partial_m(\lambda \epsilon_n)^i(\omega-\omega_0)_i}-\mathbb J_{\eta\partial_m(\lambda \epsilon_n)^m(\omega-\omega_0)_m}\\[2pt]
    &\phantom{=}+\intc\eta(\de_\omega(\lambda \epsilon_n)+\lambda\sigma)e_m\delta\omega+\eta\lambda \epsilon_n(\de_\omega\omega_m-\partial_m\omega)\delta e\\[2pt]
    &\phantom{=}+\eta(\de_{\omega_m}(\lambda \epsilon_n)+\lambda\sigma_m)e\delta\omega-\eta\delta e_m\lambda \epsilon_n F_\omega\\[2pt]
    &\phantom{=}-\eta \de_\omega(\lambda \epsilon_n e)\delta\omega_m+\intc\eta\lambda \epsilon_n F_\omega\delta e_m-\de_\omega(\eta e_m\lambda \epsilon_n)\delta\omega\\[2pt]
    &\phantom{=}+\eta\delta\omega[(\omega-\omega_0)_m\lambda \epsilon_ne]+\de_\omega(\eta(\omega-\omega_0)_m)\lambda \epsilon_n\delta e\\[2pt]
    &\phantom{=}+\intc \lambda \epsilon_n(-\eta\de_{\omega_m^0}e-d\eta e_m)\delta\omega+\lambda \epsilon_n(-\eta\de_{\omega_m^0}\omega+d\eta(\omega-\omega_0)_m)\delta e\\[2pt]
    \end{align*}
    \begin{align*}
    &=\mathbb E_{\eta\partial_m(\lambda \epsilon_n)^i}+\mathbb F_{\eta\partial_m(\lambda \epsilon_n)^n}+\mathbb K_{\eta\partial_m(\lambda \epsilon_n)^m}-\mathbb J_{\eta\partial_m(\lambda \epsilon_n)^i(\omega-\omega_0)_i}-\mathbb J_{\eta\partial_m(\lambda \epsilon_n)^m(\omega-\omega_0)_m}\\[2pt]
    &\phantom{=}+\intc\eta(\de_\omega(\lambda \epsilon_n)+\lambda\sigma)e_m\delta\omega+\eta\lambda \epsilon_n(\de_\omega\omega_m-\partial_m\omega)\delta e\\[2pt]
    &\phantom{=}+\eta(\de_{\omega_m}(\lambda \epsilon_n)+\lambda\sigma_m)e\delta\omega-\eta\delta e_m\lambda \epsilon_n F_\omega-\eta \de_\omega(\lambda \epsilon_n e)\delta\omega_m\\[2pt]
    &\phantom{=}+\intc\de\eta e_m\lambda \epsilon_n\delta\omega+\eta\lambda(\epsilon_n\de_{\omega_m}ee\sigma_m)\delta\omega-\eta e_m \de_\omega(\lambda \epsilon_n)\delta\omega\\[2pt]
    &\phantom{=}+\eta\delta\omega[(\omega-\omega_0)_m,\lambda \epsilon_ne]+d\eta(\omega-\omega_0)_m\lambda \epsilon_n\delta e\\[2pt]
    &\phantom{=}-\eta \de_\omega(\omega-\omega_0)_m\lambda \epsilon_n\delta e+\intc \lambda \epsilon_n(-\eta\de_{\omega_m^0}e-d\eta e_m)\delta\omega\\[2pt]
    &\phantom{=}+\lambda \epsilon_n(-\eta\de_{\omega_m^0}\omega+d\eta(\omega-\omega_0)_m)\delta e\\[3pt]
    &=\mathbb E_{\eta\partial_m(\lambda \epsilon_n)^i}+\mathbb F_{\eta\partial_m(\lambda \epsilon_n)^n}+\mathbb K_{\eta\partial_m(\lambda \epsilon_n)^m}-\mathbb J_{\eta\partial_m(\lambda \epsilon_n)^i(\omega-\omega_0)_i}-\mathbb J_{\eta\partial_m(\lambda \epsilon_n)^m(\omega-\omega_0)_m}\\[2pt]
    &\phantom{=}+\intc \eta\lambda \epsilon_n(\de_\omega\omega_m-\partial_m\omega)\delta e+\eta\partial_m(\lambda \epsilon_n)e\delta\omega-\eta\lambda \epsilon_n\de_{\omega_m^0}e \delta\omega\\[2pt]
    &\phantom{=}-\eta\lambda \epsilon_n\de_{\omega_m^0}\omega\delta e+\eta\lambda \epsilon_n\de_{\omega_m}e\delta\omega+\eta\delta\omega[(\omega-\omega_0)_m,\lambda \epsilon_ne]-\eta \de_\omega\omega_m\lambda \epsilon_n\delta e\\[2pt]
    &\phantom{=}+\eta \de_\omega\omega_m^0\lambda \epsilon_n\delta e\\[3pt]
    &=\mathbb E_{\eta\partial_m(\lambda \epsilon_n)^i}+\mathbb F_{\eta\partial_m(\lambda \epsilon_n)^n}+\mathbb K_{\eta\partial_m(\lambda \epsilon_n)^m}-\mathbb J_{\eta\partial_m(\lambda \epsilon_n)^i(\omega-\omega_0)_i}-\mathbb J_{\eta\partial_m(\lambda \epsilon_n)^m(\omega-\omega_0)_m}\\[2pt]
    &\phantom{=}+\mathcal E_{\eta\partial_m(\lambda \epsilon_n)^i}+\mathcal F_{\eta\partial_m(\lambda \epsilon_n)^n}+\mathcal K_{\eta\partial_m(\lambda \epsilon_n)^m}-\mathcal J_{\eta\partial_m(\lambda \epsilon_n)^i(\omega-\omega_0)_i}-\mathcal J_{\eta\partial_m(\lambda \epsilon_n)^m(\omega-\omega_0)_m}.
\end{align*}

\end{proof}

\subsection{Proof of Theorem \ref{thm: Poisson struct}}\label{app: proof poisson}
\begin{proof}

We have
\begin{align*}
    \pi(\mathcal J_c)=\intc [c,E]\pard{}{E}+(e\de_{\omega_0}c+[c,\Omega])\pard{}{\Omega}+[c,\epsilon_m]\pard{}{\epsilon_m}+(e\de_{\omega_0}c+\epsilon_m\de_{\omega_0} c+[c,\Omega_m])\pard{}{\Omega_m}=\mathbb J_c
\end{align*}

\begin{align*}
    \pi(\mathcal F_\lambda)&=\intc[\lambda \epsilon_nr,E]\pard{}{E}+\pard{}{\Omega}(e\de_{\omega_0}(\lambda \epsilon_nr)+[e(\omega-\omega_0),\lambda \epsilon_nr])\\
    &\phantom{=\intc}-\lambda \epsilon_n(e\de_{\omega_0}\pard{}{E}-[e(\omega-\omega_0),\pard{}{E}])+\lambda \epsilon_n(F_\omega+\Theta)\pard{}{\Omega}\\
    &\phantom{=\intc}+(e\de_{\omega_m^0}(\lambda \epsilon_nr)+[\Omega_m,\lambda \epsilon_nr]+\epsilon_m \de_\omega(\lambda \epsilon_nr))\pard{}{\Omega_m}\\
    &\phantom{=\intc}+(\lambda \epsilon_n F_{\omega_m}+\Theta_m)\pard{}{\Omega_m}+[\lambda \epsilon_nr,\epsilon_m]\pard{}{\epsilon_m}+\de_\omega(\lambda \epsilon_n)\pard{}{\epsilon_m}\\
     &\phantom{=\intc}-[\lambda \epsilon_nr,\epsilon_m]\pard{}{\epsilon_m}.
\end{align*}
We now notice that the first and the third terms become
\begin{align*}
    [\lambda \epsilon_nr, E]+\de_{\omega_0}(\lambda \epsilon_n e)+[e(\omega-\omega_0),\lambda \epsilon_n]&=[\lambda \epsilon_n(\omega-\omega_0),e]+\de_{\omega_0}(\lambda \epsilon_ne)\\
    &\phantom{=}-(\omega-\omega_0)[e,\lambda \epsilon_n]+e[\omega-\omega_0,\lambda \epsilon_n]\\
    &=-\lambda \epsilon_n[\omega-\omega_0,e]+[\lambda \epsilon_n,e](\omega-\omega_0)+\de_{\omega_0}(\lambda \epsilon_n e)\\
    &\phantom{=}-(\omega-\omega_0)[e,\lambda \epsilon_n]+e[\omega-\omega_0,\lambda \epsilon_n]\\
    &=\de_{\omega_0}(\lambda \epsilon_n e)+[\omega-\omega_0,\lambda \epsilon_n e]=\de_\omega(\lambda \epsilon_n e),
\end{align*}
whereas the second and the fourth reduce to
\begin{align*}
    e\de_{\omega_0}&(\lambda \epsilon_nr)+[e(\omega-\omega_0),\lambda \epsilon_nr]+\lambda \epsilon_n(F_\omega+\Theta)=\\
    &=e\de_{\omega_0}(\lambda \epsilon_n)r+\lambda \epsilon_n e\de_{\omega_0}r-(\omega-\omega_0)[e,\lambda \epsilon_nr]\\
    &\phantom{=}+e[\omega-\omega_0,\lambda \epsilon_n]r+\lambda \epsilon_n e[\omega-\omega_0,r]+\lambda \epsilon_n(F_\omega+\Theta)\\
    &=\de_\omega(\lambda \epsilon_n)(\omega-\omega_0)-\de_\omega(\lambda \epsilon_n)\epsilon_m \epsilon_nb +\lambda \epsilon_n e\de_\omega r\\
    &\phantom{=}+[e,\lambda \epsilon_nr(\omega-\omega_0)]+\lambda \epsilon_n[e,\omega-\omega_0]r+\lambda \epsilon_n(F_\omega+\Theta)\\
    &=\de_\omega(\lambda \epsilon_n)(\omega-\omega_0)+\lambda \epsilon_n F_\omega+\lambda \epsilon_n \de_\omega(\omega-\omega_0)\\
    &\phantom{=}+\lambda \epsilon_n[e,\omega-\omega_0]r+\lambda \epsilon_n\Theta,
\end{align*}
which produces exactly $\mathbb F_\Omega$ by substituting 
\begin{align*}
    \Theta=-[e,\omega-\omega_0]r-\de_\omega(\omega-\omega_0).
\end{align*}
Regrouping all the terms proportional to $\frac{\delta}{\delta\Omega_m}$, we obtain the expected result after using 
\begin{equation*}
    \Theta_m=- [\emm,\omega-\omega_0] r -\de_{\omega_m}(\omega-\omega_0),
\end{equation*}
and after confronting the expression with $\mathbb F_{\Omega_m} $ in \cref{newdef ham v.f.}.

\begin{equation}\label{eq: pi K eta}
\begin{split}
    \pi(\mathcal{K}_\eta) =\intc & \left( \unl{[\eta\emm r, E]}{a1} + \unl{\eta[(\omega-\omega_0),E]}{a2} \right)\frac{\delta}{\delta E} - \eta \emm \left(e\de_{\omega_0}\frac{\delta}{\delta E} + [\Omega,\frac{\delta}{\delta E}] \right)\\
    & + \frac{\delta}{\delta \Omega} \left( e\de_{\omega_0}\eta\emm r+[\Omega,\eta\emm r ] \right) + [\unl{\eta\emm r}{a3},\emm]\frac{\delta}{\delta\emm} +\frac{\delta}{\delta\Omega_m} ( e\de_{\omega^0_m}\eta\emm r )\\
    & + \frac{\delta}{\delta\Omega_m} \left( \emm \de_\omega\eta\emm r + [\Omega_m,\eta\emm r] \right)
    +\eta \emm (F_\omega    +\Theta)\frac{\delta}{\delta\Omega} \\
    &+\eta\emm (F_{\omega_m}+\Theta_m)\frac{\delta}{\delta \Omega_m} +\left(\de_{\omega_m}(\eta\emm) -\unl{[\eta\emm r, \emm] }{a4}\right)\frac{\delta}{\delta\emm}.
\end{split}
\end{equation}
We first notice that the terms $\reft{eq: pi K eta}{a1}=\reft{eq: pi K eta}{a2}=0$, because $[e,\emm]=[e,r]=0$, and $\reft{eq: pi K eta}{a3}+\reft{eq: pi K eta}{a4}=0$. By inspection, one finds that the terms along $\frac{\delta}{\delta \emm}$ automatically yield $\mathbb K_{\eta,\emm}$. Similarly, regrouping the terms along $\frac{\delta}{\delta E}$, recalling $[e,\omega-\omega_0]=[e,er] +[e,\emm\en b]=0$, we obtain
    \begin{align*}
       \de_{\omega_0}(\eta\emm e) + [e(\omega-\omega_0),\eta \emm]=& \de_\omega(\eta\emm e) -[\omega-\omega_0,\eta\emm e] + [e(\omega-\omega_0),\eta \emm]\\
        =& \de_\omega(\eta\emm e) =\mathbb K_{\eta,E}.
    \end{align*}
The terms along $\frac{\delta}{\delta \Omega}$ are
    \begin{equation}\label{eq: pi K eta Omega}
    \begin{split}
        e\de_{\omega_0}&\eta\emm r +[\Omega,\eta\emm r]+\eta \emm (F_\omega+\Theta)=\\
        =&\de\eta \emm (\omega-\omega_0) - \eta\de_{\omega_0}\emm(\omega-\omega_0 -\emm\en b) + \eta \emm e\de_{\omega_0}r\\
        &+\eta \emm F_\omega -\eta\emm([e,\omega-\omega_0]r +\de_\omega(\omega-\omega_0))\\
        =&\mathbb K_{\eta,\Omega} - \eta \de_\omega \emm (\omega-\omega_0) - \eta\de_{\omega_0}\emm((\omega-\omega_0)-\unl{\emm\en b}{b5})\\
        &-\unl{\eta\emm[\omega-\omega_0,er]}{b1} +\eta \emm \de_\omega(\omega-\omega_0)-\unl{\eta \emm \de_\omega(\emm \en b)}{b6}\\
        & + \eta[\omega-\omega_0,\emm](\omega-\omega_0-\unl{\emm\en b}{b7}) -\unl{\eta\emm[\omega-\omega_0,r]e}{b4}\\
        =&\mathbb{K}_{\eta,\Omega} +\unl{\eta \de_\omega(\emm (\omega-\omega_0))}{b10}- \unl{\eta \de_\omega\emm\omega-\omega_0}{b11} + \unl{\eta \emm \de_\omega(\omega-\omega_0)}{b12}\\
        =&\mathbb{K}_{\eta,\Omega},
    \end{split}
    \end{equation}
    having noticed that $\reft{eq: pi K eta Omega}{b1}+\reft{eq: pi K eta Omega}{b4}=\reft{eq: pi K eta Omega}{b5}+\reft{eq: pi K eta Omega}{b6}+\reft{eq: pi K eta Omega}{b7}=\reft{eq: pi K eta Omega}{b10}+\reft{eq: pi K eta Omega}{b11}+\reft{eq: pi K eta Omega}{b12}=0$.\\
    From the terms along $\frac{\delta}{\delta\Omega_m}$, we obtain 
    \begin{equation}\label{eq: pi K eta Omega m}
    \begin{split}
        &e\de_{\omega^0_m}\eta \emm r + \emm \de_\omega\eta\emm r  +\unl{[\Omega_m,\eta\emm r]}{c1} + \eta \emm(F_{\omega_m}+\Theta_m)\\
        &=e\de_{\omega_m}\eta \emm r + \de_\omega\emm\eta\emm r +\eta \emm F_{\omega_m}\\
        &\phantom{=}-\eta\emm([\emm,\omega-\omega_0] r+\de_\omega(\omega-\omega_0)_m+\de_{\omega_m}(\omega-\omega_0))\\
        &=\mathbb{K}_{\eta,\Omega_m},
    \end{split}
    \end{equation}
    having noticed $\reft{eq: pi K eta Omega m}{c1}=0$.

    Lastly, we are left with showing that $\pi(\mathcal{E}_\zeta)=\mathbb E_\zeta$. We have
    \begin{equation}\label{eq: pi E zeta}
    \begin{split}
        \pi(\mathcal{E}_\zeta)=\intc &\left([\iz(\omega-\omega_0) + \iz e r, E]\right)\frac{\delta}{\delta E} - \iz e \left(e\de_{\omega_0}\frac{\delta}{\delta E} + [\Omega, \frac{\delta}{\delta E}] \right)\\
        &+\frac{\delta}{\delta\Omega}\left(e\de_{\omega_0}(\iz (\omega-\omega_0)+\iz e r) +[\Omega, \iz (\omega-\omega_0) + \iz e r] \right)\\
        &+\frac{\delta}{\delta\Omega}\left(e \de_\omega \left(\alpha_\zeta -\emm\en \frac{\delta\beta_\zeta}{\delta e}\right) -(\omega-\omega_0)[e,\alpha_\zeta -\emm\en \frac{\delta\beta_\zeta}{\delta e}] \right)\\
        &+\left( [\unl{\iz e r}{d1} + \iz (\omega-\omega_0),\emm] +\unl{[\alpha_\zeta -\emm\en \frac{\delta\beta_\zeta}{\delta e},\emm]}{d3} \right)\frac{\delta}{\delta \emm}\\
        &+\frac{\delta}{\delta\Omega_m}\left(e\de_{\omega_m^0}(\iz(\omega-\omega_0 + \iz e r + \alpha_\zeta -\emm\en \frac{\delta\beta_\zeta}{\delta e}) + \emm\de_{\omega_0}(\iz(\omega-\omega_0 + \iz e r ) \right)\\
        &+\frac{\delta}{\delta\Omega_m}\left(\emm\de_{\omega_0}(\alpha_\zeta -\emm\en \frac{\delta\beta_\zeta}{\delta e}) + [\Omega_m, \iz(\omega-\omega_0 + \iz e r + \alpha_\zeta -\emm\en \frac{\delta\beta_\zeta}{\delta e}] \right)\\
        &+\iz e (F_\omega+\Theta)\frac{\delta}{\delta\Omega} +\iz e (F_{\omega_m}+\Theta_m)\frac{\delta}{\delta\Omega_m} +\left(\de_{\omega_m}(\iz e) -\unl{[\iz e r,\emm]}{d2} \right)\frac{\delta}{\delta\emm} \\
        &- \unl{W_e^{-1}\left( [\iz e \en \emm b,\emm] \right)}{d4}\frac{\delta}{\delta\emm}.
    \end{split}
    \end{equation}
    We immediately see that $\reft{eq: pi E zeta}{d1}+\reft{eq: pi E zeta}{d2}=0$. Similarly, we see that $\reft{eq: pi E zeta}{d3}$ is such that 
    \begin{align*}
        e[\alpha_\zeta -\emm\en \frac{\delta\beta_\zeta}{\delta e},\emm]=[e\left(\alpha_\zeta -\emm\en \frac{\delta\beta_\zeta}{\delta e}\right),\emm]=[\iz e \en \emm b,\emm],
    \end{align*}
    and therefore $\reft{eq: pi E zeta}{d3}+\reft{eq: pi E zeta}{d4}=0$. Then, along $\frac{\delta}{\delta E}$, we have
    \begin{align*}
        &\phantom{=\ }[\iota_\zeta(\omega-\omega_0)+\iota_\zeta e r, E]+[\alpha_\zeta-\epsilon_m\epsilon_n\frac{\delta \beta_\zeta}{\delta e},E]-L^{\omega_0}_\zeta E+[e(\omega-\omega_0),\iota_\zeta e]\\
        &=[\iota_\zeta(\omega-\omega_0),E]-[\iota_\zeta e,e]\epsilon_m\epsilon_n b-[e, e(\alpha_\zeta-\epsilon_m\epsilon_n\frac{\delta \beta_\zeta}{\delta e}]-L^{\omega_0}_\zeta E+e[\omega-\omega_0,\iota_\zeta e]\\
        &=-[\omega-\omega_0,e]\iota_\zeta e+[e,\iota_\zeta e]\epsilon_m\epsilon_n b-[e,\iota_\zeta e\epsilon_m\epsilon_nb]-L^{\omega_0}_\zeta E\\
        &=[e,er+\epsilon_m\epsilon_n b]\iota_\zeta e+[e,\iota_\zeta e]\epsilon_m\epsilon_n b-[e,\iota_\zeta e\epsilon_m\epsilon_n b]-L^{\omega_0}_\zeta E\\
        &=-L^{\omega_0}_\zeta E,
    \end{align*}
    where we simply developed the brackets and simplified the resulting terms.
    
    Moreover, along $\frac{\delta}{\delta\Omega}$, we have
    \begin{align*}
    &\phantom{=\ } e \de_\omega(\iota_\zeta(\omega-\omega_0))
    + e \de_\omega(\iota_\zeta e r)
    - (\omega-\omega_0)[e,\iota_\zeta(\omega-\omega_0)+\iota_\zeta e r]\\
    &\quad
    + e \de_\omega(\alpha_\zeta-\epsilon_m\epsilon_n\frac{\delta\beta_\zeta}{\delta e}) - (\omega-\omega_0)[e,\alpha_\zeta-\epsilon_m\epsilon_n\frac{\delta\beta_\zeta}{\delta e}] + \iota_\zeta e (F_\omega+\Theta) \\
    &=
    \de_\omega(e \iota_\zeta(\omega-\omega_0)) + \de_\omega(\iota_\zeta e (\omega-\omega_0)) - [\iota_\zeta(\omega-\omega_0),e](\omega-\omega_0)\\
    &\quad - [e,\iota_\zeta e](e r+\epsilon_m\epsilon_n b) r - (\omega-\omega_0)[e,\alpha_\zeta-\epsilon_m\epsilon_n\frac{\delta\beta_\zeta}{\delta e}]\\
    &\quad
    - e\iota_\zeta F_\omega - \iota_\zeta e [e,\omega-\omega_0] r - \iota_\zeta e \de_\omega(\omega-\omega_0)\\
    &=
    \de_\omega\iota_\zeta\Omega- [\iota_\zeta(\omega-\omega_0),e](\omega-\omega_0)- [e,\iota_\zeta e]\epsilon_m\epsilon_n b r\\
    &\quad - (e r+\epsilon_m\epsilon_n b)[e,\alpha_\zeta-\epsilon_m\epsilon_n\frac{\delta\beta_\zeta}{\delta e}]- e\iota_\zeta F_\omega\\
    &\quad - \iota_\zeta e [e,\epsilon_m\epsilon_n b] r - \iota_\zeta e \de_\omega(\omega-\omega_0)\\
    &= E_\Omega+ \epsilon_m\epsilon_n b[e,\alpha_\zeta-\epsilon_m\epsilon_n\frac{\delta\beta_\zeta}{\delta e}],
    \end{align*}
where
\begin{align*}
    \epsilon_m\epsilon_n b[e,\alpha_\zeta-\epsilon_m\epsilon_n\frac{\delta\beta_\zeta}{\delta e}]&=\epsilon_m\epsilon_n b W_e^{-1}([e,\iota_\zeta e \epsilon_m\epsilon_n b])\\
    &=\epsilon_m\epsilon_n b W_e^{-1}([e,\iota_\zeta e \epsilon_m\epsilon_n] b)\\
    &=b W_e^{-1}(\epsilon_m\epsilon_n[e,\iota_\zeta e\epsilon_m\epsilon_n] b)\\
    &=0.
\end{align*}
Lastly, along $\frac{\delta}{\delta \Omega_m}$, we have
\begin{align*}
    &\phantom{=\ } e\de_{\omega_m^0}(\iota_\zeta(\omega-\omega_0)+\iota_\zeta e r+\alpha_\zeta-\epsilon_m\epsilon_n\frac{\delta\beta_\zeta}{\delta e})
    + \epsilon_m \de_\omega(\iota_\zeta(\omega-\omega_0)+\iota_\zeta e r+\alpha_\zeta-\epsilon_m\epsilon_n\frac{\delta\beta_\zeta}{\delta e})\\
    &\quad
    + [\Omega_m,\iota_\zeta(\omega-\omega_0)+\iota_\zeta e r+\alpha_\zeta-\epsilon_m\epsilon_n\frac{\delta\beta_\zeta}{\delta e}]
    + \iota_\zeta e(F_{\omega_m}+\Theta_m)\\
    &=
    e\de_{\omega_m}(\iota_\zeta(\omega-\omega_0)+\iota_\zeta e r+\alpha_\zeta-\epsilon_m\epsilon_n\frac{\delta\beta_\zeta}{\delta e})
    + \epsilon_m \de_\omega(\iota_\zeta(\omega-\omega_0)+\iota_\zeta e r+\alpha_\zeta-\epsilon_m\epsilon_n\frac{\delta\beta_\zeta}{\delta e})\\
    &\quad
    + (\omega-\omega_0)[\epsilon_m,\iota_\zeta(\omega-\omega_0)+\iota_\zeta e r+\alpha_\zeta-\epsilon_m\epsilon_n\frac{\delta\beta_\zeta}{\delta e}]
    - e \iota_\zeta F_{\omega_m}\\
    &\quad
    - \iota_\zeta e\de_{\omega_m}(\omega-\omega_0)- \iota_\zeta e[\epsilon_m,\omega-\omega_0]r\\
    &=
   \de_{\omega_m}(e \iota_\zeta(\omega-\omega_0)+\iota_\zeta e(\omega-\omega_0-\epsilon_m\epsilon_n b)+\iota_\zeta e \epsilon_m\epsilon_n b)\\
    &\quad
    - \de_\omega\epsilon_m(\iota_\zeta(\omega-\omega_0)+\iota_\zeta e r)
    - \de_\omega\epsilon_m(\alpha_\zeta-\epsilon_m\epsilon_n\frac{\delta\beta_\zeta}{\delta e})\\
    &\quad
    - \de_\omega(\epsilon_m\iota_\zeta(\omega-\omega_0)+\iota_\zeta e \epsilon_m r+\epsilon_m\alpha_\zeta-\epsilon_m^2\epsilon_n\frac{\delta\beta_\zeta}{\delta e})\\
    &\quad
    + \de_\omega\epsilon_m(\iota_\zeta(\omega-\omega_0)+\iota_\zeta e r+\alpha_\zeta-\epsilon_m\epsilon_n\frac{\delta\beta_\zeta}{\delta e})
    - [\iota_\zeta(\omega-\omega_0),\epsilon_m](\omega-\omega_0)\\
    &\quad
    - [\epsilon_m,(\omega-\omega_0)\iota_\zeta e r]
    + \iota_\zeta e[\epsilon_m,\omega-\omega_0]r
    + e r[\epsilon_m,\alpha_\zeta-\epsilon_m\epsilon_n\frac{\delta\beta_\zeta}{\delta e}]\\
    &\quad
    + \epsilon_m\epsilon_n b[\epsilon_m,\alpha_\zeta-\epsilon_m\epsilon_n\frac{\delta\beta_\zeta}{\delta e}]
    - e \iota_\zeta F_{\omega_m}\\
    &\quad
    + \epsilon_m r[e,\iota_\zeta(\omega-\omega_0)+\iota_\zeta e r+\alpha_\zeta-\epsilon_m\epsilon_n\frac{\delta\beta_\zeta}{\delta e}]\\
    &\quad
    - \iota_\zeta e(\de_{\omega_m}(\omega-\omega_0))
    - \iota_\zeta e[\epsilon_m,\omega-\omega_0]r,
\end{align*}
which, after simplification, reduces to
\begin{align*}
    &=\de_{\omega_m}\iota_\zeta\Omega
    + \de_\omega\iota_\zeta\Omega_m
    + \de_\omega(\epsilon_m W_e^{-1}(\iota_\zeta e \epsilon_m\epsilon_n b))
    - [\iota_\zeta(\omega-\omega_0),\epsilon_m](\omega-\omega_0)\\
    &\quad
    - [\epsilon_m,\epsilon_m\epsilon_n b \iota_\zeta e r]
    + [\epsilon_m,\iota_\zeta e \epsilon_m\epsilon_n b]r
    - e \iota_\zeta F_{\omega_m}
    - \iota_\zeta e\de_{\omega_m}(\omega-\omega_0),
\end{align*}
where we used the identities
\begin{align*}
    \epsilon_m\epsilon_n b[\epsilon_m,\alpha_\zeta-\epsilon_m\epsilon_n\frac{\delta\beta_\zeta}{\delta e}]
    &=
    \epsilon_m\epsilon_n b W_e^{-1}([\epsilon_m,\iota_\zeta e \epsilon_m\epsilon_n b])\\
    &=
    b W_e^{-1}([\epsilon_m,\iota_\zeta e \epsilon_m\epsilon_n b]\epsilon_m\epsilon_n)
    =0
\end{align*}
and
\begin{align*}
    \epsilon_m W_e^{-1}(\iota_\zeta e \epsilon_m\epsilon_n b)
    &=
    W_e^{-1}(\iota_\zeta e \epsilon_m^2\epsilon_n b)
    =0.
\end{align*}
The remaining terms in $\pi(\mathcal{E}_\zeta)$ form exactly $\mathbb E_{\zeta,\emm}$, which concludes the proof.
\end{proof}

\newpage
\nocite{*}
\emergencystretch=2em
\newrefcontext[sorting=nty]
\sloppy
\printbibliography

\end{document}